\documentclass[11 pt,onehalfspacing]{amsart}

\usepackage{amsthm,amsmath}
\usepackage{amsfonts}
\usepackage{amsmath}
\usepackage{fancyhdr}
\usepackage[T1]{fontenc}
\usepackage[latin9]{inputenc}
\usepackage{mathabx}
\usepackage{color}
\usepackage{float}
\usepackage{mathrsfs}
\usepackage{layout}
\usepackage{slashed}
\usepackage{amsthm}
\usepackage{amstext}
\usepackage{amssymb}
\usepackage{tikz}
\usetikzlibrary{positioning}
\usepackage{stmaryrd}
\usepackage{mathtools}
\usepackage{graphicx}
\usepackage{setspace}
\usepackage{esint}
\RequirePackage[colorlinks=true,citecolor=blue,urlcolor=blue]{hyperref}
\usepackage[authoryear]{natbib}
\usepackage{ulem}
\usepackage{enumitem}
\usepackage{ifpdf}
\usepackage{pdflscape}
\usepackage[figure]{hypcap}
\usepackage{fancyhdr}
\usepackage{bbm}
\usepackage{cleveref}
\usepackage{geometry}
\usepackage[parfill]{parskip}
\allowdisplaybreaks

\newtheorem{theorem}{Theorem}

\newtheorem{Corollary}{Corollary}

\newtheorem{Definition}{Definition}

\newtheorem{Lemma}{Lemma}

\newtheorem{Proposition}{Proposition}

\crefname{Lemma}{Lemma}{Lemma}
\crefname{theorem}{Theorem}{Theorem}
\crefname{Definition}{Definition}{Definition}
\crefname{Corollary}{Corollary}{Corollary}
\crefname{Proposition}{Proposition}{Proposition}

\ifpdf 
 \IfFileExists{lmodern.sty}{\usepackage{lmodern}}{}
\fi 
\let\myTOC\tableofcontents
\renewcommand\tableofcontents{  \frontmatter
  \pdfbookmark[1]{\contentsname}{}
  \myTOC
  \mainmatter }

\begin{document}

		\title{Tilting Approximate Models}
		\author{Andreas Tryphonides\\ University of Cyprus}
		
	\email{tryfonidis.antreas@ucy.ac.cy}
	\date{\today}


		\thanks{\tiny{Parts of this paper draw from chapter 3 of my PhD thesis (EUI). I thank Fabio Canova, Peter Reinhard Hansen, Giuseppe Ragusa and Frank Schorfheide for useful comments and suggestions. Earlier versions of this paper (circulated with different titles) greatly benefited from discussions with Raffaella Giacomini, George Tauchen and comments from  the participants at the 23rd MEG (Bloomington), the 1st IAAE (London), the 68th ESEM (Toulouse), the 4th International Conference in memory of Carlo Giannini (Pavia), the Econometrics Study Group (Bristol), SofiE(2022), the Economic Risk seminar (Humboldt U), the Econometrics seminar at the Tinbergen Institute, the EUI Econometrics Working Group and the University of Wisconsin Madison lunch seminar. I also acknowledge financial support from the UCY starting grant. \\\\\textit{Email}: tryfonidis.antreas@ucy.ac.cy\\\textit{Address}: Panepistimiou 1, 2109, Nicosia, Cyprus. }}

	\setcounter{page}{1}
	\pagenumbering{arabic}
	\pagestyle{plain}
	\thispagestyle{empty}

\onehalfspacing

\begin{abstract}
Model approximations are common practice when estimating structural or quasi-structural models. The paper proposes utilizing projections to re-impose information about the exact model in the form of conditional moments. The resulting estimator efficiently combines the information provided by the approximate model and the  moment conditions. The paper develops the corresponding asymptotic theory and provides simulation evidence that tilting substantially reduces the mean squared error for parameter estimates. It applies the methodology to pricing long-run risks in aggregate consumption in the US, whereas the model is solved using the \citet{doi:10.1111/j.1540-6261.1988.tb04598.x}  approximation. Tilting improves empirical fit and results suggest that approximation error is a source of upward bias in estimates of risk aversion and downward bias in the elasticity of intertemporal substitution.
\end{abstract}

\maketitle

\normalsize

\onehalfspacing
\newpage
\section{Introduction} Model approximations are quite common in structural estimation as they are regularly used to obtain the equilibrium law of motion, or more generally, the reduced form of a model. It is well documented in the literature that alternative methods deliver approximation errors that may not be uniform across the state space, while reducing the error can be quite costly in terms of running time (see e.g. \citet{FERNANDEZVILLAVERDE2016527} for structural macroeconomic models). The quality of approximation employed for estimation can induce  a form of misspecification that can be consequential for  judging the  importance of different mechanisms and structures. For example, \citet{JOFI:JOFI12615}   show that for asset pricing models with long run risks, the error in moments when using a log-linear approximation can exceed $70\%$, while the equity risk premium can be overestimated by 100 basis points. \citet{doi:10.3982/ECTA12791} find that errors from the first and second order perturbation solutions of the New Keynesian model can exceed $100\%$.\footnote{Approximations also matter for parameter identification, as they can effectively limit the model's empirical content. For example, \citet{Canova2009431} shows that in log-linearized DSGE models the mapping between the structural parameters and  the coefficients of the solution (which map to first and second moments in the data) is ill-behaved.} 

The trade-off between tractability and misspecification becomes very relevant to parameter estimation, as the solution must be repeatedly evaluated at alternative points in the parameter space.   This paper considers the use of conditional density projections (see \cite{Komunjer_Ragusa}  and \cite{Czizar} for the unconditional case) and shows that when applied to structural models, they  effectively combine the information  produced by the approximate model with the one contained in the exact equilibrium conditions. 
The approximate model is represented by a
conditional probability measure with density $f(X|Z,\varphi)$, indexed by parameter vector $\varphi$  and conditioned on $Z$. The projection tilts the approximate model by obtaining a probability distribution that is as close
as possible to the base measure and satisfies the
conditional moments of the exact model, that is $\mathbb{E}(m(X,\vartheta )|Z)=0$, where $\vartheta$ is the structural parameter of interest. For example, the approximate density can be the one generated by a complete log-linearized dynamic equilibrium model for $X$ given $Z$ (which includes lagged values of $X$), while the exact moment could be the non-linear Euler equation for consumption.\footnote{Potential applications extend beyond macroeconomics. Most structural microeconomic models that are based on optimizing agents have to satisfy analogous conditions. Other examples include models of network formation, where the conditional density of an exponential random graph model is usually approximated, and the pseudo-likelihood may not correspond to a complete set of sufficient statistics (see e.g. \citet{paula_2017,10.1214/13-AOS1155} and references therein). The information projection can potentially re-impose this information on the pseudo-likelihood.)} 
   
   More intuition can be built by previewing one of the  results of this paper. Under a sequence of approximate models local to the data generating process, maximizing the likelihood of the tilted model is equivalent to exploiting the following moment for estimation, which automatically combines the efficient GMM first order condition  with the (orthogonalized) score of the approximate model:\footnote{$(M_{z},V_{m,z})$ signify the Jacobian and Variance of the conditional moments.}
\begin{eqnarray}\underset{\text{\quad  \quad Efficient GMM}\quad  \quad\text{Orthogonalized Score of Approximate model} }{\mathbb{E}\left({{M}_{Z}'V^{-1}_{m,Z}\mathbb{E}(m|Z)}+ \frac{d\varphi}{d\vartheta}'\left(\mathfrak{s}- \mathbb{E}(\mathfrak{s}m'|Z)V^{-1}_{m,Z}\mathbb{E}(m|Z)\right)\right)}\approx0\label{eq:intro_mom}
\end{eqnarray}By construction, the implied estimate for $\vartheta$ is going to be more efficient than the one that only exploits  $\mathbb{E}(m(X,\vartheta )|Z)=0$ for estimation. Symmetrically, the bias arising from employing the approximate model will be much lower once it is combined with the exact moment condition. The result is a substantial reduction in the mean squared error for estimating $\vartheta$.\footnote{In a very different  setting, \citet{DITRAGLIA2016187} proposes using invalid instruments on purpose in an IV framework. In this paper the moment conditions that are based on "invalid instrument" are the scores of the approximate model. The difference is that the additional moment conditions are not arbitrary, but pinned down by the approximate model, which achieves the best possible efficiency in the correct specification limit. }

Information projections have been considered in different contexts, such as forecasting (see e.g. \cite{Giacomini2014145, 10.2307/3839160}) and parameter estimation using conditional or unconditional moment restrictions (see e.g. the use of exponential tilting in \cite{SusanneSchennah,KitamuraStutzer, 10.2307/2998561}, 
and Generalized Empirical Likelihood  (GEL)  i.e. \cite{ECTA:ECTA482}). Within the latter literature, a non-parametric model of the (joint) distribution of the data is forced to satisfy a set of (conditional) moment restrictions in the first step, giving rise to a non-parametric likelihood function which can be used to estimate and do inference on the parameters indexing the moment condition.
 This paper departs from this strand of the literature, and in particular \cite{SusanneSchennah}, by considering a generalized version of exponential
tilting in the first step, where the form of $f(X|Z,\varphi)$ is parametrically specified. 
This expands the range of probability measures that can be considered, as it allows us to investigate the performance of estimators based on tilting structural or quasi-structural models. Within the macroeconomic literature, information projections have been considered by \citet{10.3982/QE737} as a way to improve the numerical accuracy of discretization of Markov processes such as income, which are then used to price disaster risk. This paper considers the econometric properties of employing projections to restore information lost by approximating the law of motion in the generic case.\footnote{Similar in spirit, \citet{KRISTENSEN2017189} consider Newton-Raphson adjustments of approximate estimators to remove the first-order bias that arises from stochastic and non-stochastic approximations. In a broadly related paper, \citet{https://doi.org/10.3982/QE1413} propose composite estimators to deal with misspecification in macroeconomic models. }

The paper provides consistency and asymptotic distribution results for a generic $(\vartheta,\varphi)$, where $\varphi$ is finite dimensional. The main results 
concern a density which is induced by a structural approximation, such as in models with optimizing agents,  and hence there is a mapping from $\varphi$ to $\vartheta$. When misspecification vanishes at a fast enough rate, $\hat{\vartheta}$ attains the parametric lower bound (Cramer-Rao bound).  Analogous rate conditions have been considered by \citet{Geweke_Hahn_Ackerberg} when estimating approximated equilibrium models using likelihood methods. 
The paper also derives the asymptotic distribution along drifting parameter sequences that are local to the true model, where misspecification is in the form of improper finite dimensional restrictions to the true density and the error vanishes at a root-n rate. The pseudo maximum likelihood estimator (PMLE) of the tilted approximate model can dominate in terms of asymptotic risk the MLE of the exact model. This dominance is the result of both higher efficiency and low bias of the PMLE based on the tilted approximate model. Simulation evidence shows that tilting significantly lowers the MSE as compared to the "naive" MLE based on the approximate model, which has a higher bias.    
  For completeness and ease of exposition, the paper first provides results when the base density is a non-structural model, and hence $\varphi$ and $\vartheta$ are independent.   When misspecification vanishes at a fast enough rate, $\hat{\vartheta}$ attains the semi-parametric lower bound (SLB) (see \cite{CHAMBERLAIN1987305}).\footnote{This is the same bound achieved by other semi-parametric estimators such as GMM and GEL and its variants. 
 	In simulation evidence based on the Mean Squared Error (MSE), the estimator is compared to the Continuously Updated GMM estimator (CU-GMM) and the Empirical Likelihood (EL) estimator in the conditional moment case.  Additional evidence is provided regarding CU-GMM and the Exponentially Tilted EL (ETEL) in the unconditional moment case as well. In all cases, the estimator behaves favorably in finite samples.  }

The estimation approach proposed in the paper can be also casted as an adversarial estimator. Recent work has highlighted the connections between adversarial estimation and GEL estimation (see e.g. \citet{https://doi.org/10.48550/arxiv.2007.06169,https://doi.org/10.48550/arxiv.2204.10495,Gigliutti}) while in the machine learning literature this problem has been casted as a zero sum game between the adversary (generator) and a discriminator. In Appendix \hyperref[AppA]{A} , I discuss how the estimator can be formulated as the outcome of a sequential, non-zero sum game.

The empirical application contributes to the literature that is related to non-linearities in macro-finance models with long run risk, and the misspecification that arises from linearization techniques, as the latter ignore the level effects of consumption growth on risk premia (see i.e. \citet{JOFI:JOFI12615} and references there in.). Information projections are applied to pricing long run risks in aggregate consumption \citep{doi:10.1111/j.1540-6261.2004.00670.x}, where the base density is generated using the \citet{doi:10.1111/j.1540-6261.1988.tb04598.x} approximation. 
Using US data, the paper investigates the empirical importance of non-linearity by re-imposing the non-linear equilibrium condition.  The tilted model is strongly preferred by the data. Parameter estimates suggest that the quality of the approximation is yet another reason for the downward bias to estimates of the intertemporal elasticity of substitution and the upward bias in risk aversion, which have been puzzling in the literature, as they imply that asset prices are increasing in uncertainty.





The rest of the paper is organized as follows. Section 2 introduces information projections and provides an asset pricing example.
Section 3 presents the econometric properties, computational aspects and supportive simulation evidence. Section 4 applies the methodology to the long run risk model and provides additional simulation evidence.  Section 5 concludes. Appendix \hyperref[AppA]{A}  contains the main proofs and derivations. Appendix \hyperref[AppB]{B}  (intended for online publication) contains supplementary results. 

Finally, a word on notation.  The set of parameters
$\psi$ is decomposed in $\vartheta\in\Theta$, the set of structural
(economic) parameters, and $\varphi$ the parameters indexing the density
$f(X|Z,\varphi)$. Let $N$ denote the length of the
data and $N_{s}$ the length of simulated series. $X$ is an $n_{x}\times1$
vector of the variables of interest while $Z$ is an $n_{z}\times1$ vector
of conditioning variables. Both $X$ and $Z$ induce a probability
space $(\Omega,\mbox{\ensuremath{\mathcal{F}}},\mathbb{P})$. 
Three different probability measures are used, the true measure $\mathbb{P}$ (with $\mathbb{P}_{N}$ the corresponding empirical measure),
the base measure $F_{\varphi}$ which is indexed by parameters $\varphi$
and the ${H}_{(\varphi,\vartheta)}$ measure which is obtained after
the information projection. These measures are absolutely continuous
with respect to a dominating measure $v$, where $v$ is the Lebesgue measure, and possess the corresponding density functions $p,f$
and $h$.  Conditional measures and densities are denoted by an additional index that specifies the conditioning variable, i.e. $F_{\varphi,z}$. 
  $\mathbb{E}_{P} $ is the mathematical expectations
operator with respect to measure $P$; $\mathbb{E}_{P_{z}}$ is shorthand for conditional expectations respectively. 
$q^{l}(X,Z,\psi)$
is a general $X\otimes Z$ measurable function and ${q}(X,Z,\psi)$ is an $n_{q}\times1$
vector containing these functions. Moreover, $q_{\psi}$ abbreviates the Jacobian matrix of $q$ and $q_{\psi\psi'}$ the Hessian with respect to $\psi$. Unless otherwise stated, $||.||$ signifies the Euclidean norm.  For any
(matrix) function the subscript $i$ denotes the evaluation
at datum $(x_{i},z_{i})$. Subscript $j$ is for simulated data using the base density. The operator $\to_{p}$ signifies convergence
in probability and $\to_{d}$ convergence in distribution; $\mathcal{N}(.,.)$
is the Normal distribution. For any $q$, we may abbreviate its mathematical expectation given measure $P$ by $q_{P}$, that is, $q_{P}\equiv \mathbb{E}_{P}q$ while $V_{P,q}$ is the corresponding variance. Conditional variance is denoted by $V_{P_{z},q}$. If $P\equiv\mathbb{P}$ then $V_{\mathbb{P}}\equiv\mathbb{V}$. $V_{ll'}$ signifies the $(l,l')$ component of a matrix $V$. 

\section{Tilting the Approximate Model}
The main idea can be illustrated as follows. Economic theory implies a set of restrictions:\[\int m(X,Z,\vartheta)d\mathbb{P}(X|Z,\vartheta)=0 \] where $m(X,Z,\theta)$ is a set of conditional moment functions characterized by $\vartheta$, which includes all the economically relevant parameters. The researcher is able to obtain an approximation to $\mathbb{P}(X|Z,\vartheta)$ e.g. the density implied by the solution to the log-linearized model, $f(X|Z,\varphi)$, which by construction does not satisfy the original moment condition: $\int m(X,Z,\vartheta)f(X|Z,\varphi)dX\neq 0$.  This implies  that matching the approximated model to the data violates the main economic implications of the original specification, and this violation is state (time) and parameter dependent. 

This paper proposes to use a conditional information projection, that is, to obtain a density that is as close as possible to the approximate density $f(X|Z,\varphi)$ but satisfies the original moment conditions. This is equivalent to solving the following infinite dimensional optimization problem:\footnote{One could also consider optimizing with respect to $h(X|Z)g(Z)$ as $g(Z)$ is also unknown, yet $h^{\star}$ is the same as the moment restrictions hold for any $Z$, and therefore any $g(Z)$.} \small
\begin{eqnarray}
\min_{h(X|Z)}\int h(X|Z)log\left(\frac{h(X|Z)}{f(X|Z,\varphi)}\right)g(Z)d(X,Z) \label{principal}
\end{eqnarray}
\begin{eqnarray*}
&s.t.&\int h(X|Z)m(X,Z,\vartheta)dX=0,\int h(X|Z)dX=1,Z - a.e.
\end{eqnarray*}\normalsize where $g(Z)$ is the marginal distribution of $Z$.
In the information
projections literature  minimization  \eqref{principal} is called exponential tilting it minimizes the Kullback-Leibler ($\mathbf{KL}$) distance, whose convex conjugate has an exponential form. The solution to the above problem, if it exists\footnote{ See \cite{Komunjer_Ragusa} and section 3 where I relate my assumptions to theirs.}, is given by
\begin{equation}
h^{\star}(X|Z,\psi)=f(X|Z,\varphi)\exp\left(\lambda(Z,\psi)+\mu(Z,\psi)'m(X,Z,\vartheta)\right) \label{eq:h}
\end{equation} where $\psi=(\vartheta',\varphi')'$,  $\mu$ is the vector of the Lagrange multiplier functions enforcing
the conditional moment conditions on $f(X|Z,\varphi)$ and $\lambda$
is a scaling function. Both are pinned down by:
\begin{eqnarray}
	&\mu(Z,\psi)=\underset{\mu\in\mathbb{R}^{n_{m}}}{\arg\min} \int f(X|Z,\varphi)\exp(\mu'm(X,Z,\vartheta)dX \label{eq:mudef}\\ 
	&\lambda(Z,\psi)=-\log\left(\int f(X|Z,\varphi)\exp(\mu(Z,\psi)'m(X,Z,\vartheta)dX\right) &   \label{eq:lambdadef}
\end{eqnarray}
Given \eqref{eq:h}, $\psi$ is estimated using the (limited information) log likelihood function:
\begin{eqnarray}
\psi^{\star}&=\underset{{\psi\in \Psi}}{\arg\max} \int \log(h^{\star}(X|Z,\psi))d\mathbb{P}(X,Z)& \label{eq:agent}
\end{eqnarray}
Had we used an alternative objective function to \eqref{principal}, e.g. another particular case from the general family of divergences in \cite{cressie_read}, this would result to a different form for $h^{\star}(X|Z,\psi)$.\footnote{For brevity, in the rest of the paper I drop the $\star$ notation on $h$.} Under correct specification for $f(X|Z,\varphi)$, this choice does not matter asymptotically, while it matters in finite samples. Exponential tilting nevertheless ensures a positive density function $h$.

\subsection{Relation to other methods}

The projection in the first step (i.e. imposing the conditional moment conditions) is done using a different divergence measure than the estimation objective. A consequence is that the first order conditions of the estimator are mathematically different than a constrained estimator such as the one employed by \citet{1989}.\footnote{In other words, the primal problem of maximum tilted likelihood is not equivalent to the primal problem of constrained maximum likelihood estimation. In terms of implementation, tilting replaces the nonlinear constraint on the coefficients of the approximating density with a set of convex optimizations.} Furthermore, the projection results in exponential tilting of the base density, which  preserves its properties as a density.\footnote{This is in contrast to using other divergences which might result in the concentrated objective not having a density interpretation such as empirical likelihood (EL), i.e. the EL weights can be negative, do not satisfy the restrictions for arbitrary parameters and do not provide an estimated density \citep{2007ch}.}
The closest estimation method is  that of \cite{SusanneSchennah}. In the case of unconditional moment restrictions and when the base measure is entirely uninformative, that is $F(X;\varphi)$ is the empirical cdf, the proposed estimator  collapses to the ETEL estimator as follows:
\begin{eqnarray}
	&\max_{\vartheta\in \Theta} \sum_{i=1..N} \log(\hat{w}_{i})& 
\end{eqnarray} \normalsize
with $\hat{w}_{i}=\underset{w_{i}}{\arg\min} \sum_{i=1..N} w_{i}log\left(Nw_{i}\right)$, s.t. $\sum_{i=1..N}w_{i}m(X_{i};\vartheta)=0,\sum_{i=1..N}w_{i}=1$. \normalsize

 \citet{SusanneSchennah} argues that ETEL behaves favorably as it combines the lower higher order bias of EL and the robustness of ET to moment misspecification, as implied probabilities tend to distribute weights more evenly. In this paper, 
the tilting in the first step makes sure that $h$ has a density interpretation (similar to EL weights). Moreover, the moment conditions are not satisfied with respect to the approximate model $(f)$. Compared to other divergences, the KL distance $\int h ln(\frac{h}{f})dv$ is expected to perform better under this type of misspecification (similar to the robustness of ET).\footnote{In particular,  there is no need to have extreme values of $h$ on sample points at which the moment function $(m)$ is close to zero, because $h$ can be close to zero at points that have large values of $m$, even if $f>0$. This would not be possible with the reverse KL distance, which would require $f$ to be zero when $h=0$, making computation cumbersome.} 

Moreover, the parametric nature of the base density employed in this paper allows us to consider a wider class of probability measures that may be derived from structural models, expanding therefore the applicability of information projections to this class of models. 
As already mentioned, the base density can be generated by an explicit approximate solution to a set of optimality and equilibrium conditions, which could arise from both macroeconomic and microeconomic models.\footnote{When the approximate density is non-structural, then its choice can be informed by utilizing possible knowledge of the reduced form of the structural model and directly using such a form in constructing the base density without explicitly obtaining an approximate solution. In the (log) linearized DSGE case, for example, this corresponds to using a $VAR(p)$ or $VARMA(p,q)$ where $(p,q)$ can increase with the sample size, or a state space model in general. }
It is also worth commenting on why tilting using a conditional density (and a conditional moment restriction) might be preferable. In some potential applications, a time varying conditional density is needed to accommodate possible exogenous structural shifts in the economy. Correspondingly, the conditional moment itself can be time varying e.g. when agents display a time preference shift, which also implies endogenous time variation in the underlying density. In order to estimate the parameters governing the structural change jointly with the rest of the parameters, the tilted likelihood function has to be constructed sequentially. 

Finally, in the case in which $f(X|Z,\varphi)$ belongs to the exponential family and the moment conditions are linear, exponential tilting is a convenient choice as theoretical conditions
imposes additional structure to the moments of the base density. I present below an illustrative example of projecting on densities that satisfy moment 
conditions that arise from economic theory. Due to linearity, the resulting distribution
after the change of measure is conjugate to the base measure. 


\subsection{An Example from Asset Pricing} 
The consumption - savings decision of the representative household implies an Euler equation restriction on the joint stochastic process of consumption, $C_{t}$, and gross interest rate, $R_{t}$, where  $\mathcal{F}_{t}$ is the information set of the agent at time $t$ and $\mathbb{E}_{\mathbb{P}}$ signifies rational expectations :\vspace{-0.05 in} 
\begin{equation*}
\mathbb{E}_{\mathbb{P}}(\beta R_{t+1}U_{c}(C_{t+1})-U_{c}(C_{t})|\mathcal{F}_{t})=0
\end{equation*}
Suppose that the base model is a bivariate VAR for consumption and the
interest rate which, for analytical tractability, are not correlated. Their joint
density conditional on $\mathcal{F}_{t}$ is therefore: \small
\[
\left(\begin{array}{c}
C_{t+1}\\
R_{t+1}
\end{array}\mid\mathcal{F}_{t}\right)\sim N\left(\left(\begin{array}{c}
\rho_{c}C_{t}\\
0
\end{array}\right),\left(\begin{array}{cc}
1 & 0\\
0 & 1
\end{array}\right)\right)
\]
\normalsize
where I have set $\rho_{R}=0$ for analytical convenience.\footnote{Allowing for a non-zero mean in the interest rate is entirely possible but unnecessary. Importantly, the Euler equation has no information on  $\mathbb{E}_{t}R_{t+1}$. This is an example of an otherwise testable restriction on $\varphi$ that would not be undone by the projection.} For a quadratic utility function, that is $U(C_{t})=C_{t}-\frac{\gamma}{2} C_{t}^{2}$, the Euler equation implies that   $Cov(R_{t+1},C_{t+1}|\mathcal{F}_{t})=\frac{1}{\beta\gamma}(\gamma C_{t}-1)$.  The set of structural parameters is $\vartheta:=(\beta,\gamma)$.
Since the two variables are now correlated, the  distorted density $h(C_{t+1},R_{t+1}|\mathcal{F}_{t})$  will feature a new covariance structure: 
\[\left(\begin{array}{cc} Var_{t}(C_{t+1}) & \frac{1}{\gamma\beta}(\gamma C_{t}-1)\\
* & Var_{t}(R_{t+1})
\end{array}\right)\] where $Var_{t}(C_{t+1})=Var_{t}(R_{t+1})=\frac{1}{1-\mu_{t}^2}$, and ${\mu_{t}}$ has the interpretation of  the coefficient of projecting consumption on the interest rate, $Var^{-1}_{t}(R_{t+1})Cov(R_{t+1},C_{t+1}|\mathcal{F}_{t})$.
\normalsize
In Appendix \hyperref[AppA]{A}  I illustrate how the same expression can be obtained formally using a conditional density projection\footnote{More precisely, what is obtained is the density conditional on $Z=z$.}, that is, solving \eqref{principal}.  Notice that the projection is not simply a restriction on the reduced form parameters as it adds information. This is evident from a change in the conditional covariance matrix of $(C_{t+1},R_{t+1})$ from zero to a function of $(\mathcal{F}_{t},\vartheta,\varphi)$.  The only feature of the base model  that changes after the projection is the one that does not agree with the moment conditions e.g. the covariance matrix. Other features, such as the conditional means, do not change after the projection as they are consistent with the moment condition. 
Moreover, the fact that in this example
the Euler equation is a direct transformation of the parameters of the
base density is an artifact of the form of the utility function assumed,
and is therefore a special case. In general, an analytical
solution cannot be easily obtained and we therefore resort to solving a simulated version of \eqref{principal}.  \label{ex}

\subsection{Computational Details} $ $
With conditional moment restrictions, the projection involves computing Lagrange multipliers which are functions defined on $\Psi\times Z$. Thereby, the projection has to be implemented at each point $z_{i}$ when the tilted-likelihood function is constructed.

The general algorithm for the inner loop is therefore as follows:\footnote{In estimation, the tilted likelihood function is constructed at every proposal for the vector $\psi$. I experimented with pre-estimating the unknown functions $\mu(Z,\psi)$ and $\lambda(Z,\psi)$ by simulating at different points of the support of $Z$ and $\psi$ and using function approximation methods i.e. splines to compute multipliers at the sample points $z_{i}$ and proposals for parameter vector $\psi$ but I found this to be an inefficient way of implementing the algorithm; it is far better to compute the projection "online".}
\begin{enumerate}
	\item Given proposal for $(\varphi,\vartheta)$ and datum $z_{i}$, simulate $N_{s}$ observations
	from $F(x;z_{i},\varphi)$
	\item Compute:
	\begin{itemize}
		\item $\hat{\mu}(x;z_{i},\vartheta)=\arg\min\frac{1}{N_{s}}\sum_{j=1:N_s}\exp(\mu(z_{i},\psi)'m(x_{j};z_{i},\vartheta))$ and
		\item $\hat{\lambda}(x;z_{i},\vartheta)=-\log(\frac{1}{N_{s}}\sum_{j=1:N_{s}}\exp(\mu(z_{i},\psi)'m(x_{j};z_{i},\vartheta)))$ 
	\end{itemize}
	\item 	 Evaluate $log\left(h(x,z_{i};\psi)\right)$.
\end{enumerate}

The pseudo log-likelihood function can be constructed using $\sum_{i=1..N}log\left(h(x_{i},z_{i};\psi)\right)$. To give a sense of how much additional time it takes to compute a tilted likelihood function, for every likelihood point, with a single moment condition, computation time increases by an average of $1\%$ of a second in the simulation experiments I will refer to later on in the paper.\footnote{All computing times reported in the paper correspond to an Intel  
i5-8365U, $1.6$ Ghz processor.}

In the next section I analyze the frequentist properties of using the tilted density to estimate $\psi\equiv(\vartheta,\varphi)$. The main challenge is the fact that we project on a possibly misspecified density. Explicitly acknowledging for estimating the parameters of the density yields some useful insight to the 
behavior of the estimator. Moreover, although tilting delivers a density, as we will see the information contained in the resulting estimating equations for the structural parameter of interest is of semi-parametric nature and does not always lend itself to proper Bayesian analysis. The tilted density is thus best viewed as a statistical criterion function with an analog interpretation. When there is  an explicit mapping between the structural parameters and the parameters of the base density, then the scores of the latter do provide additional information. In this sense the statistical criterion function nests the conventional parametric likelihood function as a special case.
 Monte Carlo simulation techniques are nevertheless applicable in all cases, as in \cite{Chernozhukov2003293} and \cite{doi:10.3982/ECTA14525}. The frequentist results presented below determine the limiting behavior of this criterion function and are typically assumed to hold for the monte carlo simulation techniques, which can involve prior parameter distributions as well. 
\section{Setup and Assumptions}
Recall that we maximize the empirical analogue to \eqref{eq:agent}, which, abstracting from simulation error that comes from computing $(\mu,\lambda)$, is equivalent to choosing $\hat{\psi}$ to solve the following program: 
\begin{eqnarray*}
	&\underset{{(\vartheta,\varphi)\in\Theta\times\Phi} }{\max} Q_{N}(\vartheta,\varphi)\equiv\frac{1}{N}\sum_{i=1..N}\log\left(f(x_{i}|z_{i},\varphi)\exp(\mu_{i}'m(x_{i},z_{i},\vartheta)+\lambda_{i})\right)&
\end{eqnarray*}\small \begin{eqnarray*} 
 s.t. \quad	\forall i=1..N, & \mu_{i}: & \quad \int f(X|z_{i},\varphi)\exp(\mu_{i}'m(X,z_{i},\vartheta))m(X,z_{i},\vartheta)dX=0\\
	& \lambda_{i}:  & \int f(X|z_{i},\varphi)\exp(\mu_{i}'m(X,z_{i},\vartheta)+\lambda_{i})dX=1
\end{eqnarray*}
\normalsize           
where for notational brevity I substituted  $Z=z_{i}$ for $z_{i}$ and set $\mu_{i}=\mu(z_{i}),\lambda_{i}=\lambda(z_{i})$.
Denoting the Jacobian of the moment conditions by ${M}$, the first order conditions are the following, where $\mu_{\psi}(z_{i}),\lambda_{\psi}(z_{i})$ denote derivatives with respect to the corresponding parameter vector:
\begin{eqnarray}
	\vartheta: &\frac{1}{N}\underset{i}{\sum}\left({{M}(x_i,z_i,\vartheta)'\mu}(z_{i})+\mu_{\vartheta}(z_i)'{m}(x_i,z_i,\vartheta)+\lambda_\vartheta(z_i)\right)&=0\\
	\varphi: &\frac{1}{N}\underset{i}{\sum}\left(\mathfrak{s}(x_{i},z_{i},\varphi)+\mu_{\varphi}(z_i)'{m}(x_i,z_i,\vartheta)+\lambda_{\varphi}(z_i)\right)&=0  \label{FOC}
\end{eqnarray}
where $\mathfrak{s}(.)\equiv\frac{\partial}{\partial \varphi}log( f(X|z_{i},\varphi))$ is the score function of the base density and  
\begin{eqnarray}
	&\mu(z_i)=\underset{\mu\in\mathbb{R}^{n_{m}}}{\arg\min} \int f(X|z_i,\varphi)\exp(\mu'm(X,z_i,\vartheta)dX \label{eq:mudef}\\ 
	&\lambda(z_i)=-\log\left(\int f(X|z_{i},\varphi)\exp(\mu(z_i)'m(X,z_{i},\vartheta)dX\right) &   
\end{eqnarray}

\subsection*{ASSUMPTIONS I}

Assume a stationary
ergodic sequence of vectors of random variables $\{x_{i}',z_{i}',\}_{i=1}^{N}$ with summable autocovariances  together in an appropriate filtration $\mathcal{F}_{i}$ so that $(x_{i}',z_{i}')$ are measurable with respect to $(\mathcal{F}_{i},\mathcal{F}_{i-1})$ respectively. 
Furthermore, define innovations $\{\sigma^{l}_{i-1}\epsilon^{l}_{i}\}_{l=1..n_{m}}$ where
 $\epsilon^{l}_{i}$ a martingale difference sequence, $\sigma^{l}_{i-1}$ is $\mathcal{F}_{i-1}$-measurable and $\mathbb{E}_{P_{z}}\epsilon^{{l}^{2}}_{i}=1$.
 Moreover:
\begin{enumerate}[label=\roman*.]
	\item (\textbf{COMP)} $\Theta\subset\mathbb{R}^{n_{\vartheta}},\Phi\subset\mathbb{R}^{n_{\varphi}}$
	are compact. 
	\item (\textbf{ID)}$\exists ! \psi^{\star}\in int(\Psi):\psi^{\star}=\arg\underset{\psi\in\Psi}{\max}\mathbb{E}\log h(X|Z,\psi)$
	\item (\textbf{BD-1}a)$\forall l \in {1..n_{m}}\text{ and for \ensuremath{d\geq4}}, P\in\left\{F_{\varphi},\mathbb{P}\right\}: \\ \mathbb{E}_{P_{z}}\sup_{\psi}\|m^{l}(x,z,\vartheta)\|^{d}, \mathbb{E}_{P_{z}}\sup_{\psi}\|m^{l}_{\vartheta}(x,z,\vartheta)\|^{d} \mbox{\text{}} \text{and } \mathbb{E}_{P_{z}}\sup_{\psi}\|m^{l}_{\vartheta\vartheta}(x,z,\vartheta)\|^{d}  $ are finite, $\mathbb{P}({z})-a.s$\footnote{Here, the subtlety is that it has to
		hold for the base measure and the true measure.  Given absolute continuity of $d\mathbb{P}(X|Z)$ with respect to $dF(X|Z)$, the existence of moments under $\mathbb{P}(X|Z)$ is sufficient for the existence of moments under $F(X|Z)$.}. 
	\item (\textbf{BD-1}b)$\sup_{\psi}\mathbb{E}_{\mathbb{P}_{z}}\|e^{\mu(z)'|m(x,z,\vartheta)|}\|^{2+\delta}<\infty$
	for $\delta>0$, $\forall \mu(z)>0, \mathbb{P}({z})-a.s. $ \footnote{
		Note that \textbf{BD-1a} and \textbf{BD-1b} imply that 
		$\sup_{\psi}\mathbb{E}_{\mathbb{P}_{z}}\|e^{\mu(z)'m(x,z,\vartheta)+\lambda(z,\vartheta)}m(x,z,\vartheta_{0})\|^{1+\delta}<\infty$
		for $\delta>0$ and $\forall z$.}
	\item (\textbf{BD-2})$\mathbb{E}_{\mathbb{P}}sup_{\psi}\mid\log f(x|z,\psi)\mid^{2+\tilde{\delta}}<\infty$
	where $\tilde{\delta}>0$.
	\item \textbf{(PD-1) }For any non zero vector $\xi$ and closed $\mathcal{B}_{\delta}(\psi^{\star})$
	, $\delta>0$ and for all $P_{z}$ such that $\mathbf{KL}(P_{z},F_{\varphi,{z}})\leq \mathbb{E}_{P_{z}}e^{\mu(z)'m(x,z,\vartheta)}$, the following holds, $\mathbb{P}({z})-a.s$: \\ $0<\inf_{\xi\times\mathcal{B}_{\delta}(\psi^{\star})}\xi'\mathbb{E}_{P_{z}}{m}(x,\vartheta){m}(x,z,\vartheta)'\xi<\sup_{\xi\times\mathcal{B}_{\delta}(\psi^{\star})}\xi'\mathbb{E}_{P_{z}}{m}(x,z,\vartheta){m}(x,\vartheta)'\xi<\infty$
\end{enumerate}

   \normalsize        
Assumptions (i)-(ii) correspond to typical compactness
and identification assumptions found in \cite{Newey19942111} where (ii) is assumed to also hold for the objective function in the case in which $h(X|Z)$ does not collapse to $f(X|Z)$ (thus $\psi^{\star}$ is the pseudo-true value\footnote{It is a theoretical possibility that misspecification can push pseudo-true values to the boundary but catering for this in the econometric analysis goes beyond the scope of this paper.}) while
(iii) assumes uniform boundedness of conditional moments and their first and second derivatives, up to a set
of measure zero. Assumption (iv) assumes existence of exponential absolute $2+\delta$ moments and   (v) assumes the existence of $2+\delta$ absolute moments of the log-likelihood based on the base density.\footnote{Since $log(h(x|z,\psi))=\mu(z,\psi)'m(x,z,\vartheta)+\lambda(z,\psi)+\log f(x|z,\psi)$, existence of $\mu(z,\psi)$ and $\lambda(z,\psi)$, together with \textbf{BD-1a,b} and \textbf{BD-2} imply that $\mathbb{E}_{\mathbb{P}}\sup_{\psi}\mid\log h(x|z,\psi)\mid^{2+\tilde{\delta}}<\infty$.}  Finally, (vi) assumes
away pathological cases of perfect correlation between moment conditions under any measure $P_{z}$ that is contiguous and within a certain distance to $F_{\varphi,{z}}$.

\subsection{Existence and feasibility of the projection at $z_{i}$}
 \cite{Komunjer_Ragusa} provide primitive conditions for the case of projecting using a divergence that belongs to the $\phi-$ divergence class and moment restrictions that have unbounded moment functions. Assumptions \textbf{BD-1a} and \textbf{BD-1b} are sufficient for their primitive conditions (Theorem 3). More particularly, the feasibility assumption of \cite{Komunjer_Ragusa} i.e. there exists at least one measure $H\in \mathcal{H}$ such that $\int log\left(\frac{dH}{dF}\right)dH<\infty$ is trivially satisfied by $H=\mathbb{P}$, the true probability measure, as long as the Kullback Leibler distance of the  approximate measure $F$ to the true measure $\mathbb{P}$ is well defined. In this case, the Kullback - Leibler distance minimized at each projection step becomes as follows: 
\[\int \log\left(\frac{dH(.;z_{i},\vartheta)}{dF(.;z_{i},\varphi)}\right)dH(.;z_{i},\vartheta) = \int \log\left(\frac{d\mathbb{P}(.;z_{i},\vartheta)}{dF(.;z_{i},\varphi)}\right)d\mathbb{P}_{z_{i}}\]
This distance will be finite as long as the base density has a finite \textbf{KL} distance with respect to $\mathbb{P}(.)$ for some $\varphi\in \Phi$.  
If $F(.;z_{i},\varphi)$ is induced by a structural approximation, the distance will be finite if the approximation error is uniformly controlled  over $z$. 
To see why, consider the density $f(x_{t};z_{t},\hat{\gamma}(z;\vartheta),\vartheta)$ induced by a set of approximations $\hat{\mathbf{\gamma}}(z;\varphi)$ to functions $\gamma(z)$, such as policy functions or value functions. The latter induces the true density, ${p}(x_{t};z_{t},{\gamma}(z),\vartheta)$. Then, using a functional Taylor expansion 
\begin{eqnarray*}
	\log\left(\frac{p(x;z,\gamma(z),\vartheta)}{f(x;z,\hat{\gamma}(z;\vartheta),\vartheta)}\right)&=& \sum_{k=1:K}\frac{1}{k!}\nabla^{k}{\log f }_{\gamma}(x;z,\hat{\gamma}(z;\vartheta),\vartheta)[\hat{\gamma}-\gamma] + O_{p}(\mid\mid\hat{\gamma}-\gamma\mid\mid^{K+1})
\end{eqnarray*} Hence, if the (log) density $f$ is $K$ times pathwise differentiable and the approximation method controls the error uniformly in $z_{i}$ i.e. $\sup_{z}||\hat{\gamma}-\gamma||<\infty$, then all of the terms on the RHS are bounded. Existing work has established uniform upper bounds on errors in value function iteration \citep{10.2307/2998564} and policy function iteration \citep{doi:10.1137/S0363012902399824}, as well as bounds on invariant distributions of simulations of approximated models \citep{RePEc:ecm:emetrp:v:73:y:2005:i:6:p:1939-1976}.

In practice, at each parameter proposal by the estimation algorithm it is easy to check whether the projection exists, as the Lagrange multiplier vector $\mu$ takes implausibly high absolute values, and hence  the  parameter draw is rejected. To see this, notice that in \eqref{eq:mudef}
if $m(X,Z,\vartheta)$ is uniformly positive or negative across $X$ for any $Z$, which implies that the moment condition cannot be satisfied, then optimal $\mu$ is either positive or negative infinity. 

\section{Asymptotic Properties}
This section contains the consistency and asymptotic distribution results for $\hat{\psi}$. Appendix \hyperref[AppA]{A}  (\cref{muder}) provides expressions for the first and second order derivatives of  $(\mu(z_i),\lambda(z_i))$ which determine the behavior of $\hat\psi$ in the neighborhood of $\psi^{\star}_0$. These expressions will be useful for the characterization of the properties of the estimator.
I first outline a result which is useful in understanding the properties of the estimator.   \vspace{-0.15 in}
\begin{Lemma}
		For a base conditional density $F_{z}$ with distance to $\mathbb{P}_{z_{i}}$ equal to $\chi^{2}(F_{z_{i}},\mathbb{P}_{z_{i}})$:  \\
		(a) $\mu_{i}=O_{p_{z}}(\chi^{2}(F_{z_{i}},\mathbb{P}_{z_{i}})^{\frac{1}{2}})$\\
		(b)  $ \max_{i}\sup_{\vartheta}|\mu_{i}'m(\vartheta,x_{i})|=O_{p_{z}}(\max_{i}\chi^{2}(F_{z_{i}},\mathbb{P}_{z_{i}})^{\frac{1}{2}}N^{\frac{1}{d}})   \label{mu}
	$\end{Lemma}\vspace{-0.2 in}
\begin{proof}
	See Appendix \hyperref[Proofs]{A}. 
\end{proof}  \vspace{-0.2 in}
When $\sup_{i}\chi^{2}(F_{z_{i}},\mathbb{P}_{z_{i}}) \sim N^{-\xi}$, $\mu_{i}=o_{p_{z}}(1)$. If  $\frac{2}{d}<\xi$, $\max_{i}\sup_{\vartheta}|\mu_{i}'m(\vartheta,x_{i})|=o_{p}(1)$. The Lagrange multipliers from imposing the equilibrium conditions will converge to zero under vanishing approximation error.

\subsection{Consistency}
The uniform consistency of the estimator is shown by first proving pointwise consistency and then stochastic equicontinuity of the objective function. Details of the proof are in Appendix \hyperref[AppA]{A}. Under fixed misspecification, the estimator is consistent for $\psi_{0}^{\star}$, defined below.  
\begin{Definition}
  The true value is denoted by $\psi_{0}$. The pseudo-true value $\psi_{0}^{\star}$ is the $\psi\in\Psi$ that minimizes the Kullback-Leibler $(\mathbf{KL})$ distance between $H(X,Z,\psi)$ and $\mathbb{P}(X,Z)$,  decomposed as follows:
  \small
  \begin{eqnarray} 
  0&\leq& \mathbb{E}_{\mathbb{P}}\log\left(\frac{d\mathbb{P}_{z}}{dF_{(\varphi,z)}}\right) - \mathbb{E}_{\mathbb{P}}\left[\mu'({\psi,z})\mathbb{E}_{\mathbb{P}_{z}}m(X,\vartheta)\right]+\mathbb{E}_{\mathbb{P}}\left[\log\left(\mathbb{E}_{F_{(\varphi,z)}}\exp(\mu'({\psi,z})m(X,\vartheta))\right)\right]\label{eq:def1}
  \end{eqnarray}\normalsize  
 Correspondingly, since $\mu({\psi,z}):=\underset{\mu}{\arg\min}\left(\mathbb{E}_{F_{(\varphi,z)}}\exp(\mu'm(X,\vartheta))\right)$,  $\varphi^{\star}_{0}$ is the value of $\varphi\in\Phi$ such that $F(X|Z,\varphi)$ is as close as possible (in $\mathbf{KL}$ units) to $\mathbb{P}(X|Z,\varphi)$ {and} satisfies \[\mathbb{E}_{F_{(\varphi,z)}}\exp(\mu'({\psi,z})m(X,\vartheta^{\star}_{0}))m(X,\vartheta^{\star}_{0})=0\] where $\vartheta^{\star}_{0}$ is the value of $\vartheta\in\Theta$ such that $\mathbb{E}_{F_{(\varphi^{\star},z)}}\exp(\mu'({\psi,z})m(X,\vartheta))m(X,\vartheta)=0  \label{def_psistar}$. 
\end{Definition} 
At the pseudo-true value, $H_{(\psi_{0}^{\star},z)}$ is the closest parametric distribution to $\mathbb{P}_{z}$, while both distributions satisfy a \textit{common} moment restriction, $\mathbb{E}_{H_{(\psi^{\star},z)}}m(X,\vartheta^{\star})=\mathbb{E}_{\mathbb{P}}m(X,\vartheta_{0})=0$. The tilted distribution will satisfy the exact moment conditions, and will be -by construction- closer to the distribution implied by the economic model. The smaller $\mathbf{KL}(F,\mathbb{P})$ is, the closer  to zero are the last two terms in \eqref{eq:def1}.\footnote{Note that there is a similarity between our definition of $\psi_{0}^{\star}$ to the definition of \citet{10.2307/2646780}, but in our case the moment restriction is satisfied by the sampled population, asymptotically.}  The interpretation of $\psi_{0}^{\star}$ can be qualified further by examining the profiled $\mathbf{KL}$ distance.

\begin{Lemma}
The profiled (over $\mu_{z}$) distance $\mathbf{KL}(H_{\psi},\mathbb{P})$ in  \eqref{eq:def1} is approximately equivalent to
 \small
 \begin{eqnarray*} 
 \mathbb{E}_{\mathbb{P}}\log\left(\frac{d\mathbb{P}_{z}}{dF_{(\varphi,z)}}\right)+  \mathbb{E}_{\mathbb{P}}\left[\mathbb{E}_{F_{(\varphi,z)}}m(X,\vartheta)V^{-1}_{F_{(\varphi,z)}}\left(\mathbb{E}_{\mathbb{P}_{z}}m(X,\vartheta)-\mathbb{E}_{F_{(\varphi,z)}}m(X,\vartheta) \right)\right]  \label{pseudo_true}
 \end{eqnarray*}\normalsize    
	\end{Lemma}
\begin{proof}
	Follows from \cref{def_psistar}, the implicit map of $\mu$ in the proof of \cref{mu} and a first order approximation of the third term  in \eqref{eq:def1}.
\end{proof}
If $(\varphi,\vartheta)$ are independent, then minimizing this criterion is equivalent to finding the value of $\varphi$ that makes the approximate model as close as possible to the true data generating process for any level of $\vartheta$, which simultaneously minimizes both the first and second term, while for a fixed value of $\varphi$,  minimizing over $\vartheta$ seeks to bring $\mathbb{E}_{F_{(\varphi,z)}}m(X,\vartheta)$ as close as possible to zero, that is, the approximate model satisfying the moment condition. 
A similar interpretation holds when there is a mapping from the reduced form to the structural parameters, which arises when we consider approximate equilibrium models.  As mentioned in the introduction (see \eqref{eq:intro_mom}) and explicitly shown later in the paper, the first order condition will be equivalent to setting to zero the orthogonalized combinations of the scores of the approximate model and the efficient GMM first order condition.


Having defined the pseudotrue values, I present below the asymptotic results.

 \vspace{-0.1 in}
\begin{theorem}
	Consistency for $\psi_{0}^{\star}$\\
	Under Assumptions I :
	$
	(\hat{\vartheta},\hat{\varphi})\underset{p}{\rightarrow}(\vartheta_{0}^{\star},\varphi_{0}^{\star})   \label{consstar}
	$
\end{theorem}\vspace{-0.2 in}
\begin{proof}
	See Appendix \hyperref[Proofs]{A}. 
\end{proof}\vspace{-0.2 in}
\begin{Definition}{Vanishing Approximation Error.}$ $\\
Assuming that densities are twice differentiable, define the metric \[\Delta(F,P)=\sup_{k\in\{0,1,2\}}\mathbb{E}_{{P}}\left( D^{k}\left(\frac{f}{p}-1\right)\right)^{2}\] which measures how well density $f$ approximates $p$ up to the second derivative.\footnote{Define $D^{1}:=\sup_{l}\mid\frac{\partial (.)}{\partial x_{l}}\mid$ and $D^{2}:=\sup_{ll'}\mid\frac{\partial^{2} (.)}{\partial x_{l}\partial x_{l'}}\mid$. Uniform convergence of $\Delta$ to zero implies that the approximation error vanishes up to the second derivative, which ensures uniform convergence of the scores and the Hessian (elementwise). Note that for $k=0$, $\mathbb{E}_{{P}}\left( D^{0}\left(\frac{f}{p}-1\right)\right)^{2}$ is the $\chi^{2}$ distance between $f$ and $p$. }
An asymptotically correct density is defined as the density $f(X|Z,\varphi)$ that converges uniformly to the true conditional density, such that $\sup_{\varphi}\sup_{i}\Delta(F_{z_{i}},\mathbb{P}_{z_{i}})\to 0$.\label{def2}
\end{Definition}  
\vspace{-0.2 in}
Therefore, under vanishing approximation error, consistency is for $\vartheta_0$:\vspace{-0.1 in}
\begin{Corollary}
	If $\sup_{\phi}\sup_{i}\Delta(F_{z_{i}},\mathbb{P}_{z_{i}})\to 0$, then   $\vartheta_{0}^{\star}=\vartheta_0$.  \label{cons0}
\end{Corollary}\vspace{-0.2 in}
\begin{proof}
	See Appendix \hyperref[Proofs]{A}. 
\end{proof} \vspace{-0.1 in}

\subsection{Asymptotic Distribution}
The limiting distribution
of the estimator is derived using a first order approximation around the pseudotrue parameter $\psi_{0}^{\star}$. The results that follow hold for a generic base density $f(X|Z,\varphi)$, where I distinguish between non-structural approximations where $\varphi$ and $\vartheta$ are not functionally related, and structural approximations, where the two are related. Despite that the main focus in the case of a structural approximation, in each set of results it is instructive first to deal with the case of a non-structural approximation.
\subsubsection{Non-Structural Approximations}
The estimator for $\vartheta$ that results from a non-structural approximation of the density converges to the same distribution as estimators employed in the semi-parametric literature, where the density is treated as a nuisance parameter, estimated using  a (semi)non-parametric approximation. 
The next set of results demonstrate asymptotic Normality while semi-parametric efficiency for $\vartheta$ is retained if $N^{-\frac{1}{2}}\Delta\to 0$.
\begin{theorem}
	Asymptotic Distribution when $\Psi:=(\Phi,\Theta)\in \mathbb{R}^{dim \varphi}\times \mathbb{R}^{dim \vartheta}$\\
	Under Assumption I and $N_{s}$,$N\rightarrow\infty$ such that
	$\frac{N}{N_{s}}{\to} 0$ 
	\begin{eqnarray*}
	N^{\frac{1}{2}}(\psi-\psi_{0}^\star)\underset{d}{\rightarrow} \mathcal{N}\left(0,\bar{V}(\psi_0^\star)\right)
	\end{eqnarray*} \label{Norm}
\end{theorem}  \vspace{-0.6 in}
\begin{proof}
	See proof of Theorem \ref{Norm} in Appendix \hyperref[Proofs]{A}  .
\end{proof}

The assumption 	$\frac{N}{N_{s}}{\to} 0$ obviates the need to account for simulation error in the variance covariance matrix. This can be relaxed i.e.	it can converge to a positive number but it adds no new insights. 
In Appendix \hyperref[AppA]{A}  I derive the exact form of the covariance matrix of the estimator. 
Under (asymptotic) correct specification  $\sup_{\varphi}\sup_{i}\Delta(F_{z_{i}},\mathbb{P}_{z_{i}})\equiv \kappa^{-1}_{n}\sim N^{-\xi}$, $\mathfrak{s}_{i}\to s_{i}$ where $s_{i}$ is the score of the true density. For $\xi>1$, $ G(\psi)\equiv -\mathbb{V}_{g}(\psi)$,  $\bar{V}(\psi_0^\star)=\bar{V}(\psi_0)$. 
 In particular, the variance is block diagonal (I drop dependence on $(\vartheta,\varphi)$): 
\begin{eqnarray*}
	\bar{V}(\psi_0)& \equiv & \left(\begin{array}{cc}
		\left(\mathbb{E}(\mathbb{E}({M}|z_{i})'\mathbb{V}^{-1}_{m,i}\mathbb{E}({M}|z_{i}))\right)^{-1} & 0 \\
		0  & \left(\mathbb{E}{s}_{i}{s}_{i}'-\mathbb{E}\mathfrak{B}_{i}\mathbb{V}_{m,i}\mathfrak{B}_{i}'\right)^{-1}
	\end{array}\right)
\end{eqnarray*} \normalsize
where $\mathfrak{B}_{i}'=\mathbb{V}^{-1}_{m,i}\mathbb{E}(m{s}'|z_{i})$ is the coefficient of projecting the scores on the moment conditions, where I abbreviate $V_{\mathbb{P}_{z},m}\equiv\mathbb{V}_{m,i}$.
The upper left component of $\bar{V}(\psi_0)$ is the same as the (inverse) information matrix corresponding to $\vartheta$ when the conventional optimally weighted GMM criterion is employed. 
What this implies is that lack of  knowledge of $\varphi$ does not affect the efficiency of estimating the structural parameters $\vartheta$, at least asymptotically. Also, note that $\hat\vartheta$ attains the SLB even if $\xi>\frac{1}{2}$, as the off-diagonal block of the Jacobian (cross-derivatives) converges to zero as long as $\xi>0$. 

The expressions above have an intuitive interpretation. If the moment conditions used span the same space spanned by the scores of the density, then $\bar{V}$ trivially attains the Cramer - Rao bound in the upper left component as $m\equiv \mathfrak{s}$ and the covariance matrix becomes singular as both $m$ and $\mathfrak{s}$ give the same information.
Conversely, the less predictable is the score from the additional moment conditions used (that is, $\|\mathfrak{B}_{i}\|$ is close to zero), the higher the efficiency attained for estimating $\varphi$, where $\left(\mathbb{E}{s}_{i}(\varphi){s}_{i}(\varphi)'\right)^{-1}$, is the lowest variance possible among regular estimators. 
This also corroborates the claim that the estimator is not equivalent to a constrained MLE as in the latter case, constraints typically increase efficiency for estimating the reduced form. In this case, by using the moment conditions which involve unknown parameters ($\vartheta$) can only increase the variance for estimating $\varphi$ as the former are not essential to the identification to the latter when $(\varphi,\vartheta)$ are independent.\footnote{ The best analogy can be made to including irrelevant regressors in multivariate regression: as long as the additional regressors are not essential to the identification of the causal effect, the result is an increase in the variance of its estimate.}
\subsubsection{Structural Approximations}
If the base density is obtained by solving a structural model, there is an explicit mapping between $\varphi$ and $\vartheta$, and the first order conditions are a linear combination of the conditions analyzed above as $
\frac{dQ_{N}}{d\vartheta} = \frac{\partial Q_{N}}{\partial \vartheta} +  \frac{\partial Q_{N}}{\partial \varphi}\frac{d \varphi}{d \vartheta} =0 
$.
When the solution is exact, $f(X|Z)$ is pinned down by a unique parametric sub-model that automatically satisfies the moment conditions and $(\mu_{i},\lambda_{i})$ are zero for all $i$, which implies that both $\frac{\partial Q_{N}}{\partial \vartheta}(\hat{\vartheta}) $ and $\frac{\partial Q_{N}}{\partial \vartheta}(\hat{\vartheta}) $ will tend to zero as $\mathbb{E}M_{i}'V^{-1}_{{i},m}m_{i}=0$ and $\mathbb{E}{s}_{i}=0$ at $\vartheta_{0}$, while the second term in the variance of $\varphi$ vanishes as well.  
Moreover, $\mathbb{E}{s}_{i}(\varphi){s}_{i}(\varphi)'=\mathbb{E}\frac{d\vartheta}{d\varphi}'\mathbb{E}{s}_{i}(\vartheta){s}_{i}(\vartheta)'\frac{d\vartheta}{d\varphi}>\mathbb{E}\frac{d\vartheta}{d\varphi}'\mathbb{E}({M}|z_{i})'V^{-1}_{m,i}\mathbb{E}({M}|z_{i})\frac{d\vartheta}{d\varphi}$. The moment restrictions are trivial and add nothing to the information embedded in the scores: \vspace{-0.1 in}
\begin{theorem}
	Asymptotic Distribution when $\Upsilon \in \mathbb{C}^{2} : \Theta \to \Phi$. \\
	Under Assumption I and $N_{s}$,$N\rightarrow\infty$ such that
	$\frac{N}{N_{s}}{\to} 0$ :	
	\begin{eqnarray*}
		N^{\frac{1}{2}}(\vartheta-\vartheta_{0}^\star)\underset{d}{\rightarrow} \mathcal{N}\left(0,\bar{V}(\vartheta_0^\star)\right)
	\end{eqnarray*} \label{Norm2}
where for $\sup_{\varphi}\sup_{i}\Delta(F_{z_{i}},\mathbb{P}_{z_{i}})\equiv \kappa^{-1}_{N}\sim N^{-\xi}$ with $\xi>1$,  $\bar{V}(\vartheta_0)^{-1} \equiv \left(\frac{d\varphi}{d\vartheta}\right)'\mathbb{E}{s}_{i}{s}_{i}'\frac{d\varphi}{d\vartheta}$.
	\end{theorem}  
\begin{proof}
	See proof of Theorem \ref{Norm2} in Appendix \hyperref[Proofs]{A} .
\end{proof}
\subsection{Rate(s) of Convergence}
With a non-structural base density, the discrepancy with the true model does not have first order effects for $\hat{\vartheta}$ if $\xi>\frac{1}{2}$, and for $\hat{\varphi}$ if $\xi>1$.\footnote{
	\cite{Geweke_Hahn_Ackerberg} obtains a similar type of result  for approximated likelihoods (corresponding to $\xi>1$). Note that $\xi=1$ corresponds to the case when the approximate density is defined by parameters that are within a root-n neighborhood to the true density, and hence $\Delta(F_{z_{i}},\mathbb{P}_{z_{i}})=O_{p_{z}}(N^{-1})$.} 
As shown in the proof of Theorem \ref{Norm},\small
\begin{eqnarray*}
	{N^{\frac{1}{2}}g_{1,N}}=N^{-\frac{1}{2}}\sum_{i} M_{f_{i}}'V^{-1}_{f_{i},m}m_{i}+  O_{p}(N^{\frac{1}{2}-\xi}),\quad 
	{N^{\frac{1}{2}}g_{2,N}}=N^{-\frac{1}{2}}\underset{i}{\sum}({s}_{i} -{s}_{f_{i}}+\mu_{i,\varphi}'m_{i})+ O_{p}(N^{\frac{1}{2}-\frac{\xi}{2}}) 
\end{eqnarray*}	\normalsize
When $\xi<\frac{1}{2}$, the first term in $g_{1,N}$ converges to a normal random vector while the last term diverges. Hence, the whole term diverges. Similarly, when  $\xi<1$, the whole term $g_{2,N}$ diverges.
To have a meaningful limit in this case, we have to normalize with $N^{\xi}$ and $N^{\frac{\xi}{2}}$ respectively:\small
\begin{eqnarray*}
	{N^{\xi}g_{1,N}}= N^{\xi-\frac{1}{2}} N^{-\frac{1}{2}}\sum_{i} M_{f_{i}}'V^{-1}_{f_{i},m}m_{i}+ O_{p}(1),\quad 
		{N^{\frac{\xi}{2}}g_{2,N}}=N^{\frac{\xi}{2}-\frac{1}{2}}N^{-\frac{1}{2}}\underset{i}{\sum}({s}_{i} -{s}_{f_{i}}+\mu_{i,\varphi}'m_{i})+ O_{p}(1)
\end{eqnarray*}	\normalsize   
and therefore the rate of convergence for parameter estimates is slower than root-n. It is important to note that  $\xi>\frac{1}{2}$ allows for a wide range of cases for local-misspecification, with approximate densities converging at a much slower rate than the standard parametric rate $(\xi=1)$. This is more interesting in the case of structural approximation: The information added by tilting the model to satisfy the moment condition will indeed be driven by the efficient GMM first order condition and will not depend to first order on the degree of local misspecification, even if $\xi\leq 1$:\small \[N^{\frac{\xi}{2}}\left(g_{1,N}+\frac{d\varphi}{d\vartheta}'g_{2,N}\right)=N^{\frac{\xi}{2}-1}\sum_{i} M_{f_{i}}'V^{-1}_{f_{i},m}m_{i} + O_{p}(N^{-\frac{\xi}{2}})+  \frac{d\varphi}{d\vartheta}'N^{\frac{\xi}{2}-1}\underset{i}{\sum}\left({s}_{i} -{s}_{f_{i}}+\mu_{i,\varphi}'m_{i}\right))+ O_{p}(1)\]\normalsize
\section{Further Properties under Local Misspecification}
The previous analysis considered the asymptotic distribution under a fixed data generating process. I next analyze the asymptotic distribution under drifting parameter sequences, which is more useful to understand the behavior of the estimator under standard local misspecification $(\xi=1)$. In particular, 
I consider approximation errors that can be represented by non-linear restrictions imposed by the approximate model on the true data generating process as follows: $r(\varphi)=0$, where $r$ is some non-linear function of the reduced form parameters. 

One example of restrictions, which is also the one I will consider in the application, is ignoring non-linear effects in dynamic forward looking models. 
For illustration, consider two popular methods of approximating the solutions to dynamic equilibrium models, perturbation and projection. Suppose that $dim(X)=1$.
In the case of perturbation around $\bar{X}$, the Taylor approximation of a policy function results in $C(X,J;\vartheta) = C(\bar{X};\vartheta) + C_{1}(\bar{X};\vartheta)(X-\bar{X})  +...+ O(\|X-\bar{X}\|^{J+1})$ while in the projection case, the approximation is of the form $C(X,J,\alpha,\beta;\vartheta)=\sum^{J}_{j=1..J}\alpha_{j}b_{j}(X,\beta)$ where $b_{j}(X,b)$ is a basis function and $\alpha$ are the  projection coefficients. 
Assuming there exist  finite  $({J}^{\star},b^{\star},\alpha^{\star})$ such that approximation errors are negligible, restricting the number of approximating terms or the functional form of the basis $b_{j}(X,b)$ restricts the moments of  $C$ (and hence $\varphi$) as $\mathbb{E}C^{l}(X,J,\beta;\vartheta)$ is a function of $(J,\beta)$.

Adopting the local asymptotic experiment approach, see for example \cite{BHansen_2016}, I investigate convergence in distribution along sequences $\psi_N$ where $\psi_N= \psi^{\star}_0 + c N^{-\frac{1}{2}}$ for $\psi_N$ the true value, $\psi^{\star}_0\in\Psi$ the centering value and $c$ the localizing parameter. The true parameter is therefore ''close'' to the restricted parameter space up to arbitrary $c$, and hence misspecification is arbitrary. \vspace{-0.1 in}
\begin{theorem} Asymptotic Distribution when $\Psi:=(\Phi,\Theta)\in \mathbb{R}^{dim \varphi}\times \mathbb{R}^{dim \vartheta}$\\
	For $R(\varphi)\equiv \frac{\partial}{\partial \varphi}r(\varphi)$, 
	$\bar{G}^{-1}\equiv\left(\begin{array}{cc}
	G^{11} & G^{12}\\
	G^{21} & G^{22}
	\end{array}\right)$, $S_1\equiv[I_{n_\vartheta},0_{n_\vartheta\times n_\varphi}]$,  $S_2\equiv[ 0_{n_\varphi\times n_\vartheta}, I_{n_\varphi}]$ 
		Under assumptions I such that $ N^{\frac{1}{2}}\bar{G}(\tilde\psi)^{-1}g(\psi_N)\underset{d}{\to} \mathcal{Z}\sim \mathcal{N}(0,\Omega)$: \vspace{-0.1 in}
	\begin{enumerate}		
		\item $N^{\frac{1}{2}}(\hat\vartheta-\vartheta_N)\underset{d}{\to}  S_1\mathcal{Z}$ 
			\item $N^{\frac{1}{2}}(\hat\varphi-\varphi_N)\underset{d}{\to} \mathcal{Z}_r \equiv S_2\mathcal{Z}-G^{22}(\psi^{\star}_0)R(\varphi^{\star}_0)(R(\varphi^{\star}_0)'G^{22}(\psi^{\star}_0)R(\varphi^{\star}_0))^{-1}R(\varphi^{\star}_0)'(S_2(\mathcal{Z}+c))$
		\item For any non zero vector $\xi$, $\xi'(\mathbb{V}(S_2\mathcal{Z})-\mathbb{V}(\mathcal{Z}_r))\xi \geq 0$
	\end{enumerate}  \label{Shrinkage} \vspace{-0.1 in}
\end{theorem} 
\begin{proof}
		See proof of Theorem \ref{Shrinkage} in Appendix \hyperref[Proofs]{A} .
\end{proof} \vspace{-0.1 in}
There are two main implications of \cref{Shrinkage} for $\hat\vartheta$. First, for $c\neq0$, the asymptotic distribution of $\hat{\varphi}$ is non regular i.e. the distribution depends on $c$ (see p. 115 in \cite{VDV}). 
Second, the variance of $\hat\varphi$ is lower than the Cramer-Rao bound for regular estimators. For $\varphi_N$ arbitrarily close to the restricted subspace of $\varphi_0$, efficiency increases. The distribution of $\hat\vartheta$ is regular and its variance reaches the SLB. 
With a structural approximation, $\hat\vartheta$ inherits the bias-variance trade-off due to the restrictions imposed by the approximation. The next result clarifies how this trade-off arises.\vspace{-0.2 in}
\begin{Lemma}  Estimating equation when $\Upsilon \in \mathbb{C}^{2} : \Theta \to \Phi$. \\
Since $\xi=1$, the estimating equation is equal to 
	\begin{eqnarray}\frac{1}{N}\sum^{N}_{i=1}\left({M}_{f_{i}}'V^{-1}_{f_{i},m}m_{i}\right)+ \frac{d\varphi}{d\vartheta}'\frac{1}{N}\sum^{N}_{i=1}\left(\mathfrak{s}_{i}-\mathbb{E}_{h}(\mathfrak{s}|z_{i})- \mathbb{E}_{h}(\mathfrak{s}m'|z_{i})V^{-1}_{h_{i},m}m_{i}\right)=O_{p}(N^{-\frac{1}{2}})
	\end{eqnarray} which in the limit is equivalent to 
	\begin{eqnarray}\mathbb{E}_{\mathbb{P}}\left({M}_{i}'V^{-1}_{m_{i}}\mathbb{E}_{\mathbb{P}}(m|z_{i})+ \frac{d\varphi}{d\vartheta}'\left({s}_{i}-\mathbb{E}_{\mathbb{P}}({s}|z_{i})- \mathbb{E}_{\mathbb{P}}({s}m'|z_{i})V^{-1}_{m_{i}}\mathbb{E}_{\mathbb{P}}(m|z_{i})\right)\right)=0
	\end{eqnarray}
	\label{FOC}
\end{Lemma}\vspace{-0.4 in}
\begin{proof}
	The result is a direct application of Lemmas \ref{muder} in Appendix \hyperref[AppA]{A} and \ref{lemvar} in Appendix \hyperref[AppB]{B} to the first order conditions in the proof of Theorem \ref{Norm2}.
\end{proof}

The condition in Lemma \ref{FOC} can be read in two ways: The first is that in the absence of the projection, the corresponding estimating equation would simply be $\frac{d\varphi}{d\vartheta}'\mathbb{E}_{\mathbb{P}}\mathbb{E}_{\mathbb{P}}({s}|z_{i})=0$. The projection adds the moment conditions that are key to restoring the information lost by the approximation that is relevant for identification. Conversely, adding the score  based on the approximate model as an additional moment condition can improve the efficiency of estimating $\vartheta$ using only $\mathbb{E}_{\mathbb{P}}(m|z_{i})=0$. Nevertheless, in the limit, it is the variance of the approximate score that  determines efficiency.

\begin{theorem} Asymptotic Distribution when $\Upsilon \in \mathbb{C}^{2} : \Theta \to \Phi$.  \\
	For $R(\varphi)\equiv \frac{\partial}{\partial \varphi}r(\varphi)$, $\tilde{R}\equiv \left(\frac{d\vartheta}{d\varphi}\right)'R$. Under assumptions I such that $ N^{\frac{1}{2}}\check{\bar{G}}(\tilde\psi)^{-1}\check{g}(\psi_N)\underset{d}{\to} \mathcal{Z}\sim \mathcal{N}(0,\Omega)$:
	\begin{enumerate}
		\item  $N^{\frac{1}{2}}(\hat\vartheta-\vartheta_N)\underset{d}{\to} \mathcal{Z}_r\equiv \mathcal{Z}-\check{G}(\psi^{\star}_0)\tilde{R}(\varphi^{\star}_0)(\tilde{R}(\varphi^{\star}_0)'\check{G}(\psi^{\star}_0)\tilde{R}(\varphi^{\star}_0))^{-1}\tilde{R}(\varphi^{\star}_0)'(\mathcal{Z}+\check{c})$
		\item  For any non zero vector $\xi$, $\xi'(\mathbb{V}(\mathcal{Z})-\mathbb{V}(\mathcal{Z}_r))\xi \geq 0$
	\end{enumerate}
 \label{Shrinkage1_2} 
\end{theorem} 
\begin{proof}
	See proof of Theorem \ref{Shrinkage1_2} in Appendix \hyperref[Proofs]{A} .
\end{proof} 
The main insight from Theorem \ref{Shrinkage1_2} is that under local misspecification with $\xi=1$, the efficiency gain arising from employing the tilted approximate model is the same as the efficiency gain obtained by employing the approximate model. Nevertheless, note that the localizing vector $\check{c}$ is not necessarily the same as ${c}$ in Theorem \ref{Shrinkage}, as the bias is different. This distinction will be important for the asymptotic risk of the estimator based on the tilted model and will be explored further in Theorem \ref{Bias}. 

The next set of results characterize the estimator performance in terms of Mean Squared Error and asymptotic risk. I first introduce a loss function and the corresponding definition of risk.

\subsection{Asymptotic Risk}

\begin{Definition} For any estimator $\hat{\vartheta}_{N}$ of a parameter $\vartheta_{N}$, let $l(\vartheta_{N},\hat{\vartheta}_{N})$ be the loss function that characterizes the accuracy of the estimator. The risk of the estimator is the expected loss $\mathbb{E}_{\vartheta_{N}}l(\vartheta_{N},\hat{\vartheta}_{N})$.
\end{Definition} \vspace{-0.15 in}
\subsection*{ASSUMPTIONS II}
\begin{enumerate}
	\item \textbf{Locally Quadratic Loss}: $l(\vartheta,\hat{\vartheta})\geq 0$, $l(\vartheta,\vartheta)= 0$ and $W(\vartheta):=\frac{1}{2}l_{\vartheta,\vartheta'}(\vartheta,\hat{\vartheta})\mid_{\hat{\vartheta} ={\vartheta} }$ continuous in a neighborhood of $\vartheta_{0}$
\end{enumerate}\vspace{-0.14 in}
\begin{Definition}
	Define $\rho(c,\hat{\vartheta})$	 as the asymptotic risk of the estimator sequence $\left\{\hat{\vartheta}_{N}\right\}_{N=1,2..\infty}$ for parameter sequence $\left\{\vartheta_{N}\right\}_{N=1,2..\infty}$:\[\rho(c,\hat{\vartheta}):=\lim_{\zeta\to\infty}\liminf_{N\to\infty}\mathbb{E}_{\vartheta_{N}}\min[N l(\vartheta_{N},\hat{\vartheta}_{N}),\zeta]\] where $\zeta$ is an arbitrary trimming level that is asymptotically negligible. 
\end{Definition}\vspace{-0.14 in}
The following lemma enables us to compute the asymptotic risk given a smooth loss function.\vspace{-0.14 in}
\begin{Lemma} (\cite{BHansen_2016}, Lemma 1) For an estimator $\hat{\vartheta}_{N}$ such that $N^{\frac{1}{2}}\hat{\vartheta}_{N}\underset{\vartheta_{N}}{\to}\mathbb{Z}$, if $\vartheta_{N}\underset{N\to\infty}{\to}\vartheta_{0}$, then for any loss function that satisfies Assumption II, 
 $\rho(c,\hat{\vartheta})=\mathbb{E}(\mathbb{Z}'W \mathbb{Z})$. \label{asyrisk}
\end{Lemma}\vspace{-0.14 in}
I next characterize the Mean Squared Error of the proposed estimator in a local asymptotic experiment under a stronger set of conditions, while armed with \cref{asyrisk}, we can deduce the corresponding result for asymptotic risk. 
\subsubsection{Structural Approximation}
\begin{theorem} Mean Squared Error and Asymptotic Risk when $\Upsilon \in \mathbb{C}^{2} : \Theta \to \Phi$\\ 
	For $D:=\tilde{R}'\check{G}\tilde{R}$, 
	\begin{enumerate} 
		\item 	$
			\mathbf{MSE}(\hat{\vartheta}) = \check{G}\tilde{R}'D^{-1}\tilde{R}\check{c}\check{c}'\tilde{R}'D^{-1}\tilde{R}'\check{G}'+(I-\check{G}\tilde{R}D^{-1}\tilde{R}')\Omega(I-\check{G}\tilde{R}D^{-1}\tilde{R}')'
		$
		\item If in addition to (1), $\xi'(\tilde{R}D^{-1}\tilde{R}'\check{c}\check{c}'- I)\xi \leq 0$ then 
		\begin{itemize} 
			\item For any non zero vector $\xi$ and $\mathbf{MSE}(\vartheta^{MLE})=\left(\left(\frac{d\varphi}{d\vartheta}\right)'\mathbb{E}{s}_{i}{s}_{i}'\left(\frac{d\varphi}{d\vartheta}\right)\right)^{-1}$, \[\xi'\left(\mathbf{MSE}(\vartheta^{MLE})-\mathbf{MSE}(\hat{\vartheta})\right)\xi\geq0\quad \textit{and }\rho(\check{c},\hat{\vartheta})<\rho({\vartheta}^{MLE})\]  
		\end{itemize}
	\end{enumerate}  \label{Shrinkage22}\normalsize
\end{theorem} 
\begin{proof}
	See proof of Theorem \ref{Shrinkage22} in Appendix \hyperref[Proofs]{A} .
\end{proof}

 The efficiency gains in Theorem \ref{Shrinkage1_2} can potentially translate to lower MSE if condition (2) in Theorem \ref{Shrinkage22} is met. The reason why this condition is much more likely to be met by the tilted model as compared to the approximate model is the improved approximation delivered by tilting. I next illustrate that within the class of approximated equilibrium models, the information projection alleviates, at least partially, the misspecification caused by the approximation evaluated at $\vartheta_{0}$.

\subsubsection{Improved Approximation }
\vspace{-0.05 in}
\begin{Proposition}
	\[\mathbb{E}_{\mathbb{P}_{z}}\log\left(\frac{d\mathbb{P}_{z}}{dH_{z}(\vartheta_{0})}\right)<\mathbb{E}_{\mathbb{P}_{z}}\log\left(\frac{d\mathbb{P}_{z}}{dF_{z}(\vartheta_{0})}\right)\]\label{appr}  \vspace{-0.3 in}
	\end{Proposition}\normalsize
\begin{proof}
	See Appendix \hyperref[AppA]{A} 
\end{proof}	 \vspace{-0.1 in}
\cref{appr} implies that if one obtains an approximate solution, tilting the density to satisfy the non-linear conditions results in a more accurate approximation (in \textbf{KL} units). This is yet another important property of this method, since if one cannot obtain an accurate approximation due to computational and time limits, tilting will increase accuracy. 
A caveat of the this result is that it holds at the true parameter value. Nevertheless, more can be said about the use of information projections at any $\vartheta$. 
The next result explicitly shows that adding information by tilting the model reduces the distance to the true value, for any degree of local missspecification in the non-tilted model.
\begin{theorem} 
	Consider a sequence of dgp's $\vartheta_N= \vartheta_0 + c_{0}N^{-\frac{1}{2}}$ with $\hat{\vartheta}_0\overset{p}{\to}\vartheta_0$ the restricted MLE estimator and $c_{0}$ the  localizing vector. Similarly, define $\vartheta_{1}$ as the probability limit of the tilted PMLE estimator. Then, there exists a vector $c_{1}:=N^{\frac{1}{2}}(\vartheta_N-\vartheta_1)$ such that for some $\eta>0$,   \begin{eqnarray}\mid\mid\vartheta_{1}-\vartheta_{N}\mid\mid^{2}=\mid\mid\vartheta_{0}-\vartheta_{N}\mid\mid^{2}-\frac{\eta}{N}+O\left(\frac{1}{N^{\frac{3}{2}}}\right)
	\end{eqnarray} \label{Bias}\normalsize   \vspace{-0.2 in}
\end{theorem} 
\begin{proof}
	See proof of Theorem \ref{Bias} in Appendix \hyperref[Proofs]{A} .
\end{proof}
In Theorem \ref{Shrinkage22}, the true parameter sequence is defined to be within a $\check{c}N^{-\frac{1}{2}}$ of the restricted space, which is the manifold consistent with the restrictions implied by the approximate model and the moment restrictions. In Theorem \ref{Bias}, the true parameter sequence is defined to be within a $c_{0}N^{-\frac{1}{2}}$ of the restricted space consistent with the approximate model without further information. Tilting the model transforms the restricted parameter space such that it is closer to the true parameter sequence. The next example explicitly shows how the MSE is reduced due to lower bias. 
\subsection{Analytical Example using the Stochastic Growth model} \label{SGM}
Assuming a logarithmic utility function for the representative agent and a Cobb-Douglas production function, the stochastic growth model is characterized by the following conditions:
\begin{eqnarray*}
	1&=&\mathbb{E}_{t}\left(\beta\left(\frac{C_{t}}{C_{t+1}}\right)(1-\delta + \alpha K^{1-\alpha}_{t+1})\right)\\
	K_{t+1}&=&(1-\delta)K_{t} + Z_{t}K^{\alpha}_{t} - C_{t}\\
	log(Z_{t+1})&=& \rho_{z}log(Z_{t}) + \epsilon_{z,t+1}\\
	\epsilon_{z,t+1}&\sim& N(0,\sigma^{2}_{\epsilon})
\end{eqnarray*}
where $(\alpha,\delta)$ are the capital share of output and depreciation rate, while $(\rho_{z},\sigma_{\epsilon})$ are the parameters of the exogenous productivity shock.
In order to obtain analytical results, let us assume that $\delta=1-\kappa\frac{Y_{t+1}}{K_{t+1}}$.\footnote{This assumption can be microfounded by assuming a depreciation rate that is a function of capital utilization, i.e. $\delta(u_{t})=1-\frac{\delta_{0}}{u_{t}}$. At the optimum choice of $u_{t}$, $\delta(u_{t}^{\star})=1-\alpha\frac{Y_{t}}{K_{t}}$, and hence $\kappa=\alpha$. } For a log-Normal productivity shock, it leads to the following DGP:
\small
\[
\left(\begin{array}{c}
\Delta	\widetilde{ C}_{t+1}\\
	 \tilde{Y}_{t+1}-\tilde{K}_{t+1}\\
\end{array}\right)\sim N\left(\left(\begin{array}{c}
	\tilde{\alpha}\log(\tilde{\alpha}\beta)\\
	(\tilde{\alpha}-1)\log(\tilde{\alpha}\beta)
\end{array}\right)+\underset{}{{\left( \begin{array}{cc}
			0&(\tilde{\alpha}-1)\\
			0&(\tilde{\alpha}-1)
		\end{array}\right)}}\left(\begin{array}{c}
\Delta	\widetilde{ C}_{t}\\
	\tilde{Y}_{t}-\tilde{K}_{t}
\end{array}\right),\left(\begin{array}{cc}
	\sigma^{2}_{z} & 0\\
	0& \sigma^{2}_{z}
\end{array}\right)\right)
\]
\normalsize
where I have defined $\tilde{\alpha}:=\alpha+\kappa$. What we are after here is to investigate how incorrectly imposing that $\kappa=0$ but  tilting the model to satisfy the (unconditional) Euler equation can reduce the bias. The first order conditions for $\beta$ are as follows:
\begin{eqnarray*}
	0&=&\left(1-\tilde{\alpha}\hat{\beta}^{tilt}N^{-1}\sum^{N}_{i=1}e^{-\Delta	 \widetilde{ C}_{t+1}+\tilde{Y}_{t+1}-	\tilde{K}_{t+1}}\right)w_{1} +\\
	&& N^{-1}\sum^{N}_{i=1}\left(\alpha(\Delta	\widetilde{ C}_{t+1}-2\alpha log(\alpha\hat{\beta}^{tilt})+\log \tilde{Y}_{t+1}+(1-2\alpha) \tilde{Y}_{t})+ \tilde{K}_{t+1}-\log(\alpha\hat{\beta}^{tilt})-\tilde{Y}_{t}\right)w_{2}
\end{eqnarray*}
where the first term is the Euler equation and the second term is the score function, with weights $(w_{1},w_{2})$ respectively which depend on structural parameters.\footnote{Specifically, $w_{1}$ corrresponds to $\mathbb{E}\left((M-s'm)V^{-1}\right)$ in the first order condition in \ref{eq:intro_mom} and $w_{2}=\frac{1}{\beta \sigma^{2}_{z}}$.} Taking expectations and setting $\alpha\approxeq0$ to simplify derivations yields that to first order, \[\left(\frac{\hat{\beta}^{tilt}}{\beta}-1\right)^{2}=\left(\frac{w_{2}\kappa}{w_{1}+{w_{2}}{}}\right)^{2}<\kappa^{2}=\left(\frac{\hat{\beta}^{PMLE}}{\beta}-1\right)^{2}\]
\normalsize
Hence, if $\kappa^{2}\sim \frac{\eta}{{N}}$, then
 \[\left(\frac{\hat{\beta}^{tilt}}{\beta}-1\right)^{2}=\left(\frac{\hat{\beta}^{PMLE}}{\beta}-1\right)^{2}-\frac{\eta}{{N}}\]as in Theorem \ref{Bias}.

\section{Monte Carlo Experiments}
In this section, I provide simulation evidence for the finite sample performance of this method in terms of $MSE(\hat{\vartheta})$, where I consider sample sizes relevant for macroeconomic data at quarterly frequency. I perform several experiments using (a) the Consumption Euler equation that prices  a risk free bond and (b) the Euler equation in a production economy. In the first model, I compare the estimator's
performance to other limited information methods using a correctly specified and misspecified reduced form density (non-structural approximation case), using both the conditional and the unconditional moments. In the second model, I investigate the estimator's performance when the base model is a log-linear structural approximation to the stochastic growth model. Note that for these experiments, misspecification is kept fixed. For brevity, the Monte Carlo experiment with the model used in the empirical application is presented directly in that section.
\subsection{\textbf{Estimating the Consumption Euler equation}}
For this model, I make assumptions such as the solution and the resulting DGP are analytically tractable. Assuming a logarithmic utility function for the representative agent, the consumption Euler equation is as follows:
\small \begin{eqnarray*}
	\mathbb{E}_{t}\left(\beta\frac{C_{t}R_{t+1}}{C_{t+1}}-1\right)&=&0
\end{eqnarray*}\normalsize
The resulting DGP is a Bivariate log-Normal VAR for consumption and the interest rate:\small
\[
\left(\begin{array}{c}
\log \widetilde{\Delta C}_{t+1}\\
\log \tilde{R}_{t+1}
\end{array}\right)\sim N\left(\left(\begin{array}{c}
	\log(\beta)\\
	-(1-\rho_{R})\log(\beta)
	\end{array}\right)+{\left( \begin{array}{cc}
		0&\rho_{CR}\\
     	0&\rho_{R},
		\end{array}\right)}\left(\begin{array}{c}
\log \widetilde{\Delta C}_{t}\\
\log \tilde{R}_{t}
\end{array}\right),\left(\begin{array}{cc}
\sigma^{2}_{C,me} & 0\\
0& \sigma^{2}_{R}
\end{array}\right)\right)
\]
\normalsize
I use the following parameterization: $\rho_{CR}=1,\rho_{R} = 0.95,\beta=0.85$, $\sigma^{2}_{R} = 0.5$ and I calibrate the measurement error variance to $\sigma^{2}_{C,me} = 0.05$. Below, I plot the MSE comparisons for estimating the discount factor $\beta$ with a varying sample size, $N=\left\{10..200\right\}$. Smaller samples are relevant for typical sub-sample analysis done with macroeconomic data. 

In Figure \ref{TiltNoTilt_beta}, I compare the performance of the CU-GMM and EL (Empirical Likelihood) estimators (with data generated with no measurement error) to this paper's estimator for $\hat\beta$, both in the case of estimating $(\rho_{CR},\rho_{R},\sigma^{2}_{R})$ in the correctly specified case and when - incorrectly- setting  $\rho_{CR}=\rho_{R}$.  To avoid using more information than necessary, I do not make use of the knowledge that the mean of the density depends on $\beta$.\footnote{ For the CU-GMM estimator, to construct the unconditional moments without losing any information I used optimal instruments estimated using kernel weights with bandwidth equal to $N^{-0.2}$. For the empirical likelihood estimator, I follow \cite{2004} that use a localized version of EL instead of  optimal instruments, again with bandwidth equal to $N^{-0.2}$.} 
\begin{figure}[H]
	\includegraphics[scale=0.45]{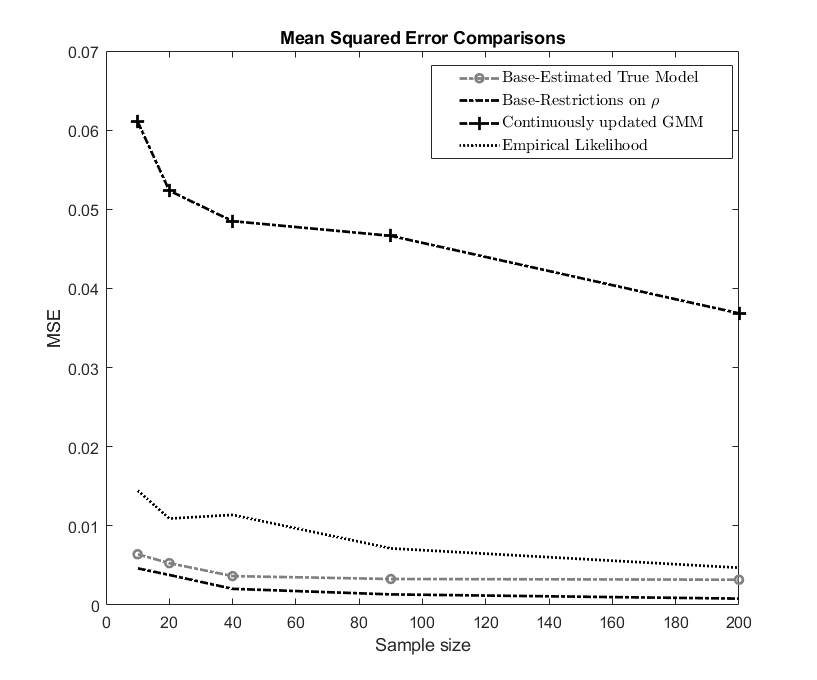}
	\vspace{-0.2 in}\protect\caption{$\hat\beta$ with different estimators (500 MC replications).}\label{TiltNoTilt_beta}
\end{figure} 
The performance of CU-GMM and EL is worse in all cases. Interpreting CU-GMM and EL as plug-in estimators using the conditional empirical CDF, where the latter is the most basic infinite dimensional model for the true conditional CDF, it is not surprising that estimating a few parameters has better MSE performance in small samples, even in the case of misspecification.\footnote{It can be shown that $\mathbb{E}_{t}M - \mathbb{E}_{f,t}M=e^{-(2-\rho_{R})log(\beta)-(1-\rho_{R})log \tilde{R}_{t}+\frac{1}{2}(\sigma^2_{R}+\sigma^2_{C})}(1-e^{\rho_{CR}\log \widetilde{\Delta C}_{t}})$. The restriction that $\rho_{CR}=\rho_{R}$ implies that $c_{1}$ is not zero.}We observe similar performance in the unconditional case (please see Appendix \hyperref[AppB]{B} for these additional Monte Carlo simulations).

In the next experiment, I extend the analysis of the estimator's performance when the base model is a log-linear approximation to the equilibrium law of motion. In section \ref{SGM}, it has already been shown that tilting reduces bias. The next numerical exercise extends the analysis of the stochastic growth model to the more general case that cannot be analytically solved. This can shed light on whether tilting actually improves on the performance of the PMLE in terms of MSE.  
\subsection{\textbf{Estimating the Stochastic Growth Model}}
Assuming a Constant Relative Risk Aversion (CRRA) utility function for the representative agent and a Cobb-Douglas production function, the stochastic growth model is characterized by the following conditions:
\begin{eqnarray*}
	1&=&\mathbb{E}_{t}\left(\beta\left(\frac{C_{t+1}}{C_{t}}\right)^{-\gamma}(1-\delta + \alpha K^{1-\alpha}_{t+1})\right)\\
	K_{t+1}&=&(1-\delta)K_{t} + Z_{t}K^{\alpha}_{t} - C_{t}\\
	log(Z_{t+1})&=& \rho_{z}log(Z_{t}) + \epsilon_{z,t+1}\\
	\epsilon_{z,t+1}&\sim& N(0,\sigma^{2}_{\epsilon})
\end{eqnarray*}
where $\gamma$ is the risk aversion coefficient and $(\alpha,\delta)$ are the capital share of output and depreciation rate, while $(\rho_{z},\sigma_{\epsilon})$ are the parameters of the exogenous productivity shock. 
The resulting DGP is a non-linear process for productivity, consumption, the capital stock, output and the real interest rate. I employ collocation methods to obtain an accurate solution using the following parameterization:
$(\alpha,\beta,\delta,\gamma,\rho_{z},\sigma_{\epsilon})=(\frac{1}{3},0.98,0.05,2,0.95,0.1)$.\footnote{I employ 100 grid points for the capital stock and 30 grid points for the productivity shock and piecewise 3rd order splines as basis functions. Collocation coefficients are computed using Newton's method.} Correspondingly, the log-linearized model is solved using QZ decomposition and the resulting density is  a Gaussian state space model, whose likelihood can be readily constructed using the Kalman filter.

Figure \ref{TiltNoTilt} plots the MSE comparisons for estimating the risk aversion coefficient $\gamma$.\footnote{I focus on the risk aversion coefficient as it is the only structural parameter that cannot be estimated from steady state values given enough observables.}   Again the CU-GMM estimator is estimated with feasible optimal instruments, while the base density is constructed using observations on consumption growth and the capital stock.  
Due to misspecification, the log-linearized model delivers a higher MSE than CU-GMM. However, tilting the likelihood to satisfy the Euler equation significantly improves performance. Given our theoretical results, the simulation evidence corroborates the claim that tilting alleviates the bias induced by the approximation to the non-linear law of motion. \subsubsection*{Computation Times} Computing the log-linear law of motion and tilting it to satisfy the non-linear moment condition takes a total of $0,017697$ seconds. The non-linear solution takes $27$ seconds.

\vspace{-0.1 in}
\begin{figure}[H]
	\includegraphics[scale=0.65]{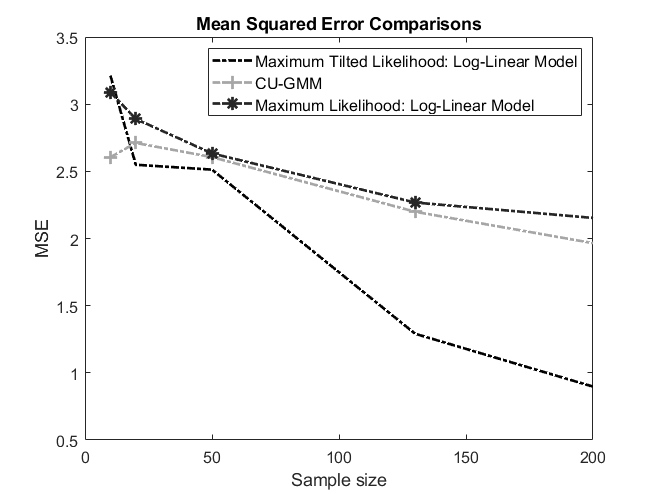}
	\vspace{-0.1 in}
	\protect\caption{$\hat\gamma$: Tilt versus No-Tilt (500 MC replications).}\label{TiltNoTilt}
\end{figure} 	\vspace{-0.15 in}

\section{Application to Pricing Macroeconomic Risk}

Low frequency fluctuations in aggregate consumption have been  shown to be important in explaining several asset pricing facts. The long run risk model of \citet{doi:10.1111/j.1540-6261.2004.00670.x} and its subsequent variations impose cross equation restrictions that link asset prices to consumption growth, where  recursive preferences \citep{
	RePEc:eee:moneco:v:24:y:1989:i:3:p:401-421,10.2307/2937817, 10.2307/1913778} differentiate between risk aversion and the intertemporal elasticity of substitution (IES). These restrictions can be summarized by the Euler equation that involves the unobserved aggregate consumption dividend $R_{\alpha,t+1}$, consumption growth $G_{c,t+1}$ and stochastic variation in the discount factor, $G_{l,t+1}$, as in \citet{doi:10.1111/jofi.12437}: 
\begin{eqnarray}
\mathbb{E}_{t}\beta^{\theta}G^{\theta}_{l,t+1}G^{\frac{-\theta}{\psi}}_{c,t+1}R^{-(1-\theta)}_{a,t+1}R_{i,t+1}&=&1  \label{lrr}
\end{eqnarray}
where $\theta = \frac{1-\gamma}{1-\frac{1}{\psi}}$,  $\gamma$ is the relative risk aversion coefficient, $\psi$ the IES and $\beta$ the discount factor. 

Relatively recent attempts to estimate this model using standard non-durable consumption data have stressed several issues that need to be taken into account. As argued by \citet{BANSAL201652}, time aggregation is an important source of bias when low frequency data is used. In fact, \citet{BANSAL201652} find empirical support for a monthly decision interval. Quarterly or yearly data are therefore  likely to be responsible for the downward bias to estimates of the IES and upward bias in risk aversion, which have been puzzling in the literature, as they imply that asset prices are increasing in uncertainty. 

A possible resolution of this issue is to use monthly data, which are nevertheless contaminated by measurement error. A recent paper \citep{doi:10.3982/ECTA14308}, SSY hereafter, provides a mixed frequency approach to make optimal use of a long span of consumption data while keeping measurement error under control. 
In this application I investigate the empirical implications of an equally important aspect of empirical macro-finance, which is the quality of the underlying approximation to the equilibrium value of the unobserved $R_{\alpha,t+1}$. 
What has been standard up to now was to use the \citet{doi:10.1111/j.1540-6261.1988.tb04598.x} log linear approximation, which has been recently criticized by \citet{JOFI:JOFI12615} as being too crude when the  underlying dynamics are persistent. I take the linear approximation as given, and impose \eqref{lrr} using the methodology in this paper. 

To isolate the informational content of imposing the Euler equation, I employ similar data and specification to \citet{doi:10.3982/ECTA14308} using only monthly data from Feb 1959 - Dec 2015 on non durable per capita consumption growth (Quantity Index for Personal Consumption Expenditures) and the risk free rate.\footnote{Please see appendix A5 of \citet{doi:10.3982/ECTA14308} on how the ex-ante real risk free rate is computed.}
 The underlying approximating model for consumption growth ($\Delta c_{t+1}=\log(G_{c,t+1})$) and the risk free rate ($r_{f,t}$) is summarized as follows:
\begin{eqnarray} 
\Delta c_{t+1}&=&\mu_{c} + x_{t} + \sigma_{c,t}\eta_{c,t+1}\\
r_{f,t}&=& B_{0} + B_{1}x_{t} + B_{1,l}x_{l,t} + B_{2,x}\sigma^{2}_{x,t} +  B_{2,c}\sigma^{2}_{c,t}\\
x_{t+1} &=& \rho x_{t} + \sqrt{1-\rho^{2}}\sigma_{x,t}\eta_{x,t+1}\\
x_{l,t+1} &=& \rho_{l}x_{t} + \sigma_{l}\eta_{l,t+1}\\
\sigma_{c,t} &=& \sigma e^{v_{x,t}}\\
\sigma_{x,t} &= &\sigma \chi_{x}e^{v_{x,t}}\\
v_{x,t+1} &= &\rho_{v_{x}}v_{x,t} + \sigma_{v,x}w_{x,t+1}
\end{eqnarray}

where $(B_{0},B_{1},B_{2,x},B_{2,c})$ are functions of the deep parameters $(\gamma,\psi,\beta)$, the risk aversion, the  elasticity of intertemporal substitution and discount factor respectively.\footnote{For details on the underlying solution mapping I encourage the reader to consult \citet{doi:10.3982/ECTA14308}.} Moreover, $x_{t}$ is the persistent component of consumption growth and the time preference shock $x_{l,t+1}:=\log(G_{l,t+1})$ is also possibly persistent. 
Regarding the observation equation, I calibrate the measurement error to the values estimated by SSY\footnote{For monthly consumption growth, I set $\sigma^2_{me}=2(\sigma^2_{\epsilon}+\sigma^2_{q})$ where $\sigma^2_{\epsilon}$ and $\sigma^2_{q}$  are the variances of the measurement error in monthly and quarterly consumption respectively, as estimated by SSY. This specification is approximately equal to SSY's specification of the measurement error in consumption growth from the third month to the first month of the next quarter, so  $\sigma^2_{me}$ is actually an upper bound to the measurement errors of the rest of the months.}. 

I perform estimation in two steps. A subset of the reduced form dynamics, that is, equations 17,19,21-23, are identified without using the long run risk model (and hence  \eqref{lrr} is irrelevant to their identification). The cash flow parameters $\phi\equiv (\rho,\chi_{x},\sigma,\rho_{v,x},\sigma_{v,c})$ are therefore estimated by using Markov Chain Monte Carlo and the particle filter.\footnote{As in {\cite{doi:10.3982/ECTA14525}}, I rely on quantiles of the posterior draws to construct confidence sets.} 
I then estimate the deep parameters $(\psi,\gamma)$ conditional on the posterior mode, $\phi^{\star}_{post}$.\footnote{Note that pre-estimating $\phi$ is consistent with the theory developed earlier as the moment condition \eqref{lrr} does not provide identifying information for $(\rho,\chi_{x},\sigma,\rho_{v,x},\sigma_{v,c})$. Moreover, since the moment condition \eqref{lrr} does not provide a measurement density for $x_{l,t}$, I would have to rely on the approximate model to identify the process parameters. I thus calibrate the time preference risk parameters $\rho_{l}$ and $\sigma_{l}$ to the posterior median estimates of SSY. A Bayesian approach has been recently proposed to deal with the lack of measurement density by \citet{GALLANT2017198}.  } This  implies that $\varphi\equiv (B_{0},B_{1},B_{2,x},B_{2,c})$, where $(B_{0},B_{1},B_{2,x},B_{2,c})$ are functions of $(\psi,\gamma)$. 
The base density used is the predictive density of the non-Gaussian state space model (Gaussian conditional on the identified volatility states) for $(\Delta c_{t+1},r_{f,t+1})$ \footnote{The are obviously alternative more efficient algorithms to deal with stochastic volatility i.e. Metropolis within Gibbs algorithm.}. Correspondingly, the conditionally Gaussian model is a bivariate Normal distribution for $(\Delta c_{t+1},r_{f,t+1})$ with conditional means $\mu_{c} + x_{t}$ and $B_{0} + B_{1}x_{t+1} + B_{2,x}\sigma^{2}_{x,t+1} +  B_{2,c}\sigma^{2}_{c,t+1}$ respectively. Conditional linearity is achieved by using the \citet{doi:10.1111/j.1540-6261.1988.tb04598.x} approximation to asset returns $r_{a,t+1}$ and solving for the price consumption ratio. 

I next investigate the usefulness of tilting the approximating density to satisfy the non linear condition \eqref{lrr}, both in terms of parameter identification and model prediction. I find that this is both empirically relevant, and economically significant.\footnote{This also corroborates the numerical results of \citet{JOFI:JOFI12615}, who have shown that this approximation can be too crude when consumption growth is persistent.}

\begin{figure}[H]
	\includegraphics[scale=0.35]{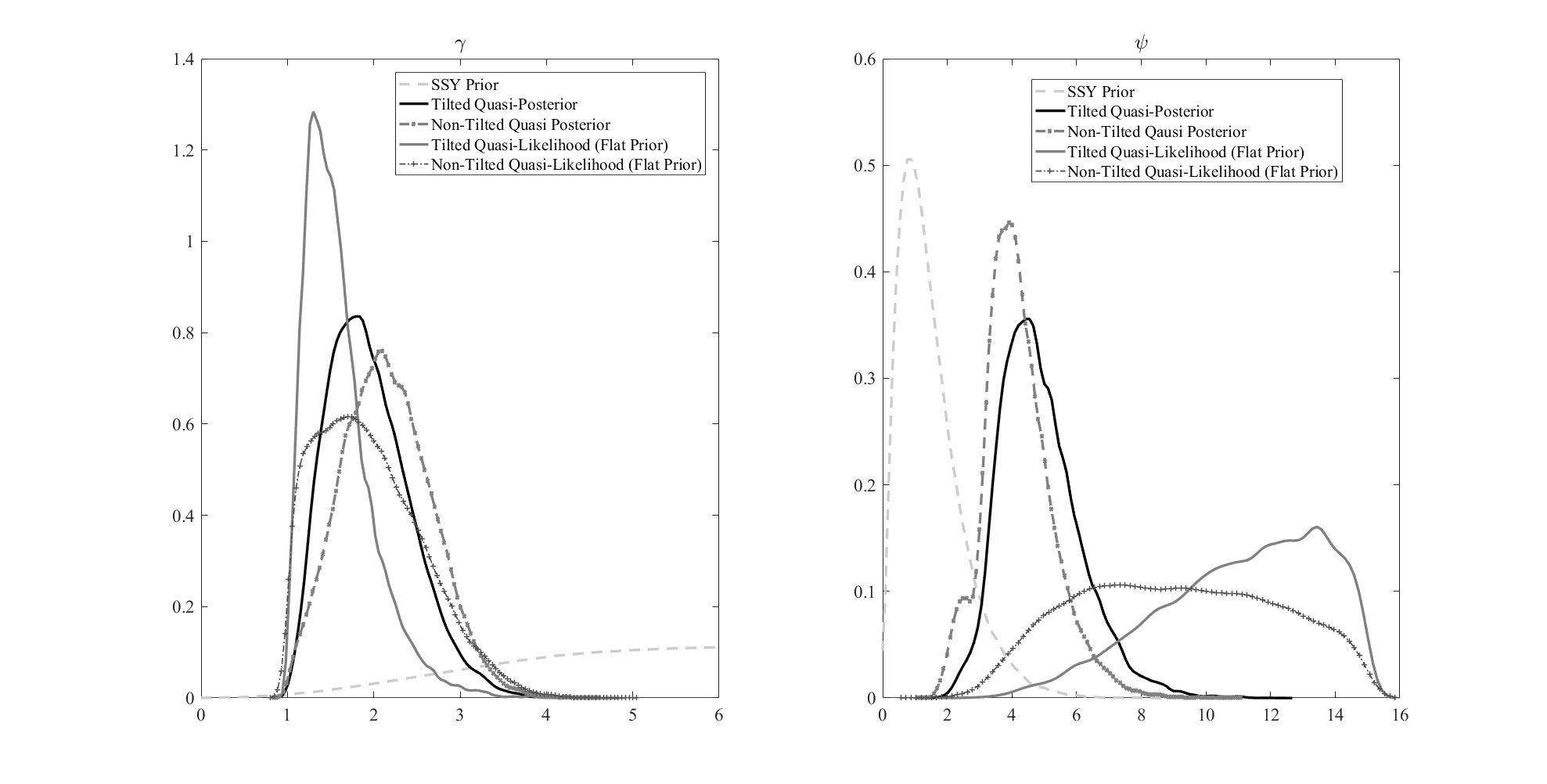}\protect\caption{ Bayesian and Frequentist estimates} \label{params} 
\end{figure}
Figure \ref{params} presents the benchmark posterior distributions together with the prior distributions, both in the log-linear approximate model (AM) and the tilted model (TM). 
As evident, correcting for the underlying non-linearity leads to improved identification as the posterior mode and MLE estimates are closer to what are considered more plausible values for risk aversion and the elasticity of intertemporal substitution. 
A direct implication is that measurement error is not the only source of upward bias in estimates for the former and downward bias in  estimates of the latter. Approximation errors is clearly another one. Moreover, posteriors are narrower.
I also report the Maximum Likelihood estimates for the tilted model, to give a sense of how much the prior information matters for both exercises.

In terms of relative fit, the tilted model is strictly preferred by the data. At their respective modes ($\psi^{\star}$), the log-posterior and log-likelihood of the tilted model are $7194.5$ and $7206.3$ respectively, while for the non-tilted approximate model they equal $7188.7$ and $7196.7$. The tilted model dominates by 5.8 log-posterior and  9.6 log-likelihood units respectively. The corresponding Bayes factor when computing the marginal likelihood using the harmonic mean ranges from 5.62 to 5.66.\footnote{Please see  \citet{Geweke99} and \citet{a5de6c49} for details on this method.}

\subsection{\textbf{Monte Carlo Simulation}}
The Monte Carlo evidence presented in Section 6 was based on more stylized models than the one employed in the application. Nevertheless, the same basic observations follow through when simulating a model similar to the one employed in this section. I next present MSE results from estimating the model, abstracting away from preference shocks while  volatility is treated as observed. 

The linear and non-linear DGP's are simulated using the code provided by  \citet{JOFI:JOFI12615}. The DGP is based in setting $(\psi=1.5, \gamma=
6, \rho=0.999,\chi_{x}=3,\sigma=0.2, \sigma_{v,c} =0.004, \sigma_{v,x}=0.01, \rho_{v,x}=\rho_{v,c}=0.99)$.  As  evident from Figure \ref{MSE_Tilting}, while increasing the sample size does eliminate a significant amount of the MSE in estimating both parameters, and in particular risk aversion, tilting yields significant reductions in the MSE at all sample sizes.

\subsubsection*{Computation Times}Finally, from the practitioner's point of view, solving the model non-linearly takes $75$ minutes, while the log-linear solution takes $0.30$ seconds. Tilting the log-linear model to satisfy the non-linear condition takes $0.008$ seconds. For estimation purposes at least, tilting is a computationally attractive method.
\begin{figure}[h]
	\includegraphics[scale=0.45]{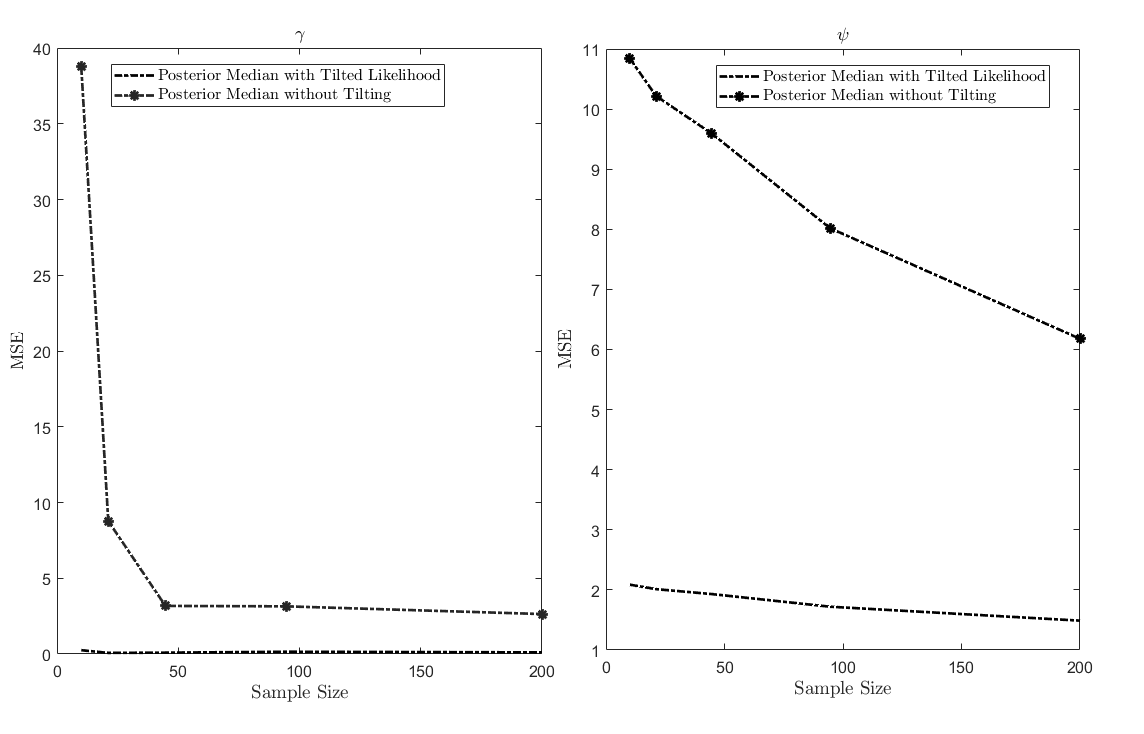}\protect\caption{Mean Squared Error comparison for estimating risk aversion and the EIS (500 MC replications).} \label{MSE_Tilting} \end{figure}

\section{Conclusion}
Model approximations are often necessary in applied settings, but they can sometimes eliminate features of the exact model that are consequential for inference, and naturally involve a bias-variance tradeoff. The paper has considered the econometric properties of tilting the approximate model to re-impose features such as conditional moment restrictions, and finds that both in theory, simulation and empirical results that tilting  substantially improves econometric performance. The  paper's application and additional monte carlo evidence also contributes to the empirical macro-finance literature by demonstrating that employing the widely used linear approximations to returns can affect parameter estimates, and hence it complements existing computational evidence   by \citet{JOFI:JOFI12615}.  The quality of the approximation is yet another reason for the downward bias to estimates of the intertemporal elasticity of substitution and the upward bias in risk aversion. 
\vspace{-0.1 in}
\bibliographystyle{econometrica}
\bibliography{thesis_references}
\newpage
\section{Appendix A} \label{AppA}

\subsection{Adversarial Estimation}
This section casts the estimation objective as a representation of a game between a modeler and an adversary see e.g. \citet{arxiv.1406.2661}. Consider a principal-agent problem where the principal delegates parameter estimation $\psi$ to a modeler, who faces a computational cost $\epsilon $ in obtaining a good approximation to the infeasible solution of the model. This cost $\epsilon$ calibrates a set of prior models:
\begin{eqnarray*}
	\mathcal{B}(p^{\star},\epsilon)&:=&\{f\in\mathcal{F}: \mathbb{E}log\left(\frac{p^{\star}}{f(\psi)}\right)<\epsilon\}
\end{eqnarray*}
The approximate model does not satisfy the moment condition $\int f(\psi)m(\psi)d\nu = 0$, which is costly to the principal. The principal acts as an adversary, and selects perturbations to $f$ that satisfy the moment conditions.  Hence, for a given strategy of the modeler $\psi$, the principal selects strategy $h$ such that it is as close as possible to the choice $f(\psi)$, up to the moment conditions being satisfied:
\begin{eqnarray*}
	\min_{h\in\mathcal{H}}\int hlog\left(\frac{h}{f(\psi)}\right)d\nu,\quad \mathcal{H}:=\big\{h\in\mathcal{L}_p:\int h m(\psi) d\nu=0,\quad \int h d\nu=1\big\}\quad ,\quad  f(\psi)\in \mathcal{B}(p^{\star},\epsilon)
\end{eqnarray*}
Given a choice of $h$, the modeler chooses $\psi$ to minimize $\int p log\left(\frac{p}{h(\psi)}\right)d\nu $ where $p$ is the empirical measure. The resulting estimator is the one that arises in the Nash equilibrium of this game. The adversary, by penalizing moment violations of the approximate model, re-introduces the information lost in $f$. As shown in the main body of the paper, this implies introducing the non-linear moment condition in the scores of $h$. In the correct approximation limit, i.e. $\epsilon\to 0$, there is no penalization by the adversary as $h=f$ and hence the cost functions of the adversary and the modeler are zero. Note that if the cost functions were symmetric, then the game would be zero-sum, which is not the case here, as the adversary pulls $f$ towards $h$ while the modeler pulls $h$ towards $p$.

\subsection{Analytical derivations for the Asset Pricing Example}
Suppressing $\lambda$, the perturbation,
$\exp(\mu'm(x,\vartheta)+\lambda)$ is proportional to 

$\exp\left(-\frac{1}{2}\left(\left(\begin{array}{c}
	c_{t+1}-\rho_{c}c_{t}\\
	R_{t+1}
\end{array}\right)^{'}\left(\begin{array}{cc}
	0 & -\mu_{t}\\
	-\mu_{t} & 0
\end{array}\right)\left(\begin{array}{c}
	c_{t+1}-\rho_{c}c_{t}\\
	R_{t+1}
\end{array}\right)\right)-\mu_{t}\frac{1}{\gamma\beta}(\gamma c_{t}-1)\right)
$

The trick here is that we can get the representation by rearranging
terms, and dropping terms that do not depend on $\mu$, and then do
the minimization. Therefore, for $\epsilon_{t+1}:=\left(\begin{array}{c}\epsilon_{1,t+1}\\ \epsilon_{2,t+1}\end{array}\right)\equiv \left(\begin{array}{c}
	c_{t+1}-\rho_{c}c_{t}\\
	R_{t+1}
\end{array}\right)$ the problem becomes as follows :\\ \small
$
\min_{\mu}\int\exp\left(-\frac{1}{2}\left(\left(\begin{array}{c}
	\epsilon_{1,t+1}\\
	\epsilon_{2,t+1}
\end{array}\right)^{'}\left(\begin{array}{cc}
	1 & -\mu_{t}\\
	-\mu_{t} & 1
\end{array}\right)\left(\begin{array}{c}
	\epsilon_{1,t+1}\\
	\epsilon_{2,t+1}
\end{array}\right)+2\mu_{t}\frac{1}{\gamma\beta}(\gamma c_{t}-1)\right)\right)d(R,C)\\
=\min_{\mu}\int \exp \left(-\frac{1}{2}\epsilon_{t+1}'\left(\begin{array}{cc}
	\frac{1}{(1-\mu_{t}^{2})} & \frac{\mu_{t}}{(1-\mu_{t}^{2})}\\
	\frac{\mu}{(1-\mu_{t}^{2})} & \frac{1}{(1-\mu_{t}^{2})}
\end{array}\right)^{-1}\epsilon_{t+1}+2\mu_{t}\frac{1}{\gamma\beta}(\gamma c_{t}-1)\right) d(R,C)
$
\normalsize

We therefore have that the F.O.C is:
\begin{eqnarray*}
	\int\exp-\frac{1}{2}\left(\epsilon_{t+1}^{'}\left(\begin{array}{cc}
		\frac{1}{(1-\mu_{t}^{2})} & \frac{\mu_{t}}{(1-\mu_{t}^{2})}\\
		\frac{\mu}{(1-\mu_{t}^{2})} & \frac{1}{(1-\mu_{t}^{2})}
	\end{array}\right)^{-1}\epsilon_{t+1}-2\mu_{t}\frac{1}{\gamma\beta}(\gamma c_{t}-1)\right)\times...\\
	...\times(-(\epsilon_{1,t+1}\epsilon_{2,t+1}+\frac{1}{\gamma\beta}(\gamma c_{t}-1))d(R,C) & = & 0
\end{eqnarray*}
Then, for the Normal scaling constant $\pi$,\\
$
\text{}\pi\int N \left(\left(\begin{array}{c}
	0\\
	0
\end{array}\right),\left(\begin{array}{cc}
	\frac{1}{(1-\mu_{t}^{2})} & \frac{\mu_{t}}{(1-\mu_{t}^{2})}\\
	\frac{\mu_{t}}{(1-\mu_{t}^{2})} & \frac{1}{(1-\mu_{t}^{2})}
\end{array}\right)\right)(\epsilon_{1,t+1}\epsilon_{2,t+1}-\frac{1}{\gamma\beta}(\gamma c_{t}-1))d(R,C) =  0$\\

which also reads as $
\frac{\mu_{t}}{(1-\mu_{t}^{2})}-\frac{1}{\gamma\beta}(\gamma c_{t}-1) =  0
$. Therefore, $\mu_{t}$ is the solution of the latter equation.

\subsection{Proofs of Main Theorems} \label{Proofs}
In the proofs to follow, I economize on notation by defining the following quantities :\\$e_{j,i}=e^{\mu_{i}'m_{j,i}(\vartheta)}$, $\tilde{e}_{j,i}=\frac{e_{j,i}}{\frac{1}{N_{s}}\sum_{j=1..s}e_{j,i}}$
, $\varkappa_{j.i}=-\frac{(e^{\mu_{i}'m_{j,i}(\vartheta)}-1)}{\mu_{i}m_{j,i}(\vartheta)'}$, $s_{j,i}:=\frac{\partial}{\partial\varphi}\log f(x_{j}|\varphi,z_{i})$ and $\mathfrak{s}_{j,i}:=\frac{s_{j,i}}{f_{j,i}}$. Index $j$ signifies evaluation at simulated point $x_{j}$ and sub-index $i$ signifies evaluation at datum $z_{i}$. \vspace{-0.2 in}
\begin{proof}  of\textbf{ \cref{mu}}. \\
	       (a) Recall that $\mu_{i}$ satisfies the moment conditions under the $H_{\varphi}(.,z_{i})$ measure, 
	        whose simulation version is $\frac{1}{N_{s}}\sum_{j=1..s}e_{j,i}m_{j,i}(\vartheta)=0$. Using the following identity, and summing over $j$ \small \[e_{j,i}m_{j,i}(\vartheta) = \frac{(e_{j,i}-1)m_{j,i}(\vartheta)m_{j,i}(\vartheta)'\mu_{i}}{m_{j,i}(\vartheta)'\mu_{i}}+m_{j,i}(\vartheta)\]\normalsize yields the following implicit map will be used to characterize the stochastic properties of $\mu_{i}$:\vspace{-0.1 in}\small \begin{eqnarray*}
	        \mu_{i}&=&\left(\frac{1}{N_{s}}{\sum_{j=1..N_{s}}\varkappa_{j,i}m_{j,i}(\vartheta)m_{j,i}(\vartheta)'}\right)^{-1}{\frac{1}{N_{s}}\sum_{j=1..s}m_{j,i}(\vartheta)}
	        \end{eqnarray*} \normalsize
	        where $\varkappa_{j,i}<0$ as $\frac{1-e^{x}}{x}<0, \forall x$.  
	       
	       First, notice that \underline{$\varkappa_{j,i}$ is bounded}: Since $\mu_{i}$ minimizes $\frac{1}{N_{s}}\sum_{j}e_{j,i}$, it must be that if $m^{l}_{j}(\vartheta)=\infty$ for some $j$ and $l$, then $\mu^{l}_{i}(\vartheta)<0$ and vice versa. Therefore,
	        $\lim_{m^{l}_{j}(\vartheta)\to \pm \infty}\varkappa_{j,i}\to 0$, where the rate of convergence to zero is $\sim e^{-m -ln(m)}$. If $|m^{l}_{j}(\vartheta)|<\infty$ then, if $\mu_{i}$ exists, $\varkappa_{j,i}$ is necessarily bounded. Finally, $\mu^{l}_{i}$ becomes unbounded if $m^{l}_{j,i}(\vartheta)$ has the same sign for any value of $\vartheta$, e.g. $m^{l}(\vartheta)$ is misspecified, which is ruled out by assumption.    	               
	               Letting $v^{ll'}_{\varkappa,j}:=[\varkappa_{j,i}m^{l}_{j,i}(\vartheta)m^{l'}_{j,i}(\vartheta)']_{ll'}$, it follows that $v^{ll'}_{\varkappa,j}<sup_{\psi}|v^{ll'}_{\varkappa,j}|$. Using \textbf{BD-1a} and Cauchy Schwarz (\textbf{CS}), we conclude that $\mathbb{E}sup_{\psi}|v_{\varkappa,j}|<\infty,\mathbb{P}({z})- a.s$. 	         
	        The denominator of $\mu_{i}$ is therefore $O_{p_{z}}(1)$. 	        
	        It is also bounded away from zero given $\textbf{PD-1}$ as $0<ln(|\varkappa_{j,i}|)\leq ln(1+e^{\mu_{i}'m_{j,i}(\vartheta)})\leq e^{\mu_{i}'m_{j,i}(\vartheta)}$ and thus $f_{i}|\varkappa_{i}|$ belongs to the set of measures for which $\textbf{PD-1}$ holds.	     
	        
	        The stochastic order of $\mu_{i}$ will therefore be determined by the numerator as follows:     \small
	        \begin{eqnarray*}
				\frac{1}{N_{s}}\sum_{j=1..s}m_{j,i}(\vartheta) & = & \int m_{i}(\vartheta)(d{F}_{N_{s},z_{i}}-d{F}_{z_{i}}+d{F}_{z_{i}}-d{\mathbb{P}}_{z_{i}}+d{\mathbb{P}}_{z_{i}})\\
					& = & o_{\mathbb{P}_{z}}(N^{-\frac{1}{2}}_{s}) + \int m_{i}(\vartheta)(d{F}_{z_{i}}-d{\mathbb{P}}_{z_{i}})
			\end{eqnarray*}			\normalsize
			where the last equality is due to the convergence of the empirical and simulation measures and correct specification of the moment condition (under $\mathbb{P}_{z}$).  				    
		    Applying Corollary \ref{disc} we have that 		    \small
		    \begin{eqnarray*}
		    	\frac{1}{N_{s}}\sum_{j=1..s}m_{j,i}(\vartheta) & = & o_{\mathbb{P}_{z}}(N^{-\frac{1}{2}}_{s}) + O_{\mathbb{P}_{z}}(\chi^{2}(F_{z_{i}},\mathbb{P}_{z_{i}})^{\frac{1}{2}}) = O_{\mathbb{P}_{z}}(\chi^{2}(F_{z_{i}},\mathbb{P}_{z_{i}})^{\frac{1}{2}})  
		    \end{eqnarray*}				
\normalsize   
and thus $\mu_{i}=O_{\mathbb{P}_{z}}(\chi^{2}(F_{z_{i}},\mathbb{P}_{z_{i}})^{\frac{1}{2}})$ (for every element of the vector  $\mu_{i}$ and all $i$).$ $\\
			
				(b) For $\bar{M}<\infty$  and $d\geq4$ (see $\mathbf{BD-1}$)
			\begin{eqnarray*}
				\mathbb{P}(\max_{i}\sup_{\vartheta}||m_{i}(\vartheta)||>\bar{M}N^{\frac{1}{d}}) & = &\mathbb{P}(\bigcup_{i\leq N}\{\sup_{\vartheta}||m_{i}(\vartheta)||>\bar{M}N^{\frac{1}{d}}\})\\
				& \leq & \sum_{i}\mathbb{P}(\sup_{\vartheta}||m_{i}(\vartheta)||>\bar{M}N^{\frac{1}{d}})\\
				& \leq & \frac{\sum_{i}\mathbb{E}(\sup_{\vartheta}||m_{i}(\vartheta)||^{d}\mathbf{1}(\sup_{\vartheta}||m_{i}(\vartheta)||^{d}>\bar{M}^{d}N)}{\bar{M}^{d}N}\\
				& = & \tilde{M}\mathbb{E}(\sup_{\vartheta}||m_{i}(\vartheta)||^{d}\mathbf{1}(\sup_{\vartheta}||m_{i}(\vartheta)||^{d}>\bar{M}^{d}N)\to 0
				\end{eqnarray*}		
		where the third bound uses the strong Markov Inequality. Therefore:
					 \begin{eqnarray*}
			 \max_{i}\sup_{\vartheta}|\mu_{i}'m_{i}(\vartheta)|&\leq& \max_{i}\sup_{\vartheta}||\mu_{i}||\max_{i}\sup_{\vartheta}||m_{i}(\vartheta)||\\
			 &\leq& n_{m}^{1/2}\max_{i}\max_{l=1..n_{m}}\sup_{\vartheta}(\mu_{l,i}) \max_{i}\sup_{\vartheta}||m_{i}(\vartheta)||= O_{p}(\max_{i}\chi^{2}(F_{z_{i}},\mathbb{P}_{z_{i}})^{\frac{1}{2}}N^{\frac{1}{d}})
			 			 \end{eqnarray*}
		 			 provided that the number of moment conditions $n_{m}$ is bounded.	
\end{proof}

\begin{proof} of \textbf{\cref{consstar}}. :
	Consider the sets $\mathcal{V}_{\mu,\delta}=\{\mu\in\mathcal{M}:||\mu-\mu_{0}||<\delta\}$and
	$\mathcal{V}_{(\vartheta,\varphi),\delta}=\{\vartheta\in\Theta:||\vartheta-\vartheta_{0}||<\delta,\varphi\in\Phi:||\varphi-\varphi_{0}||<\delta\}$
	and the objective functions they optimize respectively.
	\begin{enumerate}
\item  \underline{(Component-wise) Convergence of $\hat{\mu}_{i}$} :\\ \underline{Proofs for (a) ${\hat{\mu}_{i}-\mu_{i,0}=o_{p_{z}}(1)}$ and (b) ${Q_{N}(\psi,\hat{\mu})=Q_{N}(\psi,\mu)+o_{p_{z}}(1)}$.}\\
	(a) Given assumptions \textbf{BD-1}, using the definition of $\underset{n_{m}\times1}{\hat\mu}(\varphi,\vartheta)=\arg\inf T(z_{i},\mu)$ where $T(\mu,z_{i})=\frac{1}{N_{s}}\sum_{j=1..N_{s}}e^{\mu_{i}'m_{i}(x_{j},\vartheta)}$,  
	$\mu$
	exists for all $\vartheta,\varphi$  and is unique. For brevity, I drop dependence on $(\vartheta,\varphi)$, with the understanding that convergence holds at any $(\vartheta,\varphi)$.	Fix $Z=z_{i}$,$\forall\delta>0$.  
	
	Using a Taylor expansion of $T(\mu,z_{i})$ around $\mu_{0}$ with Lagrange remainder, we have that: \\ $T(\mu_{0},z_{i})+T'_{\mu}(\mu_{0},z_{i})(\mu-\mu_{0})+\frac{1}{2}T''_{\mu}(\tilde{\mu},z_{i})(\mu-\mu_{0})^{2}$. 
Since  $T(\mu_{0},z_{i}) \geq  T(\mu,z_{i})$, then  $\frac{1}{2}T''_{\mu}(\tilde{\mu},z_{i})(\mu-\mu_{0})^{2} + T'_{\mu}(\mu_{0},z_{i})(\mu-\mu_{0})\leq 0$, and therefore $|T'_{\mu}(\mu_{0},z_{i})|>C||\mu-\mu_{0}||$. 

We next show that $T'_{\mu}(\mu,z_{i})=o_{p_{z}}(1)$
and $\hat{\mu}_{i}-\mu_{i,0}=o_{p_{z}}(1)$.		By (\textbf{BD-1b}),  $\{e^{\mu_{i}'m_{j,i}(\vartheta)}m_{j,i}(\vartheta)\}_{j=1..N_{s}}$
	is uniformly integrable with respect to the $F-$measure , and by
	the WLLN for U.I sequences, we have that \[\frac{1}{N_{s}}\sum_{j=1..N_{s}}e^{\mu_{i}'m(x_{j},z_{i},\vartheta)+\lambda_{i}}m(x_{j},\vartheta_{0})\overset{u.p}{\to}\mathbb{E}_{h|\varphi,z_{i}}m(x_{j},\vartheta_{0},z_{i})=0\] and therefore $T'_{\mu}(\mu,z_{i})=o_{p_{z}}(1)$
  and $\hat{\mu}_{i}-\mu_{i,0}=o_{p_{z}}(1)$.   
    
    Using similar arguments, $\frac{1}{N_{s}}\sum_{1..N_{s}}e^{\mu'_{i}m(x_{j},\vartheta)}m_{i}(x_{j},\vartheta_{0})m_{i}(x_{j},\vartheta_{0})'\overset{u.p}{\rightarrow}\mathbb{E}_{H_{\varphi,z_{i}}}m_{i}(\vartheta)m_{i}(\vartheta)'$. 
  
  The above result can be strengthened. Applying the  
 Lindeberg-L\'{e}vy Central Limit Theorem, we have that $\hat{\mu}_{i}=\mu_{i,0}+o_{p}(N^{-\frac{1}{2}}_{s})$. 
 
 	(b) Defining $Q_{N}(\psi,\hat{\mu})=\frac{1}{N}\sum_{i=1..N}\log\left(f(x_{i}|z_{i},\varphi)\exp(\hat{\mu}_{i}'m(x_{i},z_{i},\vartheta))\right)$ where $\hat{\mu}$ denotes the matrix that stacks all vectors $\hat{\mu}_{i}$, we have that $Q_{N}(\psi,\hat{\mu})=Q_{N}(\psi,\mu)+o_{p_{z}}(1)$.    
 		\newpage
	\item {\underline{Uniform Convergence for $Q_{N}(\psi,\mu)$}}\\	
	By Theorem 1 in \citet{10.2307/3532442} (p. 244), we
	need to show (i) \textbf{BD} (Total Boundedness) of the metric space
	in which $(\varphi,\vartheta)$ lie together with (ii) \textbf{PC }(Pointwise
	consistency) and (iii) \textbf{SE }(Stochastic Equicontinuity).	
	Regarding (i),    Assumption \textbf{COMP} implies total
	boundedness. For (ii), using the Markov Inequality, \textbf{BD-2} and that autocovariances are summable by ergodicity:\vspace{-0.1 in} \small
	\begin{eqnarray*}
		\mathbb{P}\left(|\frac{1}{N}\sum_{i}(\log(h(x_{i};z_{i},\psi))-\mathbb{E}\log(h(x_{i};z_{i},\psi)))|>\epsilon\right)\\
		\leq  \frac{1}{N^{2}\epsilon}\mathbb{V}\left(\sum_{i}|\log(h(x_{i};z_{i},\psi))-\mathbb{E}\log(h(x_{i};z_{i},\psi))|\right)\to 0 
	\end{eqnarray*} \normalsize 
	Regarding
	(iii), Stochastic equicontinuity for the objective function can be verified
	by the ''weak'' Lipschitz condition in \citet{10.2307/3532442} (p.246):
	\[
	|{Q}_{N}(\psi,{\mu})-{Q}_{N}(\psi',{\mu})|\leq B_{N}\tilde{g}(d(\psi,\psi')),\forall(\psi,\psi')\in\Psi
	\]
	where $B_{N}=O_{p}(1)$ and $\tilde{g}$:$\lim_{y\to0}\tilde{g}(y)=0$. To verify this condition, since $Q_{N}(\psi,\mu)$ is differentiable, it suffices to use the mean value theorem:  \small
		\begin{eqnarray*}
		|{Q}_{N}(\psi,\mu)-{Q}_{N}(\psi',\mu)| &= &|\left(\nabla_{\psi} {Q}_{N}(\tilde{\psi},\mu)-\nabla_{\psi} {Q}_{N}(\tilde{\psi},'\mu)\right)'(\psi-\psi_{0})|\\
		&\leq& ||\left(\nabla_{\psi} {Q}_{N}(\tilde{\psi},\mu)-\nabla_{\psi} {Q}_{N}(\tilde{\psi},'\mu)\right)||||\psi-\psi_{0}||
	\end{eqnarray*} \normalsize
where $||\psi-\psi_{0}||$ satisfies the definition of $\tilde{g}$. 

Regarding $B_{N}:=||\left(\nabla_{\psi} {Q}_{N}(\tilde{\psi},\mu)-\nabla_{\psi} {Q}_{N}(\tilde{\psi},'\mu)\right)||$, since $\nabla_{\psi} {Q}_{N}(\tilde{\psi},\mu)$ are the first order conditions in \eqref{FOC}, it suffices to consider whether all the relevant quantities are bounded in probability. 

A sufficient condition for $B_{N}=O_{p}(1)$ is $\mathbb{E}|B_{N}|<\infty$ and thus  $\mathbb{E}||\nabla_{\psi} {Q}_{N}(\tilde{\psi},\mu)||<\infty$. 
Notice that by \eqref{FOC}, $\nabla_{\psi} {Q}_{N}(\tilde{\psi},\mu)$ is a function of the moments, the multipliers and their derivatives, $\{\lambda_{\vartheta,i},M_{i}'\mu_{i},\mu_{\vartheta,i}'m_{i},\mu_{\varphi,i}'m_{i},\mathfrak{s}_{i},\lambda_{\varphi,i}\}$. 

For $\mathbb{E}||\nabla_{\psi} {Q}_{N}(\tilde{\psi},\mu)||<\infty$  to hold, using the Cauchy-Schwarz inequality, it is sufficient that the (co)variances of $\{\lambda_{\vartheta,i},M_{i}'\mu_{i},\mu_{\vartheta,i}'m_{i},\mu_{\varphi,i}'m_{i},\mathfrak{s}_{i},\lambda_{\varphi,i}\}$ are finite. We postpone analytical derivations for the variance and covariance terms of the first order conditions to the proof of \cref{Norm}, where boundedness of the (co)variances of $\{\lambda_{\vartheta,i},M_{i}'\mu_{i},\mu_{\vartheta,i}'m_{i},\mu_{\varphi,i}'m_{i},\mathfrak{s}_{i},\lambda_{\varphi,i}\}$ follows by assumptions I. \end{enumerate}
Given that $\hat{\psi}$ is an extremum estimator, weak uniform
convergence, assumptions $\mathbf{ID}$, $\mathbf{COMP}$, and $\mathbf{BD-2}$ consistency follows by  \cite{Newey19942111}, Theorem 2.1.

\end{proof}
\begin{proof}{of \textbf{\cref{cons0}}}:
	Vanishing approximation error for $f(X|Z,\varphi)$ implies that there exists a $\varphi_0 \in \Phi: F(X|Z,\varphi_0)=\mathbb{P}(X|Z,\vartheta_0)$, for all $Z$. Since  $\sup_{\phi}\sup_{i}\Delta(F_{z_{i}},\mathbb{P}_{z_{i}})\to 0$ implies $\sup_{\phi}\sup_{i}\chi^{2}(F_{z_{i}},\mathbb{P}_{z_{i}})\to 0$,  $\lambda(Z)\to 0$ and $\mu(Z)\to0$ by \cref{mu}. Therefore $h(X|Z,\psi_{0})=f(X|Z,\varphi_{0})$. By construction, $\int m(X,Z,\vartheta^{\star}_0)\mathbb{P}(X|Z,\vartheta_0)dX=0$ as the moment condition holds under the $H$ measure. By  
	$\int m(X,Z,\vartheta_0)\mathbb{P}(X|Z,\vartheta_0)dX= 0$ and $\textbf{ID}$, we conclude that $\vartheta_0=\vartheta^{\star}_0$. 
\end{proof} 
\newpage
 \subsubsection{A comment on identification at $\psi = \psi_{0}$}       
As shown in Section 4.2.1, under asymptotic correct specification, the Jacobian is block diagonal, and each block has a well known form. For $\vartheta$, full rank of the Jacobian matrix of the moment conditions is sufficient for full rank of the corresponding block, while the same holds for the Hessian of $\varphi$.  One might be concerned with the fact that at $\psi=\psi_{0}$, the  log-likelihood does not depend on $\vartheta$ anymore as $h(X|Z)=f(X|Z,\varphi)$ and this might raise identification issues. These concerns are nevertheless not warranted as the definition of the likelihood function has a qualifying statement for $(\mu,\lambda)$ which holds for all $\psi$, that is they are solutions of problem \eqref{principal}.     
As illustrated in the proof of { \cref{mu}}, the simulated version for ${\mu}_{i}$ is characterized by the following implicit map:\vspace{-0.1 in}\small \begin{eqnarray*}
	\mu_{i}&=&\left(\frac{1}{N_{s}}{\sum_{j=1..N_{s}}\varkappa_{j,i}m_{j,i}(\vartheta)m_{j,i}(\vartheta)'}\right)^{-1}{\frac{1}{N_{s}}\sum_{j=1..N_{s}}m_{j,i}(\vartheta)}
\end{eqnarray*} \normalsize
where $\varkappa_{j.i}=\frac{1-e^{\mu_{i}'m_{j,i}(\vartheta)}}{m_{j,i}(\vartheta)'\mu_{i}}\to 1$ as $\mu_{i}\to 0$. Looking at the right hand side, it implies that under correct specification,  \small
$\left(\mathbb{E}_{\mathbb{P}}m(\vartheta,X,Z)m(\vartheta,X,Z)'\mid Z\right)^{-1}\mathbb{E}\left(m(\vartheta,X,Z)\mid Z\right) = 0       	\Leftrightarrow \mathbb{E}_{\mathbb{P}}\left(m(\vartheta,X,Z)\mid Z\right) = 0
$, \normalsize
hence, $(\vartheta,\varphi)=(\vartheta_{0},\varphi_{0})$. For any other  $(\vartheta,\varphi)$, $\mu_{i}\neq 0$, thus $\mathbb{E}_{F}\left(m(\vartheta,X,Z)\mid Z\right) \neq 0$. 
\begin{proof}{\textbf{of \cref{Norm} (Asymptotic Distribution for independent $(\varphi,\vartheta)$):}} We examine the first order expansion around the pseudotrue value where $\underset{ (n_{\vartheta}+n_{\varphi})\times 1}{g_{N}}\equiv (g_{1},g_{2})'$ is the vector of the first order conditions: $N^{\frac{1}{2}}(\psi-\psi^{\star}_{0})=-G^{-1}_{N}N^{\frac{1}{2}}g_{N}$. 	
		We first analyze the convergence in distribution of $N^{\frac{1}{2}}g_{N}$. We drop dependence of quantities on coefficients. We denote any function $q$ whose mean is computed under measure $P$ by $q_{P}$. 
Systematically applying \cref{mu} and the auxiliary Lemmata in Appendix \hyperref[AppB]{B}  to each average computed under the approximating density, we first show that only certain terms matter asymptotically at the $N^{-\frac{1}{2}}$ rate. As in  \cref{def2} and \cref{disc}, $\kappa^{-1}_{N}$ parameterizes the distance between the true and the approximating density. For  $\kappa_{N}\sim N^{\xi}$, this rate will not influence to first order $g_{1,N}$  as long as $\xi >\frac{1}{2}$, and $g_{2,N}$  as long as $\xi >1$. 

 In addition to the general notation used in the paper, I use $M_{f_{i}}$ to denote evaluation of the Jacobian of the moment restrictions under conditional measure $F_{i}$. Regarding the first term of $g_{1,N}$:\small
\begin{eqnarray*}
&&	\frac{1}{N}\underset{i}{\sum}\mu_{i,\vartheta}'{m}_{i} \equiv \frac{1}{N}\underset{i}{\sum}\left(\frac{1}{N_{s}}\sum_{j}M_{j,i}\right)'\left(\frac{1}{N_{s}}\sum_{j}e_{j,i}m_{j,i}m_{j,i}'\right)^{-1}m_{i}\\
	&=&  \frac{1}{N}\underset{i}{\sum}\left(M_{f_{i}}+O_{p_{z}}(N^{-\frac{1}{2}}_{s})\right)'\left(V^{-1}_{f_{i},m}+O_{p_{z}}(N^{-\frac{1}{2}}_{s})\right)m_{i}\\
	&=&  \frac{1}{N}\underset{i}{\sum}M_{f_{i}}'V^{-1}_{f_{i},m}m_{i}+ \frac{1}{N}\underset{i}{\sum}O_{p_{z}}(N^{-\frac{1}{2}}_{s})'V^{-1}_{f_{i},m} m_{i} \\
	&&+ \frac{1}{N}\underset{i}{\sum}M_{f_{i}}'O_{p_{z}}(N^{-\frac{1}{2}}_{s})m_{i} +  \frac{1}{N}\underset{i}{\sum}O_{p_{z}}(N^{-\frac{1}{2}}_{s})'O_{p_{z}}(N^{-\frac{1}{2}}_{s})m_{i}\\
		&=&  \frac{1}{N}\underset{i}{\sum}M_{f_{i}}'V^{-1}_{f_{i},m}m_{i}+ \frac{1}{N}\underset{i}{\sum}\left(O_{p_{z}}(N^{-\frac{1}{2}}_{s})'O_{p_{z}}(1)+O_{p_{z}}(1)'O_{p_{z}}(N^{-\frac{1}{2}}_{s})\right) m_{i} \\ &&+\frac{1}{N}\underset{i}{\sum}O_{p_{z}}(N^{-\frac{1}{2}}_{s})'O_{p_{z}}(N^{-\frac{1}{2}}_{s})m_{i}\\
			&=&  \frac{1}{N}\underset{i}{\sum}M_{f_{i}}'V^{-1}_{f_{i},m}m_{i}+ \frac{1}{N}\underset{i}{\sum}\left(O_{p_{z}}(N^{-\frac{1}{2}}_{s})\right) m_{i}+\frac{1}{N}\underset{i}{\sum}O_{p_{z}}(N^{-1}_{s})m_{i} 
	\end{eqnarray*} \normalsize Therefore, looking at the typical vector element of this term,\small
\begin{eqnarray*}
\frac{1}{N}\underset{i}{\sum}\mu_{i,\vartheta^{l}}{m}_{i} 	&\leq&  \frac{1}{N}\underset{i}{\sum} M^{l'}_{f_{i}}V^{-1}_{f_{i},m}m_{i}+  \sup_{\psi}\sup_{i}\bigg\|\left(O_{p_{z}}(N^{-\frac{1}{2}}_{s})\right) \bigg\|\frac{1}{N}\sum_{i} \|m_{i}\|+ \sup_{\psi}\sup_{i}\bigg\|\left(O_{p_{z}}(N^{-1}_{s})\right) \bigg\|\frac{1}{N}\sum_{i} \| m_{i}\|\end{eqnarray*}\normalsize
Hence, using that $N = o(N_{s})$, $N^{-\frac{1}{2}}\underset{i}{\sum}\mu_{i,\vartheta^{l}}{m}_{i}	=  N^{-\frac{1}{2}}\underset{i}{\sum} M^{l'}_{f_{i}}V^{-1}_{f_{i},m}m_{i} + o_{p}(1)$. For the typical vector element of the second term of $g_{1,N}$: 
\small
\begin{eqnarray*}
	\frac{1}{N}\underset{i}{\sum}\mu_{i}'\left({M}^{l}_{i} -\sum_{j}\tilde{e}_{i,j}M^{l}_{i,j}\right)
	&=& \frac{1}{N}\underset{i}{\sum}\mu_{i}'\left({M}^{l}_{i} -M^{l}_{h_{i}}\right)+\frac{1}{N}\underset{i}{\sum}\mu_{i}'\left(M^{l}_{h_{i}}-\sum_{j}\tilde{e}_{i,j}M^{l}_{i,j}\right)\\
		&\leq& \frac{1}{N}\underset{i}{\sum}\mu_{i}'\left({M}^{l}_{i} -M^{l}_{h_{i}}\right)+\sup_{\psi}\sup_{i}\bigg\|\mu_{i}\bigg\|\bigg\| M^{l}_{h_{i}}-\sum_{j}\tilde{e}_{i,j}M^{l}_{i,j}\bigg\|\\
	&=& \frac{1}{N}\underset{i}{\sum}\mu_{i}'\left({M}^{l}_{i} -M^{l}_{h_{i}}\right)+  \sup_{\psi}\sup_{i} O_{p_{z}}(\kappa_{i,N}^{-\frac{1}{2}})O_{p_{z}}({N_{s}}^{-\frac{1}{2}})
\end{eqnarray*}\normalsize
Therefore, even under misspecification, $\frac{1}{N}\underset{i}{\sum}\mu_{i,\vartheta^{l}}{m}_{i}=\frac{1}{N}\underset{i}{\sum} M^{l'}_{f_{i}}V^{-1}_{f_{i},m}m_{i}+\frac{1}{N}\underset{i}{\sum}\mu_{i}'\left({M}^{l}_{i} -M^{l}_{h_{i}}\right)+o_{p}(1)$.
Under asymptotic correct specification, and letting $\tilde{\mu}_{i}:={\kappa_{N}}^{\frac{1}{2}}\mu_{i}$ \small
\begin{eqnarray*}
\frac{1}{N}\underset{i}{\sum}\mu_{i}'\left({M}^{l}_{i} -M^{l}_{h_{i}}\right)&=&\frac{1}{N}\underset{i}{\sum}\mu_{i}'\left({M}^{l}_{i} -M^{l}_{\mathbb{P}_{z}}\right)+\frac{1}{N}\underset{i}{\sum}\mu_{i}'\left(M^{l}_{\mathbb{P}_{z}} - M^{l}_{h_{i}}\right)\\
&\leq&{\kappa_{N}}^{-\frac{1}{2}}\frac{1}{N}\underset{i}{\sum}\tilde{\mu}_{i}'\left({M}^{l}_{i} -M^{l}_{\mathbb{P}_{z}}\right)+ \sup_{\psi}\sup_{i}\bigg\|\mu_{i}\bigg\|\bigg\| M^{l}_{\mathbb{P}_{z}} - M^{l}_{h_{i}}   \bigg\|\\
&=&{\kappa_{N}}^{-\frac{1}{2}}O_{p}(N^{-\frac{1}{2}})+\sup_{\psi}\sup_{i} O_{p_{z}}\left({\kappa_{i,N}}^{-1}\right)
\end{eqnarray*} \normalsize where the first term in the last line uses that $\tilde{\mu}_{i}'\left({M}^{l}_{i} -M^{l}_{\mathbb{P}_{z}}\right)$ has mean zero and a bounded variance, thus its sample average converges at the root-n rate.
 Hence, \small
\begin{eqnarray*}
N^{-\frac{1}{2}}\underset{i}{\sum}\mu_{i}'\left({M}^{l}_{i} -\sum_{j}\tilde{e}_{i,j}M^{l}_{i,j}\right) 
			&=&  O_{p}(N^{-\frac{\xi}{2}})+ \sup_{\psi}\sup_{i}O_{p_{z}}(N^{\frac{1}{2}-\xi})+   \sup_{\psi}\sup_{i} O_{p_{z}}({N_{s}}^{-\frac{1}{2}}N^{\frac{1}{2}-\frac{\xi}{2}}) = O_{p}(N^{\frac{1}{2}-\xi}).
\end{eqnarray*} \normalsize  
Regarding the first order condition with respect to $\varphi$, \small
	\begin{eqnarray*}
\frac{1}{N}\underset{i}{\sum}\left(\mathfrak{s}_{i} -\sum_{j}\tilde{e}_{i,j}\mathfrak{s}_{j,i}+\mu_{i,\varphi}'m_{i}\right) 
		&=& \frac{1}{N}\underset{i}{\sum}\left(\mathfrak{s}_{i} -\mathfrak{s}_{h_{i}}+\mathfrak{s}_{h_{i}}-\sum_{j}\tilde{e}_{i,j}\mathfrak{s}_{j,i}+\mu_{i,\varphi}'m_{i}\right)\\
		&\leq& \frac{1}{N}\underset{i}{\sum}\left(\mathfrak{s}_{i} -\mathfrak{s}_{h_{i}}\right)+\sup_{\psi}\sup_{i}\left(\mathfrak{s}_{h_{i}}-\sum_{j}\tilde{e}_{i,j}\mathfrak{s}_{j,i}\right)    +\frac{1}{N}\underset{i}{\sum}\mu_{i,\varphi}'m_{i}\\
	    	&=& \frac{1}{N}\underset{i}{\sum}\left(\mathfrak{s}_{i} -\mathfrak{s}_{h_{i}}\right)+\sup_{\psi}\sup_{i} O_{p_{z}}(N_{s}^{-\frac{1}{2}})+\frac{1}{N}\underset{i}{\sum}\mu_{i,\varphi}'m_{i}	 \\
	    		&=& \frac{1}{N}\underset{i}{\sum}\left(s_{i}+(\mathfrak{s}_{i}-\mathfrak{s}_{\mathbb{P}_{i}})-(s_{i}-{s}_{\mathbb{P}_{i}})+(\mathfrak{s}_{\mathbb{P}_{i}}-{s}_{\mathbb{P}_{i}})-\mathfrak{s}_{h_{i}}\right)+o_{p}(1)+\frac{1}{N}\underset{i}{\sum}\mu_{i,\varphi}'m_{i}	     	
		\end{eqnarray*}		
	For $\xi>\frac{1}{2}$, $\mathfrak{s}_{h_{i}}=\mathfrak{s}_{f_{i}} + O_{p_{z}}(N^{-\frac{\xi}{2}}) $ while for $\xi>1$, $\mathfrak{s}_{\mathbb{P}_{i}}-{s}_{\mathbb{P}_{i}}=O_{p_{z}}(N^{-\frac{\xi}{2}})$,  $\mathfrak{s}_{f_{i}}-{s}_{f_{i}}=O_{p_{z}}(N^{-\frac{\xi}{2}})$, ${s}_{f_{i}}-{s}_{\mathbb{P}_{i}}=O_{p_{z}}(N^{-\frac{\xi}{2}})$.  Hence, for $\xi>1$
			\begin{eqnarray*}N^{-\frac{1}{2}}\underset{i}{\sum}\left(\mathfrak{s}_{i} -\sum_{j}\tilde{e}_{i,j}\mathfrak{s}_{j,i}+\mu_{i,\varphi}'m_{i}\right) &=&  N^{-\frac{1}{2}}\underset{i}{\sum}\left({s}_{i} -{s}_{f_{i}}\right)+N^{-\frac{1}{2}}\underset{i}{\sum}\mu_{i,\varphi}'m_{i}+O_{p}(N^{\frac{1}{2}-\frac{\xi}{2}})\\
\end{eqnarray*}	\normalsize
			A key driver of the results is \textbf{BD-1a}, as conditional moments are bounded for all $z\in Z$, and are therefore  bounded random variables. 
	
	Multiply the first order conditions by $N^{\frac{1}{2}}$, $N^{\frac{1}{2}}g_{N}=N^{\frac{1}{2}}A_{i,0}+o_{p}(1)$
	where the terms in $A_{i,0}$ are those terms in the above derivations that converge at this rate:	
	Then,
	\small		
	\begin{eqnarray}
		{N^{\frac{1}{2}}g_{N}}=\left[\begin{array}{c}
			N^{-\frac{1}{2}}\sum_{i} M_{f_{i}}'V^{-1}_{f_{i},m}m_{i}+  	N^{-\frac{1}{2}}\underset{i}{\sum}\left({M}_{i} -M_{h_{i}}\right)'\mu_{i}\\    N^{-\frac{1}{2}}\underset{i}{\sum}\left(\mathfrak{s}_{i} -\mathfrak{s}_{h_{i}}+\mu_{i,\varphi}'m_{i}\right)
		\end{array}\right]  +o_{p}(1)   \label{eq:local}
	\end{eqnarray}	\normalsize
To show asymptotic normality, it is sufficient to use \eqref{eq:local}. I make use of the Cramer-Wold device. Let $\xi$ be
	a $p\times 1$ vector of real numbers where $\xi_{p\times1}'=\left(\begin{array}{cc}
	\underset{\dim(\vartheta)}{\xi'_{1}}, & \underset{\dim(\varphi)}{\xi_{2}'}\end{array}\right)$ normalized such that $||\xi||=1$. 
	\begin{eqnarray*}
	N^{\frac{1}{2}}\xi_{p\times1}'g_{N} & = & N^{-\frac{1}{2}}\underset{i}{\sum}\xi_{1}'\left(M_{f_{i}}'V^{-1}_{f_{i},m}m_{i}+ \left({M}_{i} -M_{h_{i}}\right)'\mu_{i}\right)+N^{-\frac{1}{2}}\underset{i}{\sum}\xi_{2}'\left(\mathfrak{s}_{i} -\mathfrak{s}_{f_{i}}+\mu_{i,\varphi}'m_{i}\right)+o_{p}(1)\\
	& = & \hat{\Xi}_{1}+\hat{\Xi}_{2}+ o(1)
\end{eqnarray*}	
	What needs to be shown is that the variance of $\hat{\Xi}_{1}$ and $\hat{\Xi}_{2}$ is
	finite. Covariances of the	above terms can be bounded by their variances using
	\textbf{C-S} inequality.			As for $\hat{\Xi}_{1}$, the variance of the first term is  \[\mathbb{E}\mathbb{V}_{z}(\xi'M_{f_{i}}'V^{-1}_{f_{i},m}m_{i})\leq\xi_{1}'\mathbb{E}(M_{f_{i}}'V^{-1}_{f_{i},m}\mathbb{V}_{m}V^{-1}_{f_{i},m}M_{f_{i}})\xi_{1}<\infty\] 
	while for the second term, looking at one individual vector component (indexed by $l$), $\mathbb{E}\mathbb{V}_{z}(\xi_{1}'(M^{l}_{{i}}-M^{l}_{h_{i}})'\mu_{i})=\xi_{1}'\mathbb{E}\mathbb{V}_{z}(\mu_{i}'M^{l}_{{i}})\xi_{1}=\xi_{1}'\mathbb{E}(\mu_{i}'\mathbb{V}_{M^{l}}\mu_{i})\xi_{1}<\infty$ as all conditional expectations are bounded almost surely. 
	Similar argument is followed for $\hat{\Xi}_{2}$.	Combining the above results, $g_{N}$ has conditional mean zero (w.r.t $z$) and finite variance.  Using the CLT for Martingale Differences \citep{1961}\vspace{-0.1 in}
	\begin{eqnarray*}
		N^{\frac{1}{2}}\xi_{p\times1}'g_{N} & = & N^{-\frac{1}{2}}\xi_{p\times1}'\Xi_{N}+o_{p}(1) \to \mathcal{N}(0,\xi' V_{g}\xi)
	\end{eqnarray*}
	and therefore  $
	N^{\frac{1}{2}}(g_{N})\to N(0,V_g)$. This result can be extended to allow for serial correlation in $g_{N}$ using \citet{gor69}'s CLT, where  $\mathbb{V}_{g}:=\Gamma_{0}+ \sum_{j=1..\infty}(\Gamma_{j} + \Gamma'_{j})$. Finiteness of  $\Gamma_{j}\equiv \mathbb{E}(g_{t}g_{t-j}')$ follows very similar arguments as above. This concludes the proof of Asymptotic Normality.
		\subsection{\textbf{Efficiency}} From the first order conditions, $G_{N}(\hat{\vartheta},\hat{\varphi})=0$, using the mean value theorem,	\[
	\begin{array}{cc}
	0= & g_{N}(\psi_{0})+G_{N}(\tilde{\psi})(\psi-\psi_{0})\end{array}
	\]	
	Using \cref{muder} I next investigate the exact form of the non random limits of both the Jacobian term and the variance covariance matrix of the moment conditions. By the WLLN, averages converge pointwise to a constant. Furthermore, given that all of these quantities
	are functions of $m(x,z)$, $M(x,z)$ using measure $F$
	or $P$, we can obtain dominating functions by taking the supremum
	over $\Psi$. Then, By assumption $\mathbf{BD1-a}$ they are bounded. Uniform convergence follows.
		\subsubsection{\textbf{Form of Jacobian $G_{N}(\psi)$:}} The  population Jacobian matrix is the following: \small
	\begin{eqnarray*}
		\bar{G}_{N}(\tilde{\psi}) & \equiv & \left(\begin{array}{cc}
			\bar{G}_{i,\vartheta\vartheta'}(\tilde{\psi})  & \bar{G}_{i,\vartheta\varphi'}(\tilde{\psi}) \\
			\bar{G}_{i,\varphi\vartheta'}(\tilde{\psi})  & \bar{G}_{i,\varphi\varphi'}(\tilde{\psi}) 
		\end{array}\right)\\
		{G}_{i,\vartheta_{l}\vartheta'} & = & m_{i}'\mu_{i,\vartheta_{l}\vartheta'}+\left(M^{l}_{i}-\frac{1}{N_{s}}\sum_{j}^{N_{s}}\tilde{e}_{j}M^{l}_{j}\right)'\mu_{i,\vartheta'}+\mu_{i,\vartheta_{l}}'M_{i}\\
		&  & +\mu_{i}'\left(\frac{\partial M^{l}_{i}}{\partial \vartheta'}-\frac{1}{N_{s}}\sum_{j}^{N_{s}}\tilde{e}_{j}\frac{\partial M^{l}}{\partial \vartheta'}\right)-\mu_{i}'\frac{1}{N_{s}}\sum_{j}^{N_{s}}\tilde{e}_{j,\vartheta'}M_{j}\\
	{G}_{i,\vartheta_{l}\varphi'} & = &\left(M^{l}_{i}-\frac{1}{N_{s}}\sum_{j}^{N_{s}}\tilde{e}_{j}M^{l}_{j}\right)'\mu_{i,\varphi}+m_{i}'\mu_{i,\vartheta_{l}\varphi'}-\mu_{i}'\left(\frac{1}{N_{s}}\sum_{j}^{N_{s}}\tilde{e}_{j}M^{l}_{j}\mathfrak{s}_{j}'+\frac{1}{N_{s}}\sum_{j}^{N_{s}}M^{l}_{j}\tilde{e}_{j,\varphi'}\right)\\
		&  & \\
		{G}_{i,\varphi_{l}\varphi'} &=& \frac{\partial^2 log(f_{i})}{\partial \varphi_l \varphi'}-\frac{1}{N_S}\sum_j \tilde{e}_j  \frac{\partial^2 log(f_{j,i})}{\partial \varphi_l \varphi'} -  \frac{1}{N_s}\sum_j\tilde{e}_j\mathfrak{s}_{l,j,i}\mathfrak{s}'_{j}  + m_i'\mu_{i,\varphi_l\varphi'}   -\frac{1}{N_s}\sum_j\mathfrak{s}_{l,j,i}\tilde{e}_{j,i,\varphi'}   
	\end{eqnarray*}\normalsize where superscript $l$ denotes the $l_{th}$ column and $\tilde{e}_{j,i,\varphi'}=\tilde{e}_{j,i}m_{j,i}'\mu_{\varphi}$.

The Jacobian converges to,\small
\begin{eqnarray*}
	\bar{G}_{i,\vartheta^{\star}_{l}\vartheta^{\star'}} & \underset{p}{\to} & \mathbb{E}\left(\mathbb{E}(m'|z_{i})\mu_{i,\vartheta^{\star'}_{l}}\right) -\mathbb{E}\left(\left(\mathbb{E}(M^{l}|z_{i})-\mathbb{E}_{h}(M^{l}|z_{i})'\right)V^{-1}_{f_{i},m}\mu_{i,\vartheta^{\star'}}\right)\\
	&  & +\mathbb{E}\left(\mu_{i,\vartheta^{\star}_{l}}'\mathbb{E}(M|z_{i})\right)+\mathbb{E}\left(\mu_{i}'\mathbb{E}\left(\frac{\partial M^{l}}{\partial \vartheta'}|z_{i}\right)-\mathbb{E}_{h}\left(\frac{\partial M^{l}}{\partial \vartheta^{\star'}}|z_{i}\right)-\mu_{i}'\mathbb{E}_{h}(M|z_{i})\right)\\
	\bar{G}_{i,\vartheta^{\star'}_{l}\varphi'} &  \underset{p}{\to} &\mathbb{E}\left(\left(\mathbb{E}(M^{l}|z_{i})-\mathbb{E}_{f}(M^{l}|z_{i}))'\right)\mu_{i,\varphi^{\star}}\right)+\mathbb{E}\left(\mathbb{E}(m'|z_{i})\mu_{i,\vartheta^{\star}_{l}\varphi^{\star'}}\right)\\
	&  & -\mathbb{E}\left(\mu_{i}'(\mathbb{E}_{h}(M^{l}\mathfrak{s}'+M^{l}m'\mu_{i,\phi}|z_{i}))\right) \\
	\bar{G}_{i,\varphi^{\star}_{l}\varphi^{\star'}} & \underset{p}{\to}& \mathbb{E}\left(\frac{\partial^2 log(f_{i})}{\partial \varphi^{\star}_l \varphi^{\star'}}\right)-\mathbb{E}\mathbb{E}_{h}\left(\frac{\partial^2 log(f_{j})}{\partial \varphi^{\star}_l \varphi^{\star'}}|z_{i}\right) -  \mathbb{E}\mathbb{E}_{h}(\mathfrak{s}_{l}\mathfrak{s}'|z_{i})  + \mathbb{E}(\mathbb{E}(m|z_{i})\mu_{i,\varphi^{\star}_l\varphi^{\star'}})  \\
	&& +\mathbb{E}(\mathbb{E}_{h}(\mathfrak{s}_{l}m'|z_{i})\mu_{\varphi^{\star}})  
\end{eqnarray*}
\normalsize
which under asymptotic correct specification, it yields:\small
	\begin{eqnarray*}
	\bar{G}_{i,\vartheta_{l}\vartheta'}        & \underset{p}{\to} & - \mathbb{E}\left(\mathbb{E}(M^{l'}|z_{i}) \mathbb{V}_{m,i} ^{-1}\mathbb{E}(M|z_{i})\right)\\
	\bar{G}_{i,\vartheta_{l}\varphi'} & \underset{p}{\to} & 0\\
	\bar{G}_{i,\varphi_{l}\varphi'} & \underset{p}{\to} & \mathbb{E}\frac{\partial^2 log(f_{i})}{\partial \varphi_l \varphi'} + \mathbb{E} \left(\mathbb{E}(  {s}_{l}m'|z_{i})\mathbb{V}^{-1}_{m,i}\mathbb{E}(m{s}'|z_{i})\right)   \end{eqnarray*} \normalsize \vspace{-0.2 in}
	\subsubsection{\textbf{Form of $V_{g}$}} 
The variance covariance matrix is pinned down by\small
	\begin{eqnarray*}
		\mathbb{V}(N^{\frac{1}{2}}\hat{g}_{1}(\vartheta^{\star})) & = & \mathbb{V}(M_{i}'\mu_{i}+\mu_{i,\vartheta^{\star}}'m_{i}+\lambda_{i,\vartheta^{\star}})\\& =& \mathbb{E}\left[\left( M_{f_{i}}'V^{-1}_{f_{i},m}m_{i} + \left({M}_{i} -M_{h_{i}}\right)'\mu_{i}\right)\left( M_{f_{i}}'V^{-1}_{f_{i},m}m_{i}   + \left({M}_{i} -M_{h_{i}}\right)'\mu_{i}\right)'\right]+o(1)\\
		\mathbb{V}(N^{\frac{1}{2}}\hat{g}_{2}(\vartheta^{\star})) & = & \mathbb{V}(\mathfrak{s}_{i}'+\mu_{i,\varphi}'m_{i}+ \lambda_{i,\varphi})\\
		&=& \mathbb{E}\left[\left(\mathfrak{s}_{i} -\mathfrak{s}_{h_{i}} -\mathbb{E}_{f}(\mathfrak{s}m'|z_{i})V^{-1}_{f_{z},m}\right)\left(\mathfrak{s}_{i} -\mathfrak{s}_{h_{i}} -\mathbb{E}_{f}(\mathfrak{s}m'|z_{i})V^{-1}_{f_{z},m}\right)'\right]+o(1)\\
		\mathbb{C}ov(N^{\frac{1}{2}}\hat{g}_{1}(\vartheta^{\star}),N^{\frac{1}{2}}\hat{g}_{2}(\vartheta^{\star}))&=&\mathbb{C}ov\left(M_{i}'\mu_{i}+\mu_{i,\vartheta^{\star}}'m_{i}+\lambda_{i,\vartheta^{\star}},\mathfrak{s}_{i}'+\mu_{i,\varphi^{\star}}'m_{i}+ \lambda_{i,\varphi^{\star}}\right)\\
		&=& \mathbb{E}\left[\left(  M_{f_{i}}'V^{-1}_{f_{i},m}m_{i} + \left({M}_{i} -M_{h_{i}}\right)'\mu_{i}    \right)\left(\mathfrak{s}_{i} -\mathfrak{s}_{h_{i}} -\mathbb{E}_{f}(\mathfrak{s}m'|z_{i})V^{-1}_{f_{z},m}\right)'\right]+o(1)\\	
	\end{eqnarray*}	 \normalsize
When $\xi>1$ this yields:
		\begin{eqnarray*}
		\mathbb{V}(N^{\frac{1}{2}}\hat{g}_{1}(\vartheta)) & = & \mathbb{E} M_{f_{i}}'V^{-1}_{f_{i},m}m_{i}m_{i}'V^{-1}_{f_{i},m} M_{f_{i}}+o(1)=\mathbb{E} \left(\mathbb{E}(M_{i}|z_{i})'V^{-1}_{\mathbb{P}_{z},m} \mathbb{E}(M_{i}|z_{i})\right)\\
		&&+o(1)\\
			\mathbb{V}(N^{\frac{1}{2}}\hat{g}_{2}(\vartheta)) & = & \mathbb{E}\mathfrak{s}_{i}\mathfrak{s}_{i}'+ \mathbb{E}\mathfrak{s}_{i}m_{i}'\mu_{i,\varphi}+ \mathbb{E}\mu_{i,\varphi}'m_{i}\mathfrak{s}_{i}'+ \mathbb{E}\mu_{i,\varphi}'m_{i}m_{i}'\mu_{i,\varphi}+o(1)\\
			&=& \mathbb{E}{s}_{i}{s}_{i}'-\mathbb{E}\left(\mathbb{E}({s}m'|z_{i})V^{-1}_{\mathbb{P}_{z},m}\mathbb{E}(m{s}'|z_{i})\right)+o(1)\\
			\mathbb{C}ov(N^{\frac{1}{2}}\hat{g}_{1}(\vartheta),N^{\frac{1}{2}}\hat{g}_{2}(\vartheta))&=&\mathbb{E}\left(M_{f_{i}}'V^{-1}_{f_{i},m}m_{i}\right)\left(\mathfrak{s}_{i} -\mathfrak{s}_{f_{i}}+\mu_{i,\varphi}'m_{i}\right)'+o(1)\\
			&=&\mathbb{E}M_{f_{i}}'V^{-1}_{f_{i},m}m_{i}{s}_{i}' + \mathbb{E}M_{f_{i}}'V^{-1}_{f_{i},m}m_{i}m_{i}'\mu_{i,\varphi} +o(1) = o(1) 
		\end{eqnarray*}	
    
           \end{proof}      
      
              \begin{proof}{\textbf{of \cref{Norm2} (Asymptotic Distribution for  $\vartheta$ when ${\Upsilon}:\Theta\to\Phi$):}}
       The proof follows very similar steps as \cref{Norm} and is therefore omitted. The main difference is that in this case the first order condition, the Jacobian and Variance under correct specification $(\xi>1)$ are:\small
       	\begin{eqnarray*}
       	{\check{g}_{N}}&:=& N^{-1}\sum_{i} M_{f_{i}}'V^{-1}_{f_{i},m}m_{i} + \left(\frac{d\varphi}{d \vartheta}\right)'N^{-1} \underset{i}{\sum}\left(\mathfrak{s}_{i} -\mathfrak{s}_{f_{i}}+\mu_{i,\varphi}'m_{i}\right) +o_{p}(1)\\
        		{\check{\bar{G}}}_{i,\vartheta_{l}\vartheta'} &=& 	{{\bar{G}}}_{i,\vartheta_{l}\vartheta'} +   \bar{g}_{2}'\frac{d}{d\vartheta'}\left(\frac{d\varphi}{d\vartheta_{l}}\right) + \left(\frac{d\varphi}{d\vartheta_{l}}\right)'\left( {\bar{G}}_{i,\varphi\varphi'}\left(\frac{d\varphi}{d\vartheta}\right) +  {\bar{G}}_{i,\varphi\vartheta'} \right)\\
   		 &\underset{p}{\to} &  - \mathbb{E}\left(\mathbb{E}(M^{l'}|z_{i}) \mathbb{V}_{m,i} ^{-1}\mathbb{E}(M|z_{i})\right) + \left(\frac{d\varphi}{d\vartheta_{l}}\right)'\mathbb{E}\frac{\partial^2 log(f_{i})}{\partial \varphi \varphi'}\left(\frac{d\varphi}{d\vartheta}\right)  \\
   		 &&+\left(\frac{d\varphi}{d\vartheta_{l}}\right)'\mathbb{E} \left(\mathbb{E}(  {s}m'|z_{i})\mathbb{V}^{-1}_{m,i}\mathbb{E}(m{s}'|z_{i})\right)\left(\frac{d\varphi}{d\vartheta}\right)\\
   		 &=& \left(\frac{d\varphi}{d\vartheta_{l}}\right)'\mathbb{E}\frac{\partial^2 log(f_{i})}{\partial \varphi \varphi'}\left(\frac{d\varphi}{d\vartheta}\right) = -\left(\frac{d\varphi}{d\vartheta_{l}}\right)'\mathbb{E}({s}{s}')\left(\frac{d\varphi}{d\vartheta}\right)\\
       		\mathbb{V}(N^{\frac{1}{2}}\check{g}_{2}(\vartheta))&=& 		\mathbb{V}(N^{\frac{1}{2}}\hat{g}_{1}(\vartheta)) + 	\mathbb{V}(N^{\frac{1}{2}}\hat{g}_{2}(\vartheta))\\
       		&& + 	\mathbb{C}ov(N^{\frac{1}{2}}\hat{g}_{1}(\vartheta),N^{\frac{1}{2}}\hat{g}_{2}(\vartheta)) + 	\mathbb{C}ov(N^{\frac{1}{2}}\hat{g}_{2}(\vartheta),N^{\frac{1}{2}}\hat{g}_{1}(\vartheta))\\
       		&=& \mathbb{E} \left(\mathbb{E}(M|z_{i})'V^{-1}_{\mathbb{P}_{z},m} \mathbb{E}(M|z_{i})\right)\\
       		&&+ \left(\frac{d\varphi}{d\vartheta}\right)'\left(\mathbb{E}{s}_{i}{s}_{i}'-\mathbb{E}\left(\mathbb{E}({s}m'|z_{i})V^{-1}_{\mathbb{P}_{z},m}\mathbb{E}(m{s}'|z_{i})\right)\right)\left(\frac{d\varphi}{d\vartheta_{l}}\right) + o(1)\\
       		&=&\left(\frac{d\varphi}{d\vartheta}\right)'\mathbb{E}{s}_{i}{s}_{i}'\left(\frac{d\varphi}{d\vartheta_{l}}\right) + o(1)       		
       	\end{eqnarray*}	
       	\normalsize
       	where we use that $\mathbb{E}(M|z_{i})=-\mathbb{E}(m{s}'|z_{i})\left(\frac{d\varphi}{d\vartheta}\right) $ and $\mathbb{E}(\frac{\partial^2 log(f_{i})}{\partial \varphi \varphi'}|z_{i})=-\mathbb{E}({s}{s}'|z_{i})\left(\frac{d\varphi}{d\vartheta}\right) $  (See  e.g. \citet{10.2307/2096601}). 	Hence, the asymptotic variance of the estimator collapses to 
       	$\left[\left(\frac{d\varphi}{d\vartheta}\right)'\mathbb{E}{s}_{i}{s}_{i}'\left(\frac{d\varphi}{d\vartheta}\right)\right]^{-1} $.
       \end{proof}
       \vspace{-0.2 in}
\begin{proof} \textbf{of Theorem \ref{Shrinkage}}\\
		\begin{enumerate}
	\item The first order conditions for $\varphi$ under restrictions $r(\varphi)=0$ are as follows:	
	\begin{eqnarray*}
		\hat\varphi-\varphi_N&=& - \bar{G}^{21}(\tilde\psi){g}_1(\psi_N)-\bar{G}^{22}(\tilde\psi)({g}_2(\psi_N)+{\pi} R(\hat{\varphi}))
	\end{eqnarray*} \vspace{-0.05 in}
     Expanding the constraint around $\varphi^\star_0$, $0=r(\hat\varphi)={{r(\varphi^\star_0)}}+R(\tilde{\tilde{\varphi}})'(\hat\varphi-\varphi^\star_0)$   and substituting for $\hat\varphi-\varphi_N$ and $\varphi_{N}-\varphi^\star_0$ , 	
	\begin{eqnarray*}
		\pi &=& -(R'\bar{G}^{22}R)^{-1}R'(\bar{G}^{21}g_1+\bar{G}^{22}{g}_2-c_{2} N^{-\frac{1}{2}}) 
	\end{eqnarray*} where I dropped dependence on $N$ and $\psi$ and used that $r(\varphi^\star_0)=0$.	Substituting for $\pi$ in $\hat\varphi-\varphi_N$ and plugging $\pi$ in the first order conditions for $\hat\vartheta - \vartheta_N$ and $\hat\varphi - \varphi_N$ the result follows. Note that the limiting Jacobian is identical as long as $\xi\geq \frac{1}{2}$. As  shown in \cref{Norm}, $\bar{G}^{21}\underset{p}{\to}0$ and hence the asymptotic distribution for $\hat\vartheta - \vartheta_N$ is orthogonal to that of $\hat\varphi - \varphi_N$ .\\
\item Let $D:=R'\bar{G}^{22}R$. Evaluating the difference between the variances,\small \begin{eqnarray*}
	\mathbb{V}(\mathcal{Z}_r)-\mathbb{V}(S_{2}\mathcal{Z}) &=& (I-\bar{G}^{22}RD^{-1}R')\bar{G}^{22}V_{22}\bar{G}^{22}(I-\bar{G}^{22}RD^{-1}R')'-V_{22}	\\
	&=& (V_{22}-V_{22}RD^{-1}R'V_{22}) - (V_{22}-V_{22}RD^{-1}R'V_{22})(V_{22}RD^{-1}R')'-V_{22}	\\
	&=& -V_{22}RD^{-1}R'V_{22} - V_{22} (V_{22}RD^{-1}R')' + V_{22}RD^{-1}R'V_{22}RD^{-1}R'V_{22}	\\
	&=& -V_{22}RD^{-1}R'V_{22} - V_{22} (V_{22}RD^{-1}R')' + V_{22}RD^{-1}R'V_{22} = - V_{22} RD^{-1}R'V_{22}	\end{eqnarray*}\normalsize
where in the second line, I used that $\bar{G}^{22} \to G_{22}= -V_{22}$. The result follows by the positive definiteness of $ V_{22} RD^{-1}R'V_{22}$.
\end{enumerate}
	\end{proof}
\begin{proof} \textbf{of Theorem \ref{Shrinkage1_2}}
	\begin{enumerate}
		\item 	
	The first order conditions for $\vartheta$ under restrictions $r(\varphi(\vartheta))=0$ are as follows:	
	\begin{eqnarray*}
		\hat\vartheta-\vartheta_N&=& - \bar{\check{G}}^{-1}(\tilde\vartheta)[{\check{g}}(\vartheta_N)+\pi \tilde{R}]
	\end{eqnarray*} \vspace{-0.05 in}
	     Expanding the constraint around $\vartheta^\star_0$, $0=r(\varphi(\hat\vartheta))={{r(\varphi(\vartheta^\star_0))}}+R(\tilde{\tilde{\varphi}}(\vartheta))'\left(\frac{d\varphi}{d\vartheta}\right)(\hat\vartheta-\vartheta^\star_0)$   and substituting for $\hat\vartheta-\vartheta_N$ and $\vartheta_{N}-\vartheta^\star_0$ , 	
	$\pi= -(\tilde{R}'\bar{\check{G}}\tilde{R})^{-1}\tilde{R}'(\bar{\check{G}}^{-1}\check{g}-\check{c} N^{-\frac{1}{2}}) 
$. Substituting for $\pi$ in $\hat\vartheta-\vartheta_N$ and plugging $\pi$ in the first order conditions for $\hat\vartheta - \vartheta_N$ the result follows.
\item  $D:=\tilde{R}'\bar{G}^{22}\tilde{R}$. Evaluating the difference between the variances,\small \begin{eqnarray*}
	\mathbb{V}(\mathcal{Z}_r)-\mathbb{V}(\mathcal{Z}) &=& (I-\check{G}\tilde{R}D^{-1}\tilde{R}')\Omega(I-\check{G}\tilde{R}D^{-1}\tilde{R}')'-\Omega	\\
	&=& (\Omega-\check{G}\tilde{R}D^{-1}\tilde{R}'\Omega) - (\Omega-\check{G}\tilde{R}D^{-1}\tilde{R}'\Omega)(\check{G}\tilde{R}D^{-1}\tilde{R}')'-\Omega	\\
	&=& -\check{G}\tilde{R}D^{-1}\tilde{R}'\Omega - \Omega (\check{G}\tilde{R}D^{-1}\tilde{R}')' + \check{G}\tilde{R}D^{-1}\tilde{R}'\Omega \tilde{R}D^{-1}\tilde{R}'\check{G}'	\\
	&=& -\Omega\tilde{R}(\tilde{R}'\Omega\tilde{R})^{-1}\tilde{R}'\Omega - \Omega (\Omega\tilde{R}(\tilde{R}'\Omega\tilde{R})^{-1}\tilde{R}')' + \Omega\tilde{R}(\tilde{R}'\Omega\tilde{R})^{-1}\tilde{R}'\check{G}\tilde{R}(\tilde{R}'\Omega\tilde{R})^{-1}\tilde{R}'\Omega	\\
	&=&  -\Omega\tilde{R}(\tilde{R}'\Omega\tilde{R})^{-1}\tilde{R}'\Omega
\end{eqnarray*} \normalsize where in the fourth  line, I used that $\check{G} =-\Omega$. The result follows by the positive definiteness of $\tilde{R}(\tilde{R}'\Omega\tilde{R})^{-1}\tilde{R}'$.
\end{enumerate}
	\end{proof}
\normalsize  \vspace{-0.2 in}

\begin{proof} \textbf{of Theorem \ref{Shrinkage22}}
		\begin{enumerate} 
			\item 
	Given the asymptotic distribution result in \cref{Shrinkage1_2}, and for  \small $D:=\tilde{R}'\check{G}\tilde{R}$, \\
	$\mathcal{Z}_r\equiv (I - \check{G}(\psi^{\star}_0)\tilde{R}(\varphi^{\star}_0)D^{-1}\tilde{R}(\varphi^{\star}_0)' )\mathcal{Z} + \check{G}(\psi^{\star}_0)\tilde{R}(\varphi^{\star}_0)D^{-1}\tilde{R}(\varphi^{\star}_0)'\check{c}$. \normalsize Hence, the result follows as the bias is equal to \small $\check{G}(\psi^{\star}_0)\tilde{R}(\varphi^{\star}_0)D^{-1}\tilde{R}(\varphi^{\star}_0)'\check{c}\quad$\normalsize  and \small	
	$\mathbb{V}(\mathcal{Z}_r) = (I - \check{G}(\psi^{\star}_0)\tilde{R}(\varphi^{\star}_0)D^{-1}\tilde{R}(\varphi^{\star}_0)' )\Omega (I - \check{G}(\psi^{\star}_0)\tilde{R}(\varphi^{\star}_0)D^{-1}\tilde{R}(\varphi^{\star}_0)' )'$ \normalsize 
	\item The matrix difference between the variances of the tilted  and the true model is:\small \begin{eqnarray*}
	\mathbb{V}(\mathcal{Z}_r)-\mathbb{V}(\mathcal{Z}) &=& (I-\check{G}\tilde{R}D^{-1}\tilde{R}')\Omega(I-\check{G}\tilde{R}D^{-1}\tilde{R}')'-\Omega	\\
	&=& (\Omega-\check{G}\tilde{R}D^{-1}\tilde{R}'\Omega) - (\Omega-\check{G}\tilde{R}D^{-1}\tilde{R}'\Omega)(\check{G}\tilde{R}D^{-1}\tilde{R}')'-\Omega	\\
    &=& -\check{G}\tilde{R}D^{-1}\tilde{R}'\Omega - \Omega (\check{G}\tilde{R}D^{-1}\tilde{R}')' + \check{G}\tilde{R}D^{-1}\tilde{R}'\Omega \tilde{R}D^{-1}\tilde{R}'\check{G}'	\\
	 &=& -\check{G}\tilde{R}D^{-1}\tilde{R}'\Omega - \Omega (\check{G}\tilde{R}D^{-1}\tilde{R}')' + \check{G}\tilde{R}D^{-1}\tilde{R}'\check{G}\tilde{R}D^{-1}\tilde{R}'\check{G}'	\\
	  &=& -\check{G}\tilde{R}D^{-1}\tilde{R}'\Omega - \Omega (\check{G}\tilde{R}D^{-1}\tilde{R}')' + \check{G}\tilde{R}D^{-1}\tilde{R}'\Omega	\\
	    &=&  -\Omega \tilde{R}(\tilde{R}'\Omega\tilde{R})^{-1}\tilde{R}'\Omega	\\
\mathbf{MSE}(\vartheta^{MLE})-\mathbf{MSE}(\hat{\vartheta})&=&	\check{G}\tilde{R}'D^{-1}\tilde{R}\check{c}\check{c}'\tilde{R}'D^{-1}\tilde{R}'\check{G}'- \Omega \tilde{R}(\tilde{R}'\Omega\tilde{R})^{-1}\tilde{R}'\Omega	\\ 
&=&	\Omega[\tilde{R}'(\tilde{R}'\Omega\tilde{R})^{-1}\tilde{R}\check{c}\check{c}'- I]\tilde{R} (\tilde{R}'\Omega\tilde{R})^{-1}\tilde{R}'\Omega
	\end{eqnarray*} \normalsize
The matrix difference is negative definite if $ [\tilde{R}'(\tilde{R}'\Omega\tilde{R})^{-1}\tilde{R}\check{c}\check{c}'- I]$ is negative definite and hence $\xi'(\tilde{R}'(\tilde{R}'\Omega\tilde{R})^{-1}\tilde{R}\check{c}\check{c}'- I)\xi \leq 0$ for any non zero vector $\xi$.
 Hence:
	\begin{itemize}
			\item $\mathbf{MSE}(\hat{\vartheta}) = \Omega + \check{G}\tilde{R}'D^{-1}\tilde{R}\check{c}\check{c}'\tilde{R}'D^{-1}\tilde{R}'\check{G}'- \Omega \tilde{R}(\tilde{R}'\Omega\tilde{R})^{-1}\tilde{R}'\Omega$
			\item Let $\tilde{W}:\tilde{W}\tilde{W}'=W$ and $\tilde{R}'(\tilde{R}'\Omega\tilde{R})^{-1}\tilde{R}\check{c}\check{c}'- I$ be negative definite. Then: \small \begin{eqnarray*}
					\rho(\check{c},\hat{\vartheta}) &=& \mathbb{E}(\mathcal{Z}_{r}'W\mathcal{Z}_{r}) = tr(W\mathbb{E}(\mathcal{Z}_{r}\mathcal{Z}'_{r}))= tr(\tilde{W}\mathbf{MSE}(\hat{\vartheta}) \tilde{W}')\\
					&=& tr(\tilde{W} \Omega\tilde{W}') + tr(\tilde{W} (\check{G}\tilde{R}'D^{-1}\tilde{R}\check{c}\check{c}'\tilde{R}'D^{-1}\tilde{R}'\check{G}')\tilde{W}')-tr(\tilde{W}\Omega \tilde{R}(\tilde{R}'\Omega\tilde{R})^{-1}\tilde{R}'\Omega\tilde{W}')\\
					 &<&tr(\tilde{W} \Omega\tilde{W}') = tr(\tilde{W}\mathbf{MSE}(\vartheta^{MLE}) \tilde{W}')=\rho({\vartheta}^{MLE})
					\end{eqnarray*} \normalsize
	\end{itemize}
\end{enumerate} 
\end{proof}
\begin{proof} \textbf{of \cref{appr}}\\
  a) Similar to \citet{Giacomini2014145}, 
		\begin{eqnarray*}
	\mathbb{E}_{\mathbb{P}_{z}}\log\left(\frac{d\mathbb{P}_{z}}{dH_{z}(\vartheta_{0})}\right)-\mathbb{E}_{\mathbb{P}_{z}}\log\left(\frac{d\mathbb{P}_{z}}{dF_{z}(\vartheta_{0})}\right)&=&\mathbb{E}_{\mathbb{P}_{z}}\log f_{z}(\vartheta_{0})-\mathbb{E}_{\mathbb{P}_{z}}\log h_{z}(\vartheta_{0})=	-\lambda({Z})
\end{eqnarray*}
	By construction, $\lambda(Z)>0 $ as $0 \leq\mathbb{E}_{{h}_{z}(\vartheta)}\log\left(\frac{h_{z}(\vartheta)}{f_{z}(\vartheta)}\right)=\mathbb{E}_{{h}_{z}(\vartheta)}m(\vartheta)+\lambda(Z)=\lambda(Z)$.
\end{proof}

\begin{proof} \textbf{of \cref{Bias}}\\
Consider the sequence of dgp's $\vartheta_{N} = \vartheta_{0} + c_{0}N^{-\frac{1}{2}}$ where $\vartheta_{0}$ lies in the restricted parameter space and is such that \[\bar{l}(\vartheta_{0}):=\left(\frac{d\varphi}{d \vartheta}\right)'\mathbb{E}\mathfrak{s}(x_{i};z_{i},\varphi_{0})=0\] and  $\mathfrak{s}$ is the score function of the restricted model which does not satisfy the moment conditions, $\bar{m}_{f}(\vartheta_{0}):=\mathbb{E}\left(\mathbb{E}_{f}(m(x,;\vartheta_{0},z))\right)\neq 0$. Expanding $\bar{l}(\vartheta_{0})$ around $\vartheta_{N}$, yields that \small
\begin{eqnarray}
	0&=&\left(\frac{d\varphi}{d \vartheta}\right)'l(\varphi_{N}) + \left(\frac{d\varphi}{d \vartheta}\right)'\mathcal{H}(\tilde{\varphi}_{N})\left(\frac{d\varphi}{d \vartheta}\right)(\vartheta_{0}-\vartheta_{N})\nonumber \\
	&=&l(\vartheta_{N}) + \mathcal{H}(\tilde{\vartheta}_{N})(\vartheta_{0}-\vartheta_{N})=l(\vartheta_{N}) - \mathcal{H}(\tilde{\vartheta}_{N})N^{-\frac{1}{2}}c_{0}  \label{thetan}
\end{eqnarray}\normalsize
Correspondingly, $\vartheta_{N}$ is such that  $\bar{m}(\vartheta_{N}):=\mathbb{E}\left(\mathbb{E}(m(x,;\vartheta_{N},z))\right)=0$ and hence $\bar{g}_{1}(\vartheta_{N})=0$. Finally, let $\vartheta_{1}$ be the pseudotrue parameter that solves $\mathbb{E}\left(g_{1}(\vartheta_{1})+ \left(\frac{d\varphi}{d \vartheta}\right)'g_{2}(\psi_{1})\right) =0$, the first order conditions of the tilted model, in population. Since $\bar{g}_{2}(\psi_{1})=\bar{l}(\varphi_{1})-\mathfrak{B}\bar{m}(\vartheta_{1})$ (see page 13 for defn. of $\mathfrak{B}$ ),  \small
\begin{eqnarray*}
\bar{g}_{1}(\vartheta_{1})+ \left(\frac{d\varphi}{d \vartheta}\right)'\left(\bar{l}(\varphi_{1})-\mathfrak{B}\bar{m}(\vartheta_{1})\right)&=&	\left(\bar{M}'(\vartheta_{1})-\left(\frac{d\varphi}{d \vartheta}\right)'\widebar{\mathfrak{s}(\varphi_{1}){m}(\vartheta_{1})}'\right)V^{-1}(\vartheta_{1})\bar{m}(\vartheta_{1}) + \bar{l}(\vartheta_{1})=0
\end{eqnarray*}\normalsize
Expanding $m(\vartheta_{1})$ and $l(\vartheta_{1})$ around $\vartheta_{N}$ yields that\footnote{Note that the midpoint $\tilde{\vartheta}$ is different for different points ($\vartheta_{1}$ and $\vartheta_{0}$) but all differ by $O(N^{-\frac{1}{2}})$.} \small
\begin{eqnarray*}
	\left(\bar{M}'(\vartheta_{1})-\left(\frac{d\varphi}{d \vartheta}\right)'\widebar{\mathfrak{s}(\varphi_{1}){m}(\vartheta_{1})}'\right)V^{-1}\bar{M}(\vartheta_{n})(\vartheta_{1}-\vartheta_{N}) + l(\vartheta_{N}) + \mathcal{H}(\tilde{\vartheta}_{N})(\vartheta_{1}-\vartheta_{N})+O(N^{-\frac{1}{2}})&=&0
		\end{eqnarray*}\normalsize
Using that at $\vartheta_{1}$, $\bar{m}(\vartheta_{1})= O(N^{-\frac{1}{2}})$	and hence $\bar{M}'(\vartheta_{1})=-\left(\frac{d\varphi}{d \vartheta}\right)'\widebar{\mathfrak{s}(\varphi_{1}){m}(\vartheta_{1})}'+ O(N^{-\frac{1}{2}})$, and denoting $2\bar{M}'(\vartheta_{n})V^{-1}\bar{M}(\vartheta_{n})$ by $\hat{G}'(\vartheta_{1})$,\small
	\begin{eqnarray*}
			2\bar{M}'(\vartheta_{n})V^{-1}\bar{M}(\vartheta_{n})(\vartheta_{1}-\vartheta_{N})+ l(\vartheta_{N}) + \mathcal{H}({\tilde{\vartheta}}_{N})(\vartheta_{1}-\vartheta_{N})+O(N^{-\frac{1}{2}})&=&0\\		
		\left(\hat{G}'(\vartheta_{1})+\mathcal{H}(\tilde{\vartheta}_{N})\right)(\vartheta_{1}-\vartheta_{N}) + l(\vartheta_{N})+O(N^{-\frac{1}{2}}) &=&0
	\end{eqnarray*}\normalsize
Using \eqref{thetan} and that there exists a $c_{1}$ such that $\vartheta_{N} = \vartheta_{1} + c_{1}N^{-\frac{1}{2}}$ and hence $0= \vartheta_{1}-\vartheta_{0} + (c_{1}-c_{0})N^{-\frac{1}{2}}$, \small
	\begin{eqnarray*}
	\hat{G}'(\vartheta_{N})(\vartheta_{1}-\vartheta_{N})+				\mathcal{H}(\tilde{\vartheta}_{N})(\vartheta_{1}-\vartheta_{0}+\vartheta_{0}-\vartheta_{N}) + l(\vartheta_{N})+O(N^{-\frac{1}{2}})  &=&0\\
	\hat{G}'(\vartheta_{1})(\vartheta_{1}-\vartheta_{N})+				\mathcal{H}(\tilde{\vartheta}_{N})(\vartheta_{1}-\vartheta_{0}) +O(N^{-\frac{1}{2}})&=&0\\
	N^{-\frac{1}{2}}\hat{G}'(\vartheta_{1})c_{1}+				N^{-\frac{1}{2}}\mathcal{H}(\tilde{\vartheta}_{N})(c_{0}-c_{1})+O(N^{-\frac{1}{2}}) &=&0\\
	(\hat{G}'(\vartheta_{1})-\mathcal{H}(\tilde{\vartheta}_{N}))c_{1}&=& -\mathcal{H}(\tilde{\vartheta}_{N})c_{0} +O(N^{-\frac{1}{2}})\\
	(\hat{G}'(\vartheta_{1})+\check{\mathcal{H}}(\tilde{\vartheta}_{N}))c_{1} &=& \check{\mathcal{H}}(\tilde{\vartheta}_{N})c_{0}	+O(N^{-\frac{1}{2}})
\end{eqnarray*}\normalsize
where $\check{\mathcal{H}}\equiv -\mathcal{H}(\tilde{\vartheta}_{N})$ is positive definite.
Hence, for any localizing vector $c_{0}$, the vector $c_{1}$ that corresponds to the sequence of pseudo-true parameters in the tilted model is as follows:
\[c_{1}=		(\hat{G}'(\vartheta_{1})+\check{\mathcal{H}}(\tilde{\vartheta}_{N}))^{-1}\check{\mathcal{H}}(\tilde{\vartheta}_{N})c_{0} +O(N^{-\frac{1}{2}})\] The relative squared bias is
\small \begin{eqnarray}
		\frac{c_{1}'c_{1}}{c_{0}'c_{0}}&=&		\frac{c'\check{\mathcal{H}}(\tilde{\vartheta}_{N})'(\hat{G}'(\vartheta_{1})+\check{\mathcal{H}}(\tilde{\vartheta}_{N}))^{-1}(\hat{G}(\vartheta_{1})+\check{\mathcal{H}}(\tilde{\vartheta}_{N}))^{-1}\check{\mathcal{H}}(\tilde{\vartheta}_{N})c}{c_{0}'c_{0}} +O(N^{-\frac{1}{2}}) \label{relativebias}
\end{eqnarray} \normalsize
The first term in the RHS is  the Rayleigh quotient, which is bounded above by the  maximum eigenvalue of $\check{\mathcal{H}}(\tilde{\vartheta}_{N})'(\hat{G}'(\vartheta_{1})+\check{\mathcal{H}}(\tilde{\vartheta}_{N}))^{-1}(\hat{G}'(\vartheta_{1})+\check{\mathcal{H}}(\tilde{\vartheta}_{N}))^{-1}\check{\mathcal{H}}(\tilde{\vartheta}_{N})$. What remains to show is that this eigenvalue is less than one. 
Denote by $\bar{\lambda}[A]$ and $\underline{\lambda}[A]$ the smallest and largest eigenvalues of any matrix $A$.  The first term in the RHS of \eqref{relativebias} involves the product of the matrix $\check{\mathcal{H}}(\tilde{\vartheta}_{N})'(\hat{G}'(\vartheta_{1})+\check{\mathcal{H}}(\tilde{\vartheta}_{N}))^{-1}$ and its transpose. Hence, requiring that the maximum eigenvalue is less than one is equivalent to requiring that the maximum squared singular value of $A:=(\check{\mathcal{H}}^{-1}(\tilde{\vartheta}_{N})\hat{G}(\vartheta_{1})+I)^{-1}$, $\sigma^{2}(A)$, is less than one and correspondingly, $\underline{\sigma}^{2}[A^{-1}]>1$ :
\begin{eqnarray}
	\bar{\lambda}[A'A]<1 &\Leftrightarrow& \underline{\lambda}[A^{-1}A^{'-1}]>1 \Leftrightarrow \underline{\sigma}^{2}[A^{-1}]>1 \Leftrightarrow \mid\underline{\sigma}[A^{-1}]\mid>1\\
&\Leftrightarrow& \mid\underline{\sigma}[\check{\mathcal{H}}^{-1}(\tilde{\vartheta}_{N})\hat{G}(\vartheta_{1})+I]\mid>1 
\end{eqnarray}
The above condition is satisfied as \[\underline{\sigma}[\check{\mathcal{H}}^{-1}(\tilde{\vartheta}_{N})\hat{G}(\vartheta_{1})+I]\geq \underline{\sigma}[\check{\mathcal{H}}^{-1}(\tilde{\vartheta}_{N})]\underline{\sigma}[\hat{G}(\vartheta_{1})]+1>1 \]
where the second to last inequality uses Theorem 8.13 in \citet{zhang2011matrix}\footnote{For  $i\in\{1,2..n\}$ where $\bar{\sigma}=\sigma_{1}>\sigma_{2}\geq...\geq \sigma_{n}=\underline{\sigma}$, and $A$ and $B$ are symmetric (Hermitian) and of the same size, $\sigma_{i}(A+B)\geq \sigma_{i}(A)+\sigma_{n}(B)$ and $\sigma_{i}(AB)\geq \sigma_{i}(A)\sigma_{n}(B)$. } and the last inequality uses the fact that both matrices are symmetric and hence have positive singular values.
Hence, for any $\eta>0$   \[\frac{c_{1}'c_{1}}{N}=(\vartheta_{N}-\vartheta_{1})'(\vartheta_{N}-\vartheta_{1})=\frac{c_{0}'c_{0}}{N}-\frac{\eta}{N}+O(N^{-\frac{3}{2}})=(\vartheta_{N}-\vartheta_{0})'(\vartheta_{N}-\vartheta_{0})-\frac{\eta}{N}+O\left(\frac{1}{N^{\frac{3}{2}}}\right)\]\normalsize
\end{proof}

\begin{Lemma} \underline{Limits of derivatives of $(\mu,\lambda)$ with respect to $(\vartheta,\varphi)$}: \\
	Under correct specification, the unconditional moments of all derivatives are as follows: \small
	\begin{itemize}
		\item  \underline{First order derivatives and relevant Second order derivatives}  
		\begin{eqnarray*}
			\mathbb{E}_{\mathbb{P}_N}\mu_{i,\vartheta} & \overset{}{\to}& -\mathbb{V}^{-1}_{m}M_{P}\\
			\mathbb{E}_{\mathbb{P}_N}\lambda_{i,\vartheta} & \overset{}{\to}&0 \\
			\mathbb{E}_{\mathbb{P}_N}\mu_{i,\varphi}&\overset{}{\to}&-V^{-1}_{m}\mathbb{E}(m\mathfrak{s}')\\
			\mathbb{E}_{\mathbb{P}_N}\lambda_{i,\varphi}&\overset{}{\to}&0\\
			\mathbb{E}_{\mathbb{P}_N}\lambda_{i,\varphi_l\varphi'}&{\to}& \mathbb{E}(\mathfrak{s}^{l}_{j}m_{j}'V^{-1}_{m}m\mathfrak{s}')\\
			\mathbb{E}_{\mathbb{P}_N}\lambda_{i,\vartheta_l\vartheta'}&{\to}& M'V^{-1}_{m}M^{l}\\
			\mathbb{E}_{\mathbb{P}_N}\lambda_{i,\varphi_l\vartheta'}&\to&-\mathbb{E}(\mathfrak{s}^{l}m')V^{-1}_{m}M
		\end{eqnarray*}
	\end{itemize}\label{muder}
\end{Lemma}\normalsize \vspace{-0.2 in}
\begin{proof} of \cref{muder}. \\
	Defining the following quantities :$e_{j,i}=e^{\mu_{i}'m_{j,i}(\vartheta)}$, $\tilde{e}_{j,i}=\frac{e_{j,i}}{\frac{1}{N_{s}}\sum_{j=1..s}e_{j,i}}$
	, $\varkappa_{j.i}=-\frac{(e^{\mu_{i}'m_{j,i}(\vartheta)}-1)}{\mu_{i}m_{j,i}(\vartheta)'}$, $s_{j,i}:=\frac{\partial}{\partial\varphi}\log f(x_{j}|\varphi,z_{i})$ and $\mathfrak{s}_{j,i}:=\frac{s_{j,i}}{f_{j,i}}$, the first and second derivatives of $(\mu,\lambda)$ with respect to $\psi$ are: $  $	
	\small
	\begin{eqnarray*}
		\underset{n_{m}\times n_{\vartheta}}{\mu_{i,\vartheta}} & = & -\left(\frac{1}{N_{s}}\sum_{j}e_{j,i}m_{j,i}m_{j,i}'\right)^{-1}\left(\frac{1}{N_{s}}\sum_{j}M_{j,i}-\frac{1}{N_{s}}\sum_je_{j,i}m_{j,i}\mu_{i}'M_{j,i}\right)\\
		\underset{1\times n_{\vartheta}}{\lambda_{i,\vartheta}} & = & -\mu_{i}'\frac{1}{N_{s}}\sum_{j}\tilde{e}_{j,i}M_{j,i} \\
		\underset{n_{m}\times n_{\varphi}}{\mu_{i,\varphi}}&=&-\left(\sum_je_{j,i}m_{j,i}m_{j,i}'\right)^{-1}\sum_je_{j,i}m_{j,i}\mathfrak{s}'_{j,i}\\
		\underset{1\times n_{\varphi}}{\lambda_{i,\varphi}}&=& -\frac{1}{N_{s}}\sum_{j}\tilde{e}_{j,i}\mathfrak{s}'_{j,i}\\
		\lambda_{i,\vartheta_l\vartheta'}& = & -\mu_{i,\vartheta}'\frac{1}{N_s}\sum_{j}\tilde{e}_{j}M^{l}_{j,i}-\frac{1}{N_{s}}\sum_{j}\tilde{e}_{j,i,\vartheta}M^{l'}_{j,i}\mu_{i}-\frac{1}{N_{s}}\sum_{j}\tilde{e}_{j,i}\frac{\partial M^{l'}_{j,i}}{\partial \vartheta'}\mu_{i}\\
		\lambda_{i,\varphi_l\varphi'}&=&-\frac{1}{N_{s}}\sum_{j}\tilde{e}_{j}\mathfrak{s}^{l}_{j}m_{j}'\mu_{i,\varphi}\\
		\lambda_{i,\varphi_l\vartheta'}&=&-\frac{1}{N_{s}}\sum_{j}\tilde{e}_{j}(\mu_{\vartheta}'m_{j,i}+M_{j,i}'\mu_{i}-\sum_{j}\tilde{e}_{j}M'_{j,i}\mu_{i})\mathfrak{s}^{l'}_j
	\end{eqnarray*}	\normalsize
	where $\tilde{e}_{j,i,\vartheta}=\tilde{e}_{j}(\mu_{\vartheta}'m_{j,i}+M_{j,i}'\mu_{i}-\frac{1}{N_{s}}\sum_{j}\tilde{e}_{j}M'_{j,i}\mu_{i})$.	We have already established that as long as the base density is asymptotically correctly specified, then $\mu_i \underset{p}{\to} 0$ for almost all $z_i$. Therefore, $e_{j,i}\underset{p}{\to}1$,  and $\varkappa_{j,i}\underset{p}{\to}-1$. Again, using \textbf{BD-1a}, WLLN for uniformly integrable sequences applies and the simulated averages converge to their population values as $N_s\to \infty$. Taking the unconditional expectation   using the empirical distribution $\mathbb{P}_{N}$, we conclude by applying the Portmanteau Lemma as $N\to\infty$.
\end{proof}

\newpage
\section{Appendix B (Online)}  \label{AppB}

This Appendix \hyperref[AppB]{B}  contains: Analysis for counterfactual effects, auxiliary results (Lemma \ref{wass} to \ref{lemvar}), analytical details for the asset pricing example and the application, base density estimates for the latter, and additional simulation evidence. 

\subsection{Counterfactual Distributions}
An additional advantage of the method used in this paper is the ability to perform counterfactual experiments. What is more important is that
this method readily gives a counterfactual distribution, while the
distribution of the endogenous variables is hardly known in non-linear
equilibrium models. Knowing the distribution of outcomes is extremely important
for policy analysis, especially when non linear effects take place,
and therefore the average effect is not a sufficient statistic to
make a decision. 

At $(\varphi^{\star}_{0},\vartheta^{\star}_{0})$, $H(X_{t}|Z_{t},\psi^{\star}_{0})$ is the best approximation to the conditional distribution that satisfies the moment conditions. The meaning of a counterfactual exercise is to examine what happens in the case in which there is a change in some structural feature, i.e. a change in $\psi$ and how that affects the distribution of outcomes. By necessity, a change in $\psi$ implies moving away from the best approximating model to the underlying data and hence the term "counterfactual"\footnote{While a change in $\vartheta$ typically has a straightforward interpretation, a change in $\varphi$ only has a structural interpretation in the case of an exogenous process.}.
\begin{Proposition}{Counterfactual Effects\\}
	For any measurable function $\zeta(X_{t})$, the average effect of a change in $\psi$ is equal to
	\begin{eqnarray}
		\frac{\partial\mathbb{E}_{H(X_{t}|Z_{t},\psi)}\zeta(X_{t})}{\partial\psi_{l}} =\left[ \begin{array}{c}\mu_{\varphi}'\mathbb{C}ov_{H}(m_{t},\zeta_{t})+\mathbb{C}ov_{H}(\mathfrak{s}_{t},\zeta_{t})\\
			\mu_{\vartheta}'\mathbb{C}ov_{H}(m_{t},\zeta_{t})-\mathbb{C}ov_{H}(M_{t}',\zeta_{t})\mu_{t}
		\end{array}\right]
	\end{eqnarray} where $m_{t}$ and $\mathfrak{s}_{t}$ abbreviate the moment and density score functions respectively.
	\normalsize
	\label{counter_def}
\end{Proposition}    \vspace{-0.15 in}
\begin{proof} \textbf{of \cref{counter_def}}\\
	The average effect computed under the $H$ measure is $\int \zeta(X_{t}) h(X_{t}|Z_{t},\psi) dX_{t}$. Taking derivatives with respect to $\vartheta$ and $\varphi$ and suppressing dependence on variables: \small
	\[\int \zeta_{t}\frac{\partial h_{t}}{\partial \varphi'} dX_{t}=\int \zeta_{t} \left(\frac{\partial f_{t}}{\partial \varphi'}(e^{\mu'm_{t}+\lambda})+h_{t}(m'\mu_{\varphi}+\lambda_{\varphi})\right)dX_{t}=\mathbb{C}ov(\zeta_{t},\mathfrak{s}'_{t})+\mathbb{C}ov(\zeta_{t},m'_{t})\mu_{\varphi}\]
	\[\int \zeta_{t}\frac{\partial h_{t}}{\partial \vartheta'} dX_{t}=\int \zeta_{t}h_{t}\left(m'\mu_{\vartheta}+\mu'M_{t}+\lambda_{\vartheta}\right)dX_{t}=\mathbb{C}ov(\zeta_{t},m_{t}')\mu_{\vartheta}-\mu'\mathbb{C}ov(\zeta_{t},M_{t})\] \normalsize
\end{proof}
Focusing on changes in an element of $\vartheta$, the average effect is equal to \small
\begin{eqnarray*}
	\frac{\partial\mathbb{E}_{H}(\zeta|Z)}{\partial\vartheta}&=&-\mathbb{E}_{H}({M}|Z)'V^{-1}_{H,m}\mathbb{C}ov_{H}(m,\zeta|Z)\\
	&&+\mathbb{C}ov_{H}(m,\mu'{M}|Z)'V^{-1}_{m}\mathbb{C}ov_{H}(m,\zeta|Z)-\mathbb{C}ov_{H}({M},\zeta_{t}|Z)'\mu\\
	&& = -\mathbb{E}_{H}({M}|Z)'V^{-1}_{m}\mathbb{C}ov_{H}(m,\zeta|Z)-\mathbb{E}_{H}(\mathfrak{P}_{m}^{\bot}{M}\zeta|Z)'\mu\end{eqnarray*}\normalsize
where $\mathfrak{P}_{m}^{\bot}$ is the orthogonal complement of projecting on the moment functions.
The total effect is thus the sum of the average effect on $\zeta$ through changes in $m$ and the direct effect through $\mathfrak{P}_{m}^{\bot}{M}$, which is the variation in the Jacobian that is unrelated to $m$, weighted by $\mu$ . When $\mu\approxeq 0$, the average effect collapses to $-\mathbb{E}({M}|Z)'V^{-1}_{m}\mathbb{C}ov(m,\zeta|Z)$, that is all structural information is contained in $m$.

Beyond average effects, one might be interested in the distribution itself. 
We present below the case for the asset pricing example in Section \ref{ex}, where the utility function
is now of the Constant Relative Risk Aversion form. 
The counterfactual
experiment consists of increasing the CRRA coefficient. 

\begin{figure}[H]
	\includegraphics[width=4.3 in, height = 3.1 in]{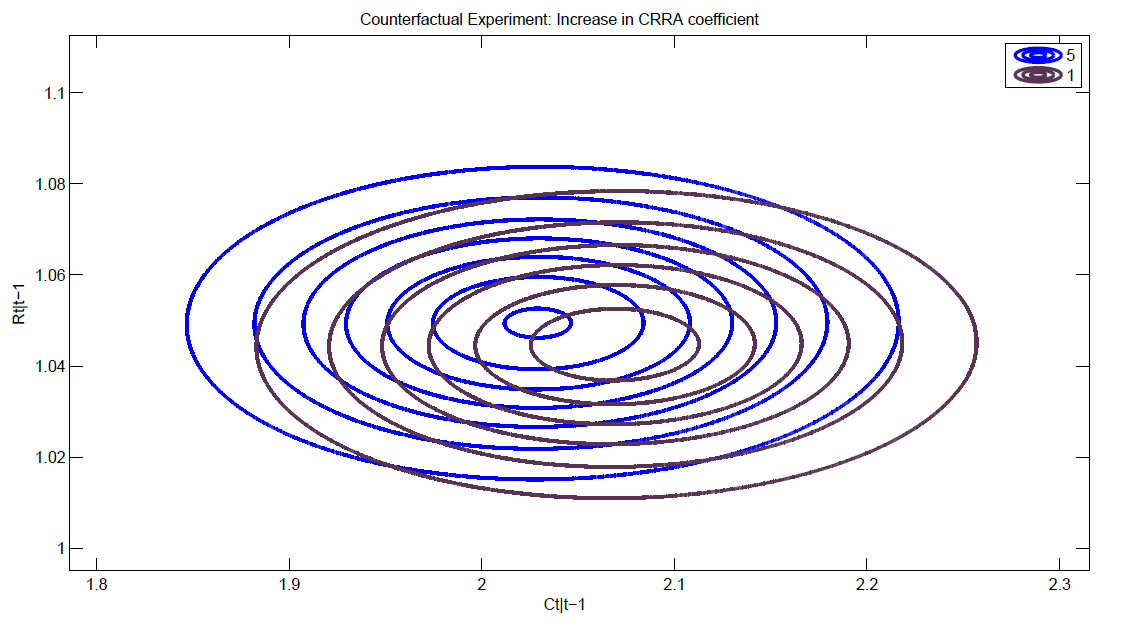}\vspace{-0.1 in}\protect\caption{Increase in Risk Aversion Coefficient from 1 to 5} \label{riskaversion}
\end{figure}
Figure \ref{riskaversion} plots
the contour maps of the conditional joint density of $(R_{t+1},C_{t+1})$
with a change in the risk aversion coefficient. An increase in risk
aversion is consistent with higher mean interest rate, and lower mean
consumption. Moreover, consumption and interest rates are less negatively
correlated. This is also consistent what the log - linearized Euler
equation implies, $c_{t}=-\frac{1}{\sigma}r_{t}$. 
\begin{Lemma} Influence function for plug-in estimator \citep{Wasserman_2006}\\
	For a general function $W(x,z)$, conditional density $Q(x|z)$ and $\mathcal{L}(x,z)\equiv W(x,z)-\int W(x,z)d\mathbb{P}(x|z)$ 
	\begin{eqnarray*}
		W_{Q_N}-W_{P}& \equiv & \int W(x,z)d(Q(x|z)d\mathbb{P}(z))-\int W(x,z)d(\mathbb{P}(x|z)d\mathbb{P}(z))\\\
		&=&\int \int \mathcal{L}(x,z)dQ(x|z)d\mathbb{P}(z)
	\end{eqnarray*}  \label{wass}
\end{Lemma} \vspace{-0.4 in}
\begin{Corollary}Parametric Density.\\ For any $(x,z)$ - measurable function $W(.)$ and $P\equiv \mathbb \mathbb{P}(\varphi)$, $\mathbb{P}(\varphi)$ 1-differentiable in $\varphi$, with score function $\mathfrak{s}(x,z)$, the following statement holds:
	 \begin{eqnarray*} 
		W_{P(\varphi_{0}+cN^{-\frac{1}{2}})}-W_{P}&=& N^{-\frac{1}{2}}c \mathbb{E}\left(Cov_{P}(W,\mathfrak{s}|z)\right)  + o(1)
	\end{eqnarray*}
\end{Corollary}
\begin{proof}
	In the parametric case within the class of smooth densities, we can rewrite $dQ(x|z)\equiv dP(x|\varphi+N^{-\frac{1}{2}}c,z)$. Therefore, using a Taylor expansion of around $\varphi_0$ and denoting by $s_\varphi(x,z)$ the derivative of the density of $P$,
	\begin{eqnarray*}
		{dP}(x|\varphi+N^{-\frac{1}{2}}c,z) &=& dP(x|\varphi,z)+  (N^{-\frac{1}{2}}c )s_\varphi(x,z)dx+o(N^{-\frac{1}{2}}c)
	\end{eqnarray*}	
	Evaluating $\int \int \mathcal{L}(x,z)dQ(x|z)\mathbb{P}(z)$ in \cref{wass} gives the result:
	\begin{eqnarray*}
		W_{Q_N}-W_{P}& \equiv & \int \int \mathcal{L}(x,z)(s_\varphi(x,z)N^{-\frac{1}{2}}c +o(N^{-\frac{1}{2}}c))d\mathbb{P}(z)  \\
		&=& N^{-\frac{1}{2}}c\int \int  \mathcal{L}(x,z)(s_\varphi(x,z)dx d\mathbb{P}(z)  + o(1)\\
		&=& N^{-\frac{1}{2}}c\int \int  \mathcal{L}(x,z)(\mathfrak{s}(x,z)dP(x|\varphi,z) d\mathbb{P}(z)  + o(1)\\
		& = & N^{-\frac{1}{2}}c \mathbb{E}\left(Cov_{P}(W,\mathfrak{s}|z)\right)  + o(1)
	\end{eqnarray*}
\end{proof} \vspace{-0.2 in}
\begin{Corollary}  For any $z$ - measurable and integrable function $W(.)$ that satisfies $\textbf{BD-1a}$ and density $q_{N}(x|z)$ which is absolutely continuous with respect to $p(x|z)$ and $\Delta(p{(x|z)},q_{N}(x|z))=O_{p_{z}}(\kappa_{N}^{-1})$ the following statement holds:
	\begin{eqnarray*} 
			W_{Q_N}-W_{P}&=& \int O_{p_{z}}(\kappa_{N}^{-\frac{1}{2}})  d\mathbb{P}(z) 
	\end{eqnarray*}
	\vspace{-0.2 in}\begin{proof} \vspace{-0.2 in}
	\begin{eqnarray*}
		&&W_{Q_N}-W_{P}\\
		& \equiv & \int \int \mathcal{L}(x,z)q_{N}(x|z)dx d\mathbb{P}(z) = \int \int \mathcal{L}(x,z)\left(\frac{q_{N}(x|z)}{p(x|z)}\right)dP(x|z) d\mathbb{P}(z) \\
				&= &  \mathbb{E}Cov_{P}\left(W(x,z),\frac{q_{N}(x|z)}{p(x|z)}|z\right) = \mathbb{E}Cov_{P}\left(W(x,z),\frac{q_{N}(x|z)-p(x|z)}{p(x|z)}|z\right)  \\
		&\leq& \int  \mathbb{V}_{P}(W|z)^{\frac{1}{2}}\left(\int \left(\frac{q_{N}(x|z)-p(x|z)}{p(x|z)}\right)^{2}p(x|z)dx\right)^{\frac{1}{2}}d\mathbb{P}(z) \\
		&=& \int  \mathbb{V}_{P}(W|z)^{\frac{1}{2}}\chi^{2}(q_{N}(x|z),p{(x|z)})^{\frac{1}{2}}d\mathbb{P}(z) \\
						&=& \int O_{p_{z}}(\kappa_{N}^{-\frac{1}{2}}) d\mathbb{P}(z)
	\end{eqnarray*} \normalsize The last line uses that the first term is bounded by assumption $\textbf{PD-1}$ and that the $\chi^{2}$ distance is bounded above by $\Delta$ defined in \cref{def2}.

\end{proof}
\label{disc}
\end{Corollary}\vspace{-0.2 in}
\begin{Lemma}
	For some invertible matrix  $C_{f}=\frac{1}{N_{s}}\sum_{j}^{N_{s}}C_{j}$, denote $\bar{C}_{f}^{-1}:=\mathbb{E}_{F_{(z)}}C_{f}$
	\begin{enumerate}
		\item  $C_{f}^{-1} = \bar{C}_{f}^{-1}+O_{p_{z}}(N_{s}^{-\frac{1}{2}})$. 
		\item  More generally, for some integrable density $g:\Delta(g,f)=O_{p_{z}}(\kappa_{N}^{-1})$, \\$\bar{C}_{g}^{-1}  = \bar{C}_{f}^{-1}+O_{p_{z}}(\kappa_{N}^{-\frac{1}{2}})$.
	\end{enumerate}
\label{lemvar}
\end{Lemma}\vspace{-0.1 in}
	\begin{proof} of \cref{lemvar}
		\begin{enumerate}
		\item 	$C_{f}^{-1} =\bar{C}_{f}^{-1}-\bar{C}_{f}^{-1}(C_{f}-\bar{C}_{f})\bar{C}_{f}^{-1}+O_{p_{z}}(N_{s}^{-1})= \bar{C}_{f}^{-1}+O_{p_{z}}(N_{s}^{-\frac{1}{2}})$
		\item 	$\bar{C}_{g}^{-1} =\bar{C}_{f}^{-1}-\bar{C}_{f}^{-1}(\bar{C}_{g}-\bar{C}_{f})\bar{C}_{f}^{-1}+O_{p_{z}}(||g-f||)$. 
		Therefore, 		
		\begin{eqnarray*}
			\bar{C}_{g}^{-1} & = & \bar{C}_{f}^{-1}-\bar{C}_{f}^{-1}(\bar{C}_{g}-\bar{C}_{f})\bar{C}_{f}^{-1}+O_{p_{z}}(||g-f||)\\
			& = & \bar{C}_{f}^{-1}-\bar{C}_{f}^{-1}(\bar{C}_{g}-\bar{C}_{f})\bar{C}_{f}^{-1}+O_{p_{z}}(||g-f||)\\
			& = & \bar{C}_{f}^{-1}+O_{p_{z}}(\kappa_{N}^{-\frac{1}{2}})
		\end{eqnarray*}	 where the last equality applies \cref{disc} elementwise.	
		\end{enumerate}
		\end{proof}

	\subsection{Additional Monte Carlo Results}
	\paragraph{\textit{Unconditional Case}}
	To construct the unconditional moments, I use lagged consumption growth in the case of one instrument, while I add second and third order powers of lagged consumption and interest rate in the case of six instruments. 
	
	In Figure \ref{MCunc} I compare the performance of the CU-GMM and ETEL estimators to our estimator for $\hat\beta$, both in the unrestricted and in the restricted case. 
	Both estimators perform worse than our estimator in all cases.\footnote{Although it does not belong to the GEL class, the ETEL estimator is expected to have similar performance to empirical likelihood as shown by \cite{SusanneSchennah} and thus the CU-GMM estimator which is also member of the GEL family.} Utilizing five more instruments in the case of CU-GMM does not deliver a lower MSE to our estimator, which uses only one instrument in both cases.  
		\begin{figure}[H]
		\includegraphics[scale=0.62]{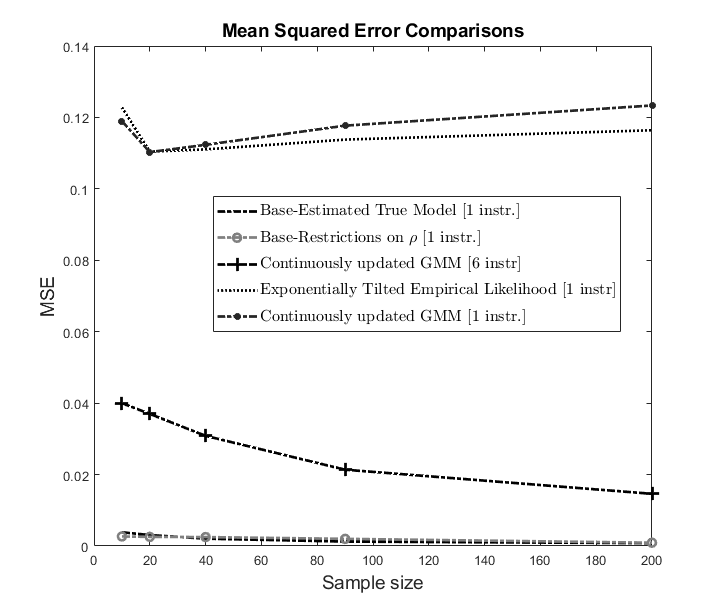}
		\vspace{-0.1 in} \label{MCunc}
		\protect\caption{$\hat\beta$ with different estimators (500 MC replications)}
	\end{figure} 
	In Figure \ref{MCunc2} I present the same results but I focus on the relative difference between incorrect restrictions $(\rho_{CR}=\rho_{R})$ versus using the unrestricted density. In this case the bias - variance trade-off shows up at a relatively small sample size.

	\begin{figure}[H]
		\includegraphics[scale=0.62]{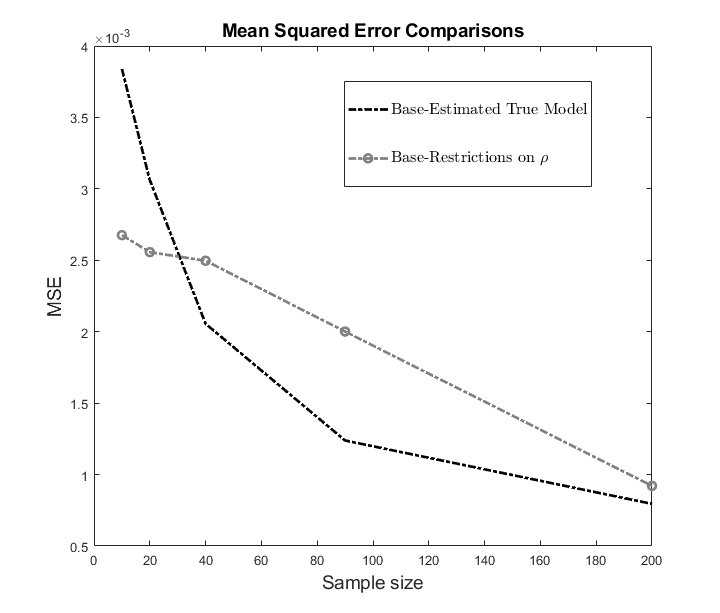}
		\vspace{-0.1 in}  \label{MCunc2}
		\protect\caption{$\hat\beta$ for restricted $\Phi$ vs True model (500 MC replications)}
	\end{figure}

\subsection{Additional Analytical and Empirical Results for Application}

\subsubsection{Basic Elements of Long Run Risks model with no time preference risk}

The representative agent maximizes a constant elasticity of substitution recursion 
\[U_{t}=\left((1-\beta)C^{1-\rho}_{t}+\beta R_{t}(U_{t+1})^{1-\rho}\right)^\frac{1}{1-\rho}\]
where $R_{t}(U_{t+1})$ is the conditional certainty equivalent, $\mathbb{E}_{t}\left(U^{1-\gamma}_{t+1}\right)^{\frac{1}{1-\gamma}}$, and $EIS=\psi=\frac{1}{\rho}$. The inter-temporal marginal rate of substitution is therefore equal to \[M_{t+1}:=\beta\left(\frac{C_{t+1}}{C_{t}}\right)^{-\rho}\left(\frac{U_{t+1}}{R_{t}(U_{t+1})}\right)^{\rho-\gamma}\]

Using Euler's theorem, the value function $U_{t}$, which is homogeneous of degree one, can be written as $U_{t}=\frac{\partial U_{t}}{\partial C_{t}}C_{t}+\mathbb{E}_{t}\frac{\partial U_{t}}{\partial U_{t+1}}U_{t+1}$. Moreover, define wealth as the value of the aggregate consumption stream in equilibrium, $W_{t}=\frac{U_{t}}{\frac{\partial U_{t}}{\partial C_{t}}}$. 

Then, using the definition of $U_{t}$, the wealth to consumption ratio is equal to $\frac{W_{t}}{C_{t}}=  \frac{1}{1-\beta}\left(\frac{U_{t}}{C_{t}}\right)^{1-\rho}=\frac{1}{1-\beta}\left(\frac{U_{t}}{C_{t}}\right)^{\frac{\psi-1}{\psi}}$ and $W_{t}=C_{t}+\mathbb{E}_{t}M_{t+1}W_{t+1}$. Finally, using the definition of $R_{a,t+1}=\frac{W_{t+1}}{W_{t}-C_{t}}$ and the relationship between the wealth to consumption ratio and the continuation value, we have that $R_{a,t+1}=\left(\beta\left(\frac{C_{t+1}}{C_{t}}\right)^{-\rho}\left(\frac{U_{t+1}}{R_{t}(U_{t+1})}\right)^{\rho-1}\right)^{-1}$. Substituting for $\left(\frac{U_{t+1}}{R_{t}(U_{t+1})}\right)$ in $M_{t+1}$ leads to the relation between $R_{a,t+1}$ and $M_{t+1}$  in \eqref{lrr}: \[M_{t+1}:=\beta^{\theta}\left(\frac{C_{t+1}}{C_{t}}\right)^{-\frac{\theta}{\psi}}R^{\theta-1}_{a,t+1}\] 

\subsubsection{}
\textit{Empirical Results}
\begin{table}[H]
	\begin{centering}
		\caption{Robust Confidence Set for Reduced form Parameters}\bigskip
		\begin{tabular}{|c||c||c|c|}
			\hline
			Parameter & $q_{2.5\%}$ & $\text{mode}$  & $ q_{95\%}$ \\ \hline 
			$\rho$    &  0.8678 & 0.9757 & 0.9934                    \\  \hline
			$\chi_{x}$& 0.1321 & 0.1539  & 0.2101\\  \hline
			$\sigma$  &	0.0001 & 0.0008 & 0.0024\\ \hline 		    
			$\rho_{l}$&- &0.9560& - \\ \hline 			
			$\sigma_{l}$&- & 0.0004&- \\  \hline 			
			$\rho_{v,x}$	& 0.8747 & 0.9078& 0.9555 \\ \hline 			
			$\sigma_{v,x}$	& 0.0011 & 0.0078 & 0.0099 \\   \hline
		\end{tabular}
		\par\end{centering}
\end{table}

\end{document}